\documentclass[10pt,reqno]{amsart}

\usepackage{amsaddr}
\usepackage{etoolbox}
\patchcmd{\section}{\scshape}{\bfseries\scshape}{}{}
\makeatletter
\renewcommand{\@secnumfont}{\bfseries}
\makeatother

\input{our_macros.sty}

\title{Correlation detection in trees for planted graph alignment}

\author{Luca Ganassali$^\mathrm{\lowercase{a} }$, \, Laurent Massouli\'e$^\mathrm{\lowercase{a,b} }$, \, Marc Lelarge$^\mathrm{\lowercase{a}}$}
\address{$^\mathrm{\lowercase{a}}$Inria, DI/ENS, PSL Research University, Paris, France, \, $^\mathrm{\lowercase{b}}$MSR-Inria Joint Centre}

\date{}                  

\begin{document}
	
	\begin{abstract} 
		Motivated by alignment of correlated sparse random graphs, we introduce a hypothesis testing problem of deciding whether or not two random trees are correlated. We obtain sufficient conditions under which this testing is impossible or feasible.
		
		We  propose \texttt{MPAlign}, a message-passing algorithm for graph alignment inspired by the tree correlation detection problem. We prove \texttt{MPAlign} to succeed in polynomial time at partial alignment whenever tree detection is feasible. As a result our analysis of tree detection reveals new ranges of parameters for which partial alignment of sparse random graphs is feasible in polynomial time.
		
		We then conjecture that graph alignment is not feasible in polynomial time when the associated tree detection problem is impossible. If true, this conjecture together with our sufficient conditions on tree detection impossibility would imply the existence of a hard phase for graph alignment, i.e. a parameter range where alignment cannot be done in polynomial time even though it is known to be feasible in non-polynomial time.
		
		\emph{A short version of this work has been presented at the ITCS'22 conference \cite{DBLP:conf/innovations/GanassaliML22}.}
		
	\end{abstract}

\maketitle

\section*{Introduction}

\textbf{Graph alignment} 
Given two graphs $G=(V,E)$ and $H=(V',E')$ with $\card{V}=\card{V'}$, the problem of \emph{graph alignment} consists of identifying a bijective mapping, or \emph{alignment} $\pi: V \to V'$ that minimizes 
\begin{equation*}
    \sum_{u,v \in V} \left( \one_{\set{u,v}\in E} - \one_{\{\pi(u),\pi(v)\} \in E'} \right)^2,
\end{equation*} that is the number of disagreements between adjacencies in the two graphs under the alignment $\pi$. 

This problem reduces to the graph isomorphism problem in the noiseless setting where the two graphs can be matched perfectly, i.e. are isomorphic. The paradigm of graph alignment has found numerous applications across a variety of diverse fields, such as network privacy \cite{Narayanan08}, computational biology \cite{Singh08}, computer vision \cite{CFVS04}, and natural language processing.

Given the adjacency matrices $A$ and $B$ of the two graphs, the graph matching problem can be viewed as an instance of the quadratic assignment problem (QAP) \cite{Pardalos94}: 
\begin{equation}\label{eq:QAP}
    \argmax_\Pi \langle A, \Pi B \Pi^T\rangle
\end{equation}
where $\Pi$ ranges over all $n\times n$ permutation matrices, and $\langle \cdot, \cdot \rangle$ denotes the matrix inner product. QAP is known to be NP-hard in general, as well as some of its approximations \cite{Pardalos94,Makarychev14}. These hardness results are applicable in the worst case, where the observed graphs are designed by an adversary.  In many applications, the graphs can be modeled by random graphs; accordingly our focus will be the \emph{planted} version of the problem, which concerns random graph instances, and has a different objective. \\

\textbf{Correlated \ER model} 
A recent thread of research \cite{Cullina2017,Cullina18,Ding18,Dai19,fan2019ERC,Ganassali20a,ganassali2021impossibility} has focused on the study of graph alignment when the two considered graphs are drawn from a generative model under which they are both Erd\H{o}s-Rényi random graphs \cite{Erdos59}. 
Specifically, for $(\lambda,s)\in \dR_+\times [0,1]$, the correlated Erd\H{o}s-Rényi random graph model, denoted $\G(n,q,s)$ with $q=\lambda/n$, consists of two random graphs $G,G'$ both with node set $[n]$ generated as follows. Consider an i.i.d. collection $\{ ( A_{u,v},A'_{u,v}) \}_{u<v \in [n]}$ of pairs of correlated Bernoulli random variables with distribution 
\begin{equation}\label{eq:CER_model}
  (A_{u,v},A'_{u,v}) =
    \begin{cases}
      (1,1) & \text{with probability $\lambda s/n$}\\
      (1,0) & \text{with probability $\lambda (1-s)/n$}\\
      (0,1) & \text{with probability $\lambda (1-s)/n$}\\
      (0,0) & \text{with probability $1-\lambda (2-s)/n$}.
    \end{cases}       
\end{equation} 

An important graph mentioned in the sequel is the \emph{aligned intersection graph}, made of nodes from $[n]$ and edges $\set{u,v}$ such that $A_{u,v}=A'_{u,v}=1$. We next perform a \emph{relabeling step}: consider then a permutation $\pi^*$ drawn independently of $A,A'$ and uniformly at random from the symmetric group $\cS_n$. The two graphs $(G,H)$ are then defined by their adjacency matrices $A$ and $B$ such that for all $i<j \in [n]$:
\begin{equation*}
    A_{u,v}=A_{v,u},\; B_{u,v}=B_{v,u}=A'_{\pi^*(u)\pi^*(v)}.
\end{equation*}
In this setting, the marginal distributions of $G$ and $H$ are identical, namely that of the \ER model $\G(n,q)$ with $q=\lambda/n$. The above model is entirely described by three parameters: the number of nodes $n$; the marginal mean degree $\lambda$; and the correlation $s$.

\begin{remark}[A colored view]\label{rem:colored_view}
    A visual representation of the pair $(G,G')$ -- before the relabeling step -- is as follows. The two aligned graphs are superimposed according to the node set $[n]$, and each edge $\set{u,v}$ is colored in blue (resp. in red) if it appears in $G$ (resp. in $G'$). The two-colored edges appearing in both graphs are drawn thick and purple -- see Fig. \ref{fig:unplanted_pb}.A hereafter. Note what with the previous definition, the intersection graph is exactly that made of the two-colored edges.
\end{remark}

\begin{figure}[h]
     \centering
     \begin{subfigure}[b]{\textwidth}
         \centering
         \includegraphics[scale=0.8]{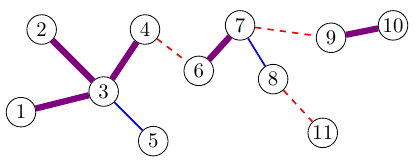}
         \caption{Graphs $G,G'$ in colored view -- see Remark \ref{rem:colored_view}}
         \label{fig:G_union}
     \end{subfigure}
     \hfill
     \begin{subfigure}[b]{\textwidth}
         \vspace{0.3cm}
         \centering
         \includegraphics[scale=0.8]{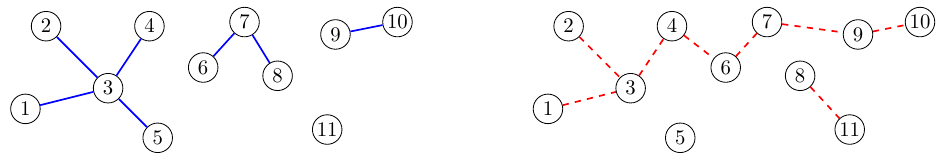}
         \caption{Graphs $G, G'$ (in separate views)}
         \label{fig:GGtilde}
     \end{subfigure}
     \hfill
     \begin{subfigure}[b]{\textwidth}
         \centering
         \vspace{0.1cm}
         \includegraphics[scale=0.8]{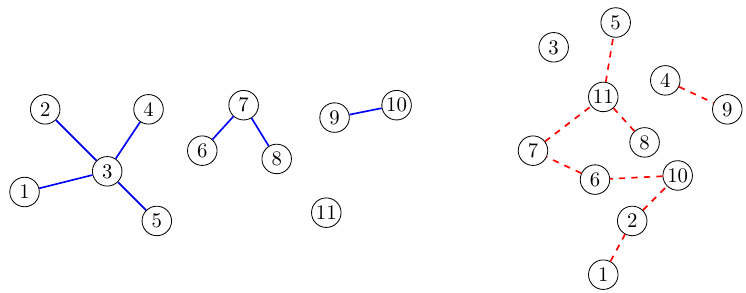}
         \caption{Graphs $G,H$ (in separate views)}
         \label{fig:GH}
     \end{subfigure}
     
    \caption{A sample from model $\G(n,q=\lambda/n,s)$ with $n=11$, $\lambda=1.9$, $s=0.7$.}
    \label{fig:unplanted_pb}
\end{figure}

\textbf{Planted graph alignment} 
Recall that we focus on planted graph alignment in the previously described model; the question now consists in finding an estimator $\hat{\pi}$ of the planted solution $\pi^*$, upon observing $G$ and $H$. 
At first sight, when $n$ gets large, an optimal requirement for $\hat{\pi}$ would be to 
agree with $\pi^*$ on every node, that is achieving \emph{exact recovery}, or on a fraction $1-o(1)$ of them -- which is known as \emph{almost exact recovery}. However, in this so-called \emph{sparse regime}  where the graphs have constant mean degree $\lambda$, it is established \cite{Cullina2017,Cullina18} that the presence of $\Omega(n)$ isolated vertices in the underlying aligned intersection graph of $G$ and $G'$ makes these exact and almost exact recovery out of the reach of any estimator. 

As a result, we will consider an alternative objective, instead of QAP: recovering the ground truth $\pi^*$, as stated before. Note that the solution to QAP \eqref{eq:QAP} coincides with maximizing the posterior distribution of $\pi^*$ given $G,H$ -- hence taking the \emph{maximum a posteriori estimator} of $\pi^*$. Nonetheless, the computation of this estimator is still NP-hard, and hence shall be put aside in the rest of the study.

In our setting, estimators $\hat{\pi}$ may only consist in partial matchings, hence not necessarily permutations in $\cS_n$. Let us now describe our measure of performance for such estimators. Consider any estimator $\hat{\pi}: \cC \to [n]$, where $\cC$ is a subset\footnote{hence set $\cC$ is considered to be part of the output of any method yielding an estimator.} of the full node set $[n]$. The performance of $\hat{\pi}$ is assessed through $\ov(\pi^*,\hat{\pi})$, its {\em overlap} with the unknown permutation $\pi^*$, defined as
\begin{equation}\label{eq:def_overlap}
    \ov(\pi^*,\hat{\pi}):=\frac{1}{n}\sum_{u\in \cC}\one_{\hat{\pi}(u)=\pi^*(u)} \, ,
\end{equation} as well as through its \emph{error fraction} with $\pi^*$, defined as 
\begin{equation}\label{eq:def_error}
    \err(\pi^*,\hat{\pi}):=\frac{1}{n}\sum_{u\in \cC}\one_{\hat{\pi}(u) \neq \pi^*(u)} = \frac{\card{\cC}}{n} - \ov(\pi^*,\hat{\pi}) \, .
\end{equation}

Note that when the QAP formulation \eqref{eq:QAP} follows from a global maximisation of the posterior distribution under rigid constraints ($\pi$ has to be a permutation), the overlap appears as a relaxed objective, relaxing the global constraints and focusing on \emph{marginal} probabilities of a node $u$ in $G$ being matched with $u'$ in $H$. 

A sequence of injective estimators $\{\hat{\pi}_n\}_n$ -- omitting the dependence in $n$ -- is said to achieve
\begin{itemize}
    \item \emph{Exact recovery} if $\; \dP(\hat\pi=\pi^*) \underset{n \to \infty}{\longrightarrow} 1$,
    \item \emph{Almost exact recovery} if $\; \dP(\ov(\pi^*,\hat\pi)= 1-o(1)) \underset{n \to \infty}{\longrightarrow} 1$,
    \item \emph{Partial recovery} if there exists some $\eta>0$ such that $\; \dP(\ov(\pi^*,\hat\pi) > \eta) \underset{n \to \infty}{\longrightarrow} 1$,
    \item \emph{One-sided partial recovery} if it achieves partial recovery and $ \dP(\err(\pi^*,\hat\pi) = o(1)) \underset{n \to \infty}{\longrightarrow} 1$.
\end{itemize}
The probability $\dP$ in the above definitions encapsulates all randomness of the graph model; in our case this shall be $\G(n,q=\lambda/n,s)$.

\begin{remark}\label{remark:one_sided_partial}
One-sided partial recovery is by definition at least as hard as partial recovery. From an application standpoint it is more appealing than partial recovery: indeed, it may be of little use to know one has a permutation with 30\% of correctly matched nodes if one does not have a clue about which pairs are correctly matched. Our proposed algorithm \texttt{MPAlign} will achieve one-sided partial recovery under suitable conditions.

Note that in order to achieve partial recovery, the set $\cC$ associated with estimator $\hat{\pi}$ has to be of cardinality $\card{\cC} \geq \eta n$, see \eqref{eq:def_overlap}. In addition, one-sided partial recovery requires $\cC$ not to be too large, see \eqref{eq:def_error}.
\end{remark} 

\textbf{Phase diagram} 
One of the main questions consists in determining the \emph{phase diagram} of the model $\G(n,q=\lambda/n,s)$ for partial recovery, for which we here give a definition. We are interested in determining the range of parameters $(\lambda,s) \in \dR_+ \times [0,1]$ for which, in the large $n$ limit:
\begin{itemize}
\item Any sequence of estimators fails to achieve partial recovery for any $\eta>0$. We refer to the corresponding range as the \emph{impossible phase};
\item There is a sequence of estimators $\hat\pi$ achieving partial recovery (not necessarily one-sided) with some $\eta>0$, which we refer to as the \emph{IT-feasible phase};
\item There is a sequence of estimators $\hat\pi$ \emph{that can be computed in polynomial-time} achieving  partial recovery with some $\eta>0$ (and sometimes even more, achieving also one-sided partial recovery): the \emph{easy phase}.
\end{itemize}

An interesting perspective on this problem is provided by research on community detection, or graph clustering, for random graphs drawn according to the stochastic block model. In that setup, above the so-called Kesten-Stigum threshold, polynomial-time algorithms for clustering are known \cite{BLM18,SpectralRedemption13, Mossel2016}, and the consensus among researchers in the field is that no polynomial-time algorithms exist below that threshold. Yet, there is a range of parameters with non-empty interior below the Kesten-Stigum threshold for which exponential-time algorithms are known to succeed at clustering \cite{Banks2016InformationtheoreticTF}. In other words, for graph clustering, it is believed that there is a non-empty \emph{hard phase}, consisting of the set difference between the IT-feasible phase and the easy (polynomial-time feasible) phase.

The picture available to date for partial graph alignment is as follows. Recent work \cite{ganassali2021impossibility} shows that the impossible phase includes the range of parameters $\{(\lambda,s): \lambda s\leq 1\}$, and Wu et al. \cite{Wu2021SettlingTS} have established that the IT-feasible phase includes the range of parameters $\{(\lambda,s): \lambda s> 4\}$ (condition $\lambda s>C$ for some large $C$ had previously been established in \cite{Hall20}). For the easy phase, Ganassali and Massoulié \cite{Ganassali20a} have established that it includes the range of parameters $\{(\lambda,s):\lambda\in[1,\lambda_0],s\in[s(\lambda),1]\}$ for some parameter $\lambda_0>1$ and some function $s(\lambda):(1,\lambda_0]\to[0,1]$. The algorithm proposed in \cite{Ganassali20a} based on tree matching weights achieves in this regime one-sided partial recovery. Figure \ref{fig:phase_diagram} depicts a phase diagram describing these prior results together with the new results in this paper.

\begin{figure}[h]
    \centering 
    \hspace{-1.2cm}
    \includegraphics[scale=0.65]{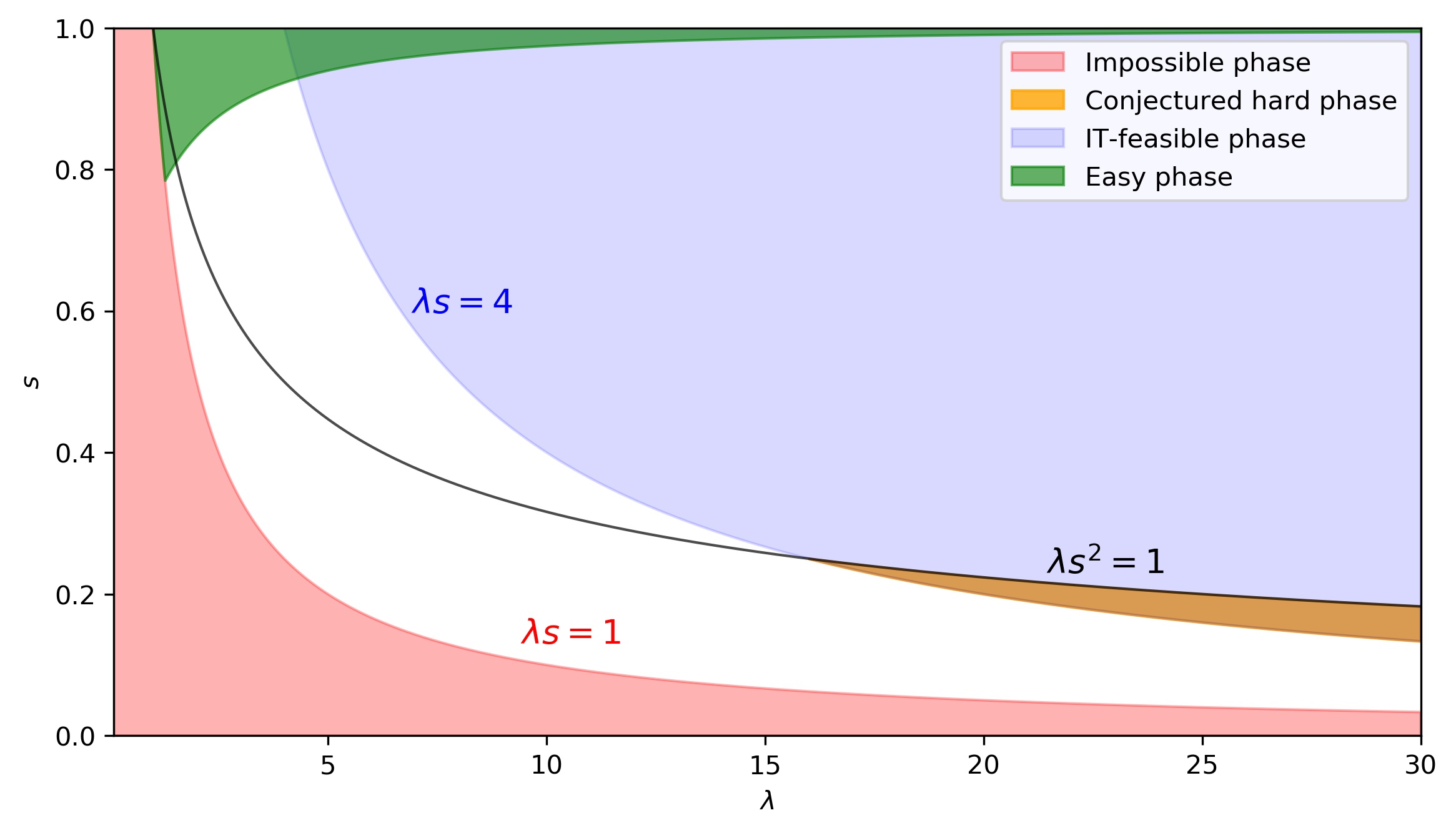}
    \caption{Diagram of the $(\lambda,s)$ regions where partial recovery is known to be impossible (\cite{ganassali2021impossibility}), IT-feasible (\cite{Wu2021SettlingTS}), or easy (\cite{Ganassali20a} and this paper). In the orange region, though partial graph alignment is IT-feasible, one-sided detectability is impossible in the tree correlation detection problem, and partial graph alignment is conjectured to be hard (this paper).}
    \label{fig:phase_diagram}
\end{figure}

\section*{Problem description and overview of contributions} 

This partial picture leaves open the question of whether, similarly to the case of graph clustering, graph alignment features a hard phase or not. The contribution of the present work can be summarized in three points:
\begin{itemize}
    \item[$(1)$] We investigate a fundamental statistical problem, which to the best of our knowledge had not been previously studied: hypothesis testing for correlation detection in trees. We study the regimes in which the optimal test on trees succeeds or fails in the setting when the trees are correlated Galton-Watson trees (see Sections \ref{section:notations}, \ref{section:KL}, \ref{section:autos_GW}, and \ref{section:hard_phase});
    
    \item[$(2)$] For this detection problem on trees, the computation of the likelihood ratio can be made recursively on the depth, which yields an optimal message-passing algorithm for this task running in polynomial-time in the number of nodes (see Section \ref{section:LR});
    
    \item[$(3)$] We finally remark that the previous detection problem on trees arises naturally from a local point of view in the related problem of one-sided partial recovery for graph alignment. In light of the previous analysis we then draw conclusions for our initial problem on graphs and doing so we precise the phase diagram shown in Figure \ref{fig:phase_diagram}, extending the regime for which one-sided partial alignment is provably feasible in polynomial time, and exhibiting the presence of a conjectured hard phase (see Theorem \ref{theorem:main_result_GRAPHS} and Section \ref{section:graph_matching}).
\end{itemize} 

Our approach to point $(3)$ follows the way originally paved in \cite{Ganassali20a}: it essentially relies on an algorithm which lets $\hat\pi(u)=u'$ for $u$ such that the local structure of graph $G$ in the neighborhood of node $u$ is 'close' to the local structure of graph $G'$ in the neighborhood of node $u'$. As exploited in \cite{Ganassali20a}, the neighborhoods to distance $d$ of two nodes $u,u'$ in $G$ and $G'$, provided that $u'=\pi^*(u)$, are asymptotically distributed as correlated Galton-Watson branching trees (distribution denoted $\dPls_{d}$). On the other hand, for pairs of nodes $(u,u')$ taken at random in $[n]$, the joint neighborhoods of nodes $u$ and $u$ in $G$ and $G'$ respectively, to depth $d$, are asymptotically distributed as a pair of independent Galton-Watson branching trees (distribution denoted $\dPl_{d}$).

Thus a fundamental step in our approach is to determine the efficiency of tests for deciding whether a pair of branching trees is drawn from either a product distribution, or a correlated distribution. \cite{Ganassali20a} relied on tests based on a so-called \emph{tree matching weight} to measure the similarity between two trees. In the present work we are instead interested in studying the existence of \emph{one-sided tests}, which are tests asymptotically guarantying a vanishing type I error and a non vanishing power. According to the Neyman-Pearson Lemma, optimal one-sided tests are based on the likelihood ratio $L_d$ of the distributions under the distinct hypotheses $\dPls_{d}$ and $\dPl_{d}$ (trees correlated or not)\footnote{This guarantees that whenever the test based on tree matching weight in \cite{Ganassali20a} succeeds, the optimal test studied in this paper also succeeds. On this point, Theorem \ref{theorem:suff_cond_KL} (see Section \ref{section:KL}) extends the sufficient conditions established in previous work \cite{Ganassali20a} for partial alignment (for small $\lambda$ and $s$ close to $1$).}. A general mathematical formalization of point $(1)$ here above is the following 
\begin{theorem}[Correlation detection in trees]\label{theorem:main_result_TREES}
Assume that\footnote{This cosmetic assumption guarantees that we place ourselves outside of a subset of the impossibility phase previously identified in \cite{ganassali2021impossibility}, see Figure \ref{fig:phase_diagram}. Assuming $\lambda s>1$ is in fact unnecessary, since it is proved that condition $(iv)$ is never satisfied when $\lambda s \leq 1$, see equation \eqref{eq:fixed_point_l} in Section \ref{subsection:martingale_prop}.} $\lambda s >1$. 
The following propositions are equivalent:
\begin{itemize}
    \item[$(i)$] There exists a one-sided test for deciding $\dPl_{d}$ versus $\dPls_{d}$.
    \item[$(ii)$] Let $ \KL_d $ be the Kullback-Leibler divergence between $\dPls_d$ and $\dPl_d$. We have $$\KL_d \underset{d \to \infty}{\longrightarrow} +\infty \, .$$
    \item[$(iii)$] There exists $(a_d)_d$ verifying $a_d \underset{d \to \infty}{\longrightarrow} +\infty$ such that $$\dPl_{d}(L_d>a_d) \underset{d \to \infty}{\longrightarrow} 0 \quad \mbox{and} \quad \liminf_{d\to\infty} \dPls_{d}(L_d>a_d) =: \beta > 0 \, . $$
    \item[$(iv)$] The martingale $(L_d)_d$ (with respect to $\dPl_{\infty}$) is not uniformly integrable (u.i.), i.e. the almost sure limit $L_\infty$ of $(L_d)_d$ satisfies $\dEl_\infty[L_\infty]<1$.
    \item[$(v)$] One has $$\dPls_{\infty}\left(\liminf_{d\to\infty}(\lambda s)^{-d} \log L_d >0\right)\geq 1-\pext(\lambda s) ,$$ where $\pext(\lambda s)$ is the probability that a Galton-Watson tree with  offspring distribution  $\Poi(\lambda s)$ gets extinct. 
\end{itemize}
\end{theorem}

\begin{remark}[On condition $(v)$]\label{rem:condition5}
Condition $(v)$ will be used in the design of the algorithm in Section \ref{section:graph_matching}, choosing an appropriate threshold that will guarantee for the method to output both a substantial part of the underlying permutation and a vanishing number of mismatches. 
\end{remark}

The above theorem is proved in Section \ref{section:LR} and gives general necessary and sufficient conditions for the existence of a one-sided test in the tree correlation detection problem. We emphasise that in addition, explicit sufficient conditions in terms of $\lambda$ and $s$ are obtained in Sections \ref{section:KL} (Theorem \ref{theorem:suff_cond_KL}), \ref{section:autos_GW} (Theorem \ref{theorem:suff_cond_auto}) and \ref{section:hard_phase} (Theorem \ref{theorem:suff_hard_phase}), conditions under which the equivalent conditions of Theorem \ref{theorem:main_result_TREES} hold (for Theorems \ref{theorem:suff_cond_KL} and \ref{theorem:suff_cond_auto}) or fail (for Theorem \ref{theorem:suff_hard_phase}). These results are globally summarized in Figure \ref{fig:phase_diagram} and constitute a substantial part of this work. 

The correspondence between tree correlation and graph alignment -- point $(3)$ here above -- enables to draw the following conclusion 
\begin{theorem}[Consequences for one-sided partial graph alignment]\label{theorem:main_result_GRAPHS}
For given $(\lambda,s)$, if one-sided correlation detection is feasible, i.e. any of the conditions in Theorem \ref{theorem:main_result_TREES} holds, then one-sided partial alignment in the correlated \ER model $\G(n,q=\lambda/n,s)$ is achieved in polynomial time by our algorithm \texttt{MPAlign} (Algorithm \ref{algo_GA}  in Section \ref{section:graph_matching}).
\end{theorem}

The above Theorem will be proved in Section \ref{section:graph_matching}, and builds upon the \emph{locally tree-like} property of the graphs sampled from the correlated \ER model. Let us conclude this introduction by forumlating the previously mentioned conjecture about the existence of a hard phase for sparse graph alignment.

\begin{conj}\label{conjecture:hard_phase}
We conjecture that if one-sided correlation detection in trees fails, i.e. none of the equivalent conditions in Theorem \ref{theorem:main_result_TREES} holds, then no polynomial-time algorithm achieves partial recovery in sparse graph alignment. In view of Theorem \ref{theorem:suff_hard_phase} of Section \ref{section:hard_phase}, which guarantees existence of a non-empty parameter region where one-sided tree detection fails while partial graph alignment can be done in non-polynomial time, our conjecture would imply the \emph{hard phase} to be non-empty.
\end{conj}

\textbf{Paper organization}
The outline of the paper is as follows. 
We first give in Section \ref{section:notations} a full description of the fundamental problem of testing tree correlation as well as precise constructions of models $\dPl_d$ and $\dPls_d$. 
We then dive into the derivation of the likelihood ratio and its properties in Section \ref{section:LR}, proving Theorem \ref{theorem:main_result_TREES} in Section \ref{subsection:proof_th1}.

Then, a first sufficient condition for one-sided tree detectability (Theorem \ref{theorem:suff_cond_KL}) is obtained in Section \ref{section:KL} by analyzing the Kullback-Leibler divergence: this condition is of the same kind as the one following from \cite{Ganassali20a}, however with a more direct derivation as well as a more explicit condition. 

Using a different approach, a second sufficient condition (Theorem \ref{theorem:suff_cond_auto}) is established in Section \ref{section:autos_GW} by analyzing the number of automorphisms of Galton-Watson trees. 

On the other side, we establish a sufficient condition for failure of one-sided detectability in Section \ref{section:hard_phase} (Theorem \ref{theorem:suff_hard_phase}). If true, Conjecture \ref{conjecture:hard_phase}, together with this condition would show that the hard phase is non-empty for graph alignment. 

Section \ref{section:graph_matching} takes us back to the graph alignment problem, introducing \texttt{MPAlign}, a message-passing method for aligning graphs which strongly relies on the tree correlation problem. Guarantees on this method as well as the proof of Theorem \ref{theorem:main_result_GRAPHS} are established. 

Section \ref{section:conclusion} brings this paper to a conclusion by raising several open questions as well as a discussion on optimal sparse graph alignment in the isomorphic case $(s=1)$. Also, we present a discussion on graph alignment in the isomorphic case in Appendix \ref{appendix:discussion}, some numerical experiments on \texttt{MPAlign2}, a slight variant of \texttt{MPAlign} in Appendix \ref{appendix:numerical}. Some additional proofs are deferred to Appendix \ref{appendix:additional_proofs}.

\section{Notations and problem statement}\label{section:notations}

\subsection{Notations and definitions}\label{subsection:notations}
In this first part we briefly introduce -- or recall -- some basic definitions that are used throughout the paper.\\

\textit{Basics.}
$\dN$ (resp. $\dR,\dR_+$) will denote the set of non-negative integers (resp. real numbers, non-negative real numbers). For all integers $n>0$, we define $[n] := \left\lbrace 1, 2, \ldots, n \right\rbrace$. For any finite set $A$, we denote by $\card{A}$ its cardinal. $\cS_{A}$ is the set of permutations on $A$. We also denote $\cS_{k} = \cS_{[k]}$ for brevity, and we will often identify $\cS_{k}$ to $\cS_{A}$ whenever $\card{A}=k$. For any $0 \leq k \leq \ell$, we will write $\cS(k,\ell)$ (resp. $\cS(A,B)$) for the set of injective mappings from $[k]$ to $[\ell]$ (resp. between finite sets $A$ and $B$). By convention, $\card{\cS(0,\ell)}=1$. Also, we will use the standard asymptotic Laudau notations: $o,O,\Omega$ and $\Theta$.\\

\textit{Graphs.}
A simple, non-oriented graph $G=(V,E)$ is defined by its node set $V$ and its edge set $E$ made of $2-$sets of distinct elements of $V$. In such a graph $G$, we denote by $\cN_{G,d}(u)$ (resp. $\cS_{G,d}(u)$) the set of vertices at graph distance $\leq d$ (resp. exactly $d$) from node $u$ in $G$. The \emph{neighborhood} of a node $u \in V$ is $\cN_G(u):= \cN_{G,1}(u)$, i.e. the set of all vertices that are connected to $u$ by an edge in $G$. The size of $\cN_G(u)$ is referred to as the \emph{degree} of node $u$. \\

\textit{Labeled rooted trees.} 
A \emph{labeled rooted tree} $t$ is an undirected graph with node set $V$ and edge set $E$ with no cycle, with a given distinguished node $\rho \in V$ called the \emph{root}. The \emph{depth} of a node is defined as its distance to the root $\rho$. The depth of tree $t$ is given as the maximum depth of all nodes in $t$. Each node $u$ at depth $d \geq 1$ has a unique \emph{parent} in $t$, which can be defined as the unique node at depth $d-1$ on the path from $u$ to the root $\rho$. Similarly, the \emph{children} of a node $u$ of depth $d$ are all the neighbors of $u$ at depth $d+1$. In addition, the nodes are labeled by ordering nodes' children - these sortings may be given, arbitrary or random, see hereafter. In any case, the labels must satisfy the following constraints. First, the label of the root node is set to the empty list $\varnothing$. Then, recursively, the label of a node $u$ is a list $\lbrace m,k \rbrace$ where $m$ is the label of its parent node, and $k$ is the rank of $u$ among the children of its parent. See Figure \ref{fig:example_trees} for an example. 

For any $u \in V$, we denote by $t_u$ the subtree of $t$ rooted at node $u$, and $c_t(u)$ -- or simply $c(u)$ where there is no ambiguity -- the number of children of $u$ in $t$. Note that for all $u \in V$, $t_u$ inherits a unique labeling from the initial labeling of $t$, by replacing the label of $u$ by $\varnothing$ in all the subtree.

We let $\cV_d(t)$ (resp. $\cL_d(t)$) be the set of nodes of $t$ at depth at most $d$ (resp. exactly $d$).\\

We denote by $\cX_d$ the collection of labeled rooted trees of depth at most $d$. Obviously, $\cX_0$ contains a single element, namely the rooted tree with only one node -- its root. Each tree $t$ in $\cX_d$ can be represented with a unique ordered list $(t_1,\ldots,t_{c(\rho)})$ where each $t_u$ is the subtree of $t$ rooted at node $u$ as defined above, and thus belongs to $\cX_{d-1}$. When $c(\rho)=0$, the previous ordered list is empty.\\

\textit{Relabelings.} 
A \emph{relabeling} $r(t)$ of a labeled rooted tree $t$ is a labeling of $t$ where for every node $u$, the rank of $u$ among all the children of its parent is given by a permutation -- the integer $k$ in the initial labeling is replaced by some $\sigma_v(k)$ where $v$ is the parent of $u$.
We define $\Rel(t)$ the set of relabelings of $t$, given as sets of permutations $\set{\sigma_v, v \in V}$. 

In the sequel we will consider \emph{uniform relabelings} $R(t)$, that are relabelings where the $\sigma_v$ are uniform. 

\begin{remark}\label{rem:relabelings}
    First, an easily verified property is that, for a given labeled tree $t \in \cX_d$, $R(t)$ is indeed uniformly distributed in $\Rel(t)$. 
    Second, note that since the random models of trees described hereafter in Section \ref{subsection:model_random_trees} are invariant by relabeling, correlation detection in unlabeled trees and correlation detection in uniformly (re)labeled trees are equivalent. 
   Third, note that distinct relabelings ion $\Rel(t)$ do not necessary give distinct rooted labeled trees: these relabelings will be called \emph{automorphisms} of $t$, see below.
\end{remark}

\textit{Injective mappings between labeled rooted trees.} 
For two labeled trees $\tau,t\in\cX_d$, the set of \emph{injective mappings (or, injections) from $\tau$ to $t$}, denoted $\cS(\tau,t)$, is the set of one-to-one mappings from the labels of vertices of $\tau$ to the labels of vertices of $t$ that preserve the rooted tree structure, in the sense that any $ \sigma  \in \cS(\tau,t)$ must verify 
\begin{equation*}
     \sigma (\varnothing)=\varnothing \quad \mbox{ and } \quad \sigma ( \left\lbrace m,k \right\rbrace)= \left\lbrace \sigma (m), j \right\rbrace \mbox{ for some $j$}.
\end{equation*} Note that $\cS(\tau,t)$ is not empty if and only if $\tau$ is, up to some relabeling, a subtree of $t$.\\

\textit{Automorphisms of rooted labeled trees.} Let $t \in \cX_d$. As mentioned in Remark \ref{rem:relabelings}, some elements of $\Rel(t)$ may be indistinguishable of $t$ (that is, they give a tree with same edges between same node labels). These relabelings are called \emph{automorphisms} of $t$, and their set is denoted by $\Aut(t)$.\\

\textit{Tree prunings.} The \emph{pruning of a labeled rooted tree $t$ at depth $d$} is the subtree of $t$ obtained by removing nodes at distance $>d$ from the root. We denote by $p_d : \bigcup_{d' \geq d} \cX_{d'} \to \cX_{d}$ the pruning operator at depth $d$.  See Figure \ref{fig:example_trees} for an example. \\

\textit{Tree subsampling.} For $s \in (0,1)$, a \emph{$s-$subsampling} of a labeled rooted tree $t$ is a rooted tree obtained by conserving every edge independently with probability $s$, and keeping the connected component of the root. The subsampling inherit a unique labeling by sorting the children according to their original labeling in $t$.\\

\begin{figure}[h]
     \centering
     \begin{subfigure}[b]{0.4\textwidth}
         \centering
         \includegraphics[scale=0.8]{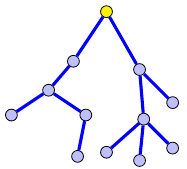}
         \caption{a rooted tree $t \in \cX_d$ (labels hidden)}
         \label{fig:t_unlabeled}
     \end{subfigure}
     \begin{subfigure}[b]{0.4\textwidth}
         \centering
         \includegraphics[scale=0.8]{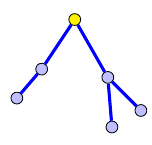}
         \caption{the pruning of $t$ at depth $2$ (labels hidden)}
         \label{fig:t_pruned}
     \end{subfigure}
     \hfill
     \begin{subfigure}[b]{\textwidth}
         \centering
         \vspace{0.4cm}
         \includegraphics[scale=0.8]{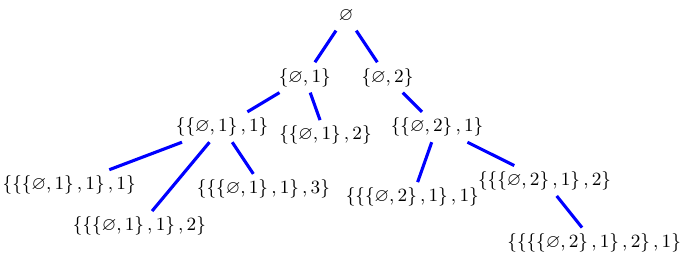}
         \caption{a random uniform labeling of $t$}
         \label{fig:t_labeled}
     \end{subfigure}
     
    \caption{A rooted tree $t$ of depth $d=4$ (the root is highlighted in yellow).}
    \label{fig:example_trees}
\end{figure}

\textit{Probability.} 
For the sake of readability, sometimes lowercase characters are used to distinguish deterministic objects from random variables (uppercase).
Some event $B$ depending on $n$ is said to be verified \emph{with high probability (w.h.p.)} if the probability of $B$ tends to $1$ when $n \to \infty$.

We also denote by $\Poi(\mu)$ the Poisson distribution of parameter $\mu \geq 0$, and $\p_\mu$ its density, namely for all $k \geq 0$,
\begin{equation*}
    \p_\mu(k) := e^{-\mu} \frac{\mu^k}{k!} \, .
\end{equation*}

For $\mu \geq 0$ and $d \geq 0$, $\GWmu_d$ denotes the distribution of a uniformly labeled rooted Galton-Watson tree of offspring $\Poi(\mu)$ stopped at depth $d$, which we do not redefine here.

\subsection{Models of random trees}\label{subsection:model_random_trees} 
The context of our problem is as follows. We observe two uniformly labeled rooted trees $t,t'$ 
and would like to test whether they are independent or correlated. We will consider two models of random rooted trees, which we present hereafter. For this purpose we first need to introduce the following\\

\emph{Tree augmentation.} For $\lambda >0, s \in [0,1]$ and $d \geq 0$, a (random) \emph{$(\lambda,s)-$augmentation} of a labeled rooted tree $\tau$ of depth at most $d$, denoted $\Augls_d(\tau)$, is defined as follows. First, to each node $u$ in $V_0$ of depth $<d$, we attach a number $Z^{+}_u$ of additional children, where the $Z^{+}_u$ are i.i.d. of distribution $\Poi(\lambda (1-s))$. Let $V^+$ be the set of these additional children. To each $v \in V^+$ at depth $d_v$, we attach another random tree of distribution $\GWl_{d-d_v}$, independently of everything else. This new tree is then relabeled uniformly at random. \\

We are now ready to describe the independent (resp. correlated) model $\dPl_d$ (resp. $\dPls_d$).

\begin{itemize}
    \item[$(i)$] \textbf{Independent model $\dPl_{d}$.} Under the independent model $\dPl_d$, $T$ and $T'$ are two independent $\GWl_d$ trees, where $\lambda>0$ is the mean number of children of a node in the tree. We denote $(T,T') \sim \dPl_{d}$.
    
    \item[$(ii)$] \textbf{Correlated model $\dPls_d$.} This model is built as follows: we start from a so-called \emph{intersection tree} $\tau^* \sim \GWls_d$, and we take $T$ and $T'$ to be two independent $(\lambda,s)-$augmentations of $\tau^*$. The two parameters are $\lambda>0$, the mean number of children of a node in the tree, and the correlation $s \in [0,1]$. We denote $(T,T') \sim \dPls_{d}$.

    \item[$(i)/(ii)$] \textbf{Labeling.} In both models, the trees $T$ and $T'$ are then both uniformly relabeled. 
\end{itemize}

It can easily be verified that $T$ and $T'$ have same marginals under $\dPl_{d}$ and $\dPls_d$, namely $\GWl_d$. These models are illustrated in Figure \ref{fig:samples_P01}.

\begin{remark}
    Let us mention two easy facts. First, the correlated model when $s=0$ coincides with the independent model, namely $\dP^{(\lambda, s=0)}_d \overset{(d)}{=} \dPl_d$. Second, the definitions of both models are still consistent for $d=\infty$, which will be used in the sequel.
\end{remark}

Another interesting fact which we will use implicitly in the sequel is the following 
\begin{fact}[Pruning\footnote{see Section \ref{subsection:notations} for a definition.} consistency]
    Let $0 \leq d \leq d'$. If $(T,T') \sim \dPl_{d'}$ (resp. $(T,T') \sim \dPls_{d'}$), then $(p_d(T),p_d(T')) \sim \dPl_{d}$ (resp. $(p_d(T),p_d(T')) \sim \dPls_{d'}$).
\end{fact}

\subsection{Hypothesis testing, one-sided tests}\label{subsection:hyp_testing} 
The hypothesis testing considered in this study can be formalized as follows: given the observation of a pair of trees $(t,t')$ in $\cX_d \times \cX_d$, we want to test
\begin{equation}
    \label{eq:test_hypotheses_p}
    \cH_0 = \mbox{"$t,t'$ are realizations under $\dPl_{d}$"} \quad \mbox{versus} \quad \cH_1 = \mbox{"$t,t'$ are realizations under $\dPls_{d}$"}.
\end{equation} 
More specifically, we are interested in being able to ensure the existence of a (asymptotic) \textit{one-sided test}, that is a test $\cT_d: \cX_d \times \cX_d \to \left\lbrace 0,1 \right\rbrace$ such that $\cT_d$ chooses hypothesis $\cH_0$ under $\dPl_{d}$ with probability $1-o(1)$, and chooses $\cH_1$ with some positive probability uniformly bounded away from 0 under $\dPls_{d}$. In other terms, a one-sided test asymptotically guarantees a vanishing type I error and a non vanishing power.

\begin{figure}[H]
     \centering
     \begin{subfigure}[b]{\textwidth}
         \centering
         \includegraphics[scale=0.4]{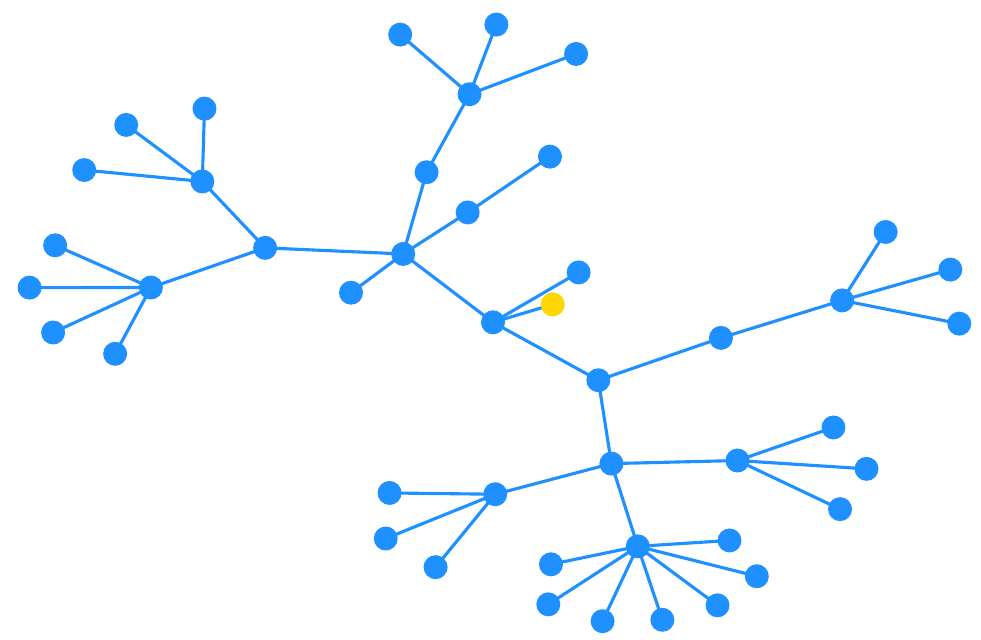}
         \hspace{0.25cm}
         \includegraphics[scale=0.4]{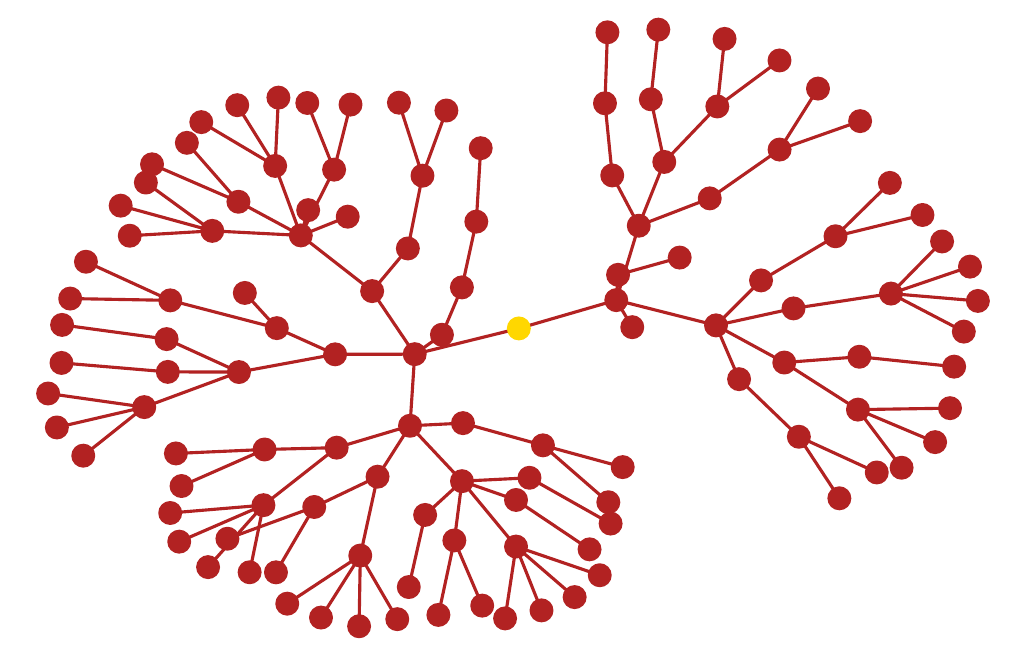}
         \caption{A realization of $T,T'$ from $\dPl_{d}$.}
         \label{fig:P0}
     \end{subfigure}
     \begin{subfigure}[b]{\textwidth}
        \vspace{0.3cm}
         \centering
         \includegraphics[scale=0.4]{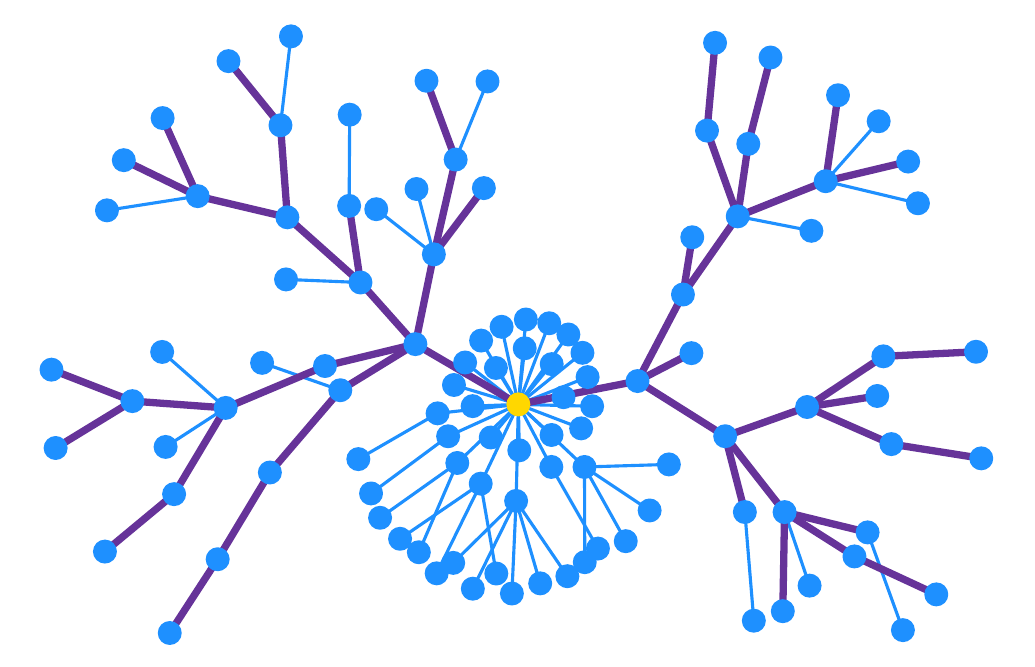}
         \hspace{0.25cm}
         \includegraphics[scale=0.4]{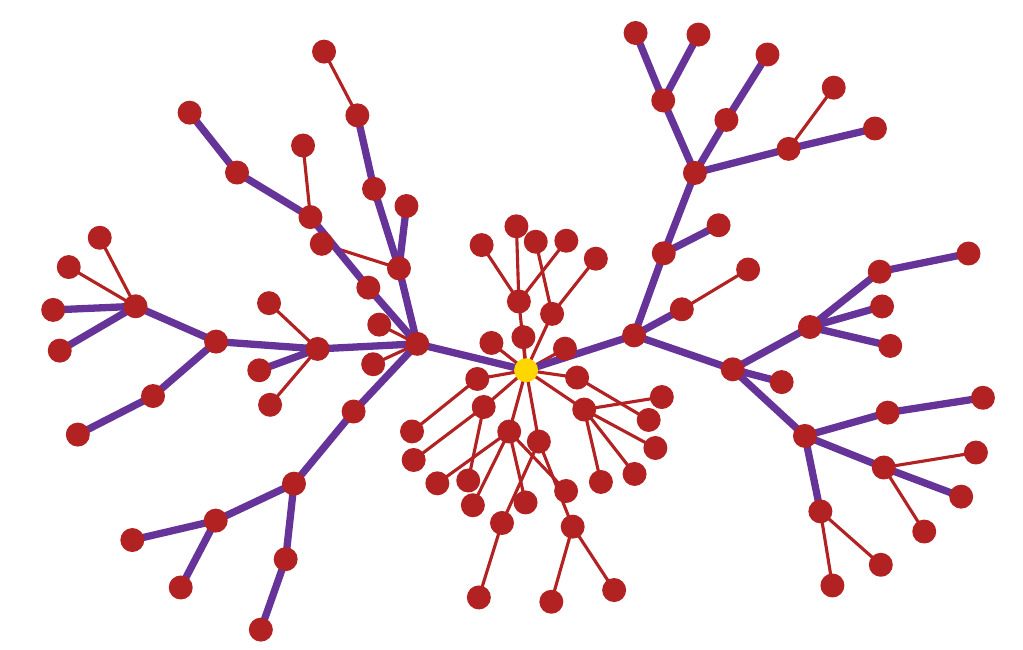}
         \caption{A realization of $T,T'$from $\dPls_{d}$. The intersection tree $\tau^*$ is drawn thick and purple.}
         \label{fig:P0}
     \end{subfigure}
     
    \caption{Samples from models $\dPl_{d}$ and $\dPls_{d}$, with $\lambda = 1.8$, $s=0.8$, and $n=5$. The root node is highlighted in yellow. Labels are not shown.}
    \label{fig:samples_P01}
\end{figure}

\begin{remark}\label{remark:one_sided_tests}
In statistical detection theory, commonly considered asymptotic properties of tests are
\begin{itemize}
	\item \emph{strong detection}, i.e. tests $\cT_d: \cX_d \times \cX_d \to \left\lbrace 0,1 \right\rbrace$  that verify
	\begin{equation*}
	\underset{d \to \infty}{\lim} \left[\dPl_d\left( \cT_d(t,t') = 1 \right) + \dPls_d\left( \cT_d(t,t') = 0 \right)\right] = 0,
	\end{equation*}
	\item \emph{weak detection}, i.e. tests $\cT_d: \cX_d \times \cX_d \to \left\lbrace 0,1 \right\rbrace$ that verify
	\begin{equation*}
	\underset{d \to \infty}{\limsup} \left[\dPl_d\left( \cT_d(t,t') = 1 \right) + \dPls_d\left( \cT_d(t,t') = 0 \right)\right] <1 \, .
	\end{equation*}
\end{itemize} In words, strong detection corresponds to correctly discriminating with high probability between $\dPl_d$ and $\dPls_d$, whereas weak detection corresponds to strictly outperforming random guessing. It is well-known that the likelihood ratio test achieves the minimal value of $\dPl_d\left( \cT_d(t,t') = 1 \right) + \dPls_d\left( \cT_d(t,t') = 0 \right)$, which is the sum of type I and type II error, and that this minimal value is given by $1 - \DTV(\dPl_d,\dPls_d)$, where $$\DTV(\dPl_d,\dPls_d) := \frac{1}{2} \sum_{(t,t') \in \cX_d^2} |\dPl_d(t,t') - \dPls_d(t,t')| $$ denotes the total variation distance
between $\dPl_d $ and $\dPls_d$. Hence, strong detection (resp. weak detection) holds if and only if $\, \DTV(\dPl_d,\dPls_d) \underset{d \to \infty}{\longrightarrow} 1$ (resp. $\underset{d \to \infty}{\liminf} \, \DTV(\dPl_d,\dPls_d) > 0$).

We now argue that these are not the right notions for our problem. For this, consider the subset $B_d$ of $\cX_d \times \cX_d$ made of a pair of trivial trees (that is trees only consisting in a single node). Note that for all $d>0$, $\dPl_d(B_d)=e^{-2\lambda}$ and $\dPls_d(B_d)=e^{-2\lambda+\lambda s}$.
\begin{itemize}
    \item First, let us prove that strong detection never holds. For all $d>0$,
    \begin{flalign*}
    \DTV(\dPl_d,\dPls_d) & = \frac{1}{2}|\dPl_d(B_d) - \dPls_d(B_d)| + \frac{1}{2} \sum_{(t,t') \in \cX_d^2 \setminus B_d} |\dPl_d(t,t') - \dPls_d(t,t')| \\
    & = \frac{1}{2}(e^{-2\lambda+\lambda s}-e^{-2\lambda}) + \frac{1}{2} \sum_{(t,t') \in \cX_d^2 \setminus B_d} |\dPl_d(t,t') - \dPls_d(t,t')| \\
    & \leq \frac{1}{2}(e^{-2\lambda+\lambda s}-e^{-2\lambda}) + \frac{1}{2}(1-\dPl_d(B_d)) +  \frac{1}{2}(1-\dPls_d(B_d)) \\
    & \leq 1 - e^{-2\lambda} \, .
    \end{flalign*} This bound is uniform in $d$ and shows that strong detection never holds. 

\item Second, weak detection is always achievable as soon as $s>0$: indeed, one has 
\begin{flalign*}
    \DTV(\dPl_d,\dPls_d) & \geq \frac{1}{2}|\dPl_d(B_d) - \dPls_d(B_d)| \\
    & \geq \frac{1}{2} e^{-2\lambda}(e^{\lambda s}-1) \, ,
    \end{flalign*} this uniform bound being positive as soon as $s>0$, so that weak detection holds.

\end{itemize}
Finally, a test of tree correlation yields efficient algorithms for graph alignment in the associated sparse correlated  \ER model if it achieves a positive power (non-vanishing alarm detection) and a vanishing type I (false alarm) error. Indeed the candidate vertex pairs returned by the algorithm will then contain i) a non-negligible fraction of correctly matched pairs by the first property, and ii) a negligible fraction of incorrectly matched pairs by the second property. 
\end{remark}

\section{Properties of the likelihood ratio} \label{section:LR}
In the following section, we derive the likelihood ratio in our problem and establish some of its properties that are key to the upcoming analysis.
Recall that for all $d \geq 0$, the likelihood ratio $L_d$ is defined for all $t,t'\in \cX_d$ as\footnote{For the purpose of martingale properties established in Section \ref{subsection:martingale_prop}, we need to extend slightly the definition of $L_d$ for trees that may not be in $\cX_d$. This is simply done by pruning the trees, namely
$$
L_d(t,t'):= \frac{\dPls_{d}(p_d(t),p_d(t'))}{\dPl_{d}(p_d(t),p_d(t'))},
$$ which is consistent with \eqref{eq:def_LR}.}
\begin{equation}
    \label{eq:def_LR}
    L_d(t,t'):= \frac{\dPls_{d}(t,t')}{\dPl_{d}(t,t')} \, .
\end{equation}

\subsection{Recursive computation}\label{subsection:recursive_computations}
We first show a recursive representation of the likelihood ratio $L_d$. First note that under $\dPl_d$, for $t,t' \in \cX_d$ with root degrees $c$ and $c'$, we have
\begin{equation} \label{eq:P0_GW}
\dPl_{d}(t,t')= \GWl_{d}(t) \cdot \GWl_{d}(t'),
\end{equation} and that we have the following recursion for $\GWl_{d}$, writing $t = (t_1, \ldots, t_c)$:
\begin{equation}\label{eq:recursion_GW}
\GWl_{d}(t) =\p_{\lambda}(c)\prod_{u \in [c]} \GWl_{d-1}(t_u) \, ,
\end{equation} the same equality being true for $t'$, replacing $c$ by $c'$. 

Under $\dPls_{d}$, when observing $t$ and $t'$ in  $\cX_d$ with root degree $c,c'$, we partition first on the number $0 \leq k \leq c \land c'$ of children of the root in $t$ and $t'$ belonging to $\tau^*$ (the underlying intersection tree), and second on mappings $\sigma \in\cS(k,c)$ and $ \sigma' \in \cS(k,c)$ making these children correspond pairwise. This gives
\begin{multline*}
\dPls_{d}(t,t') = \sum_{k=0}^{c \wedge c'} \p_{\lambda s}(k) \p_{\lambda (1-s)}(c-k)\p_{\lambda (1-s)}(c'-k) \\
\times \sum_{\substack{\sigma \in \cS(k,c) \\ \sigma' \in \cS(k,c')}}\frac{(c-k)! \cdot (c-k')! }{c! \cdot c'!} \left(\prod_{u=1}^k\dPls_{d-1}(t_{\sigma(u)},t'_{\sigma'(u)})\right) \left(\prod_{u=k+1}^d \GWl_{d-1}(t_{\sigma(u)})\right) 
\left(\prod_{u=k+1}^{d'} \GWl_{d-1}(t'_{\sigma'(u)})\right) \, .
\end{multline*}

This together with Equations \eqref{eq:P0_GW}, \eqref{eq:recursion_GW} readily implies the following recursive formula for the likelihood ratio $L_d$:

\begin{lemma}[Recursive formula for $L_d$]\label{lemma:LR_rec}
    We have
    \begin{equation}\label{eq:lemma:LR_rec}
L_d(t,t')=\sum_{k=0}^{c \wedge c'}\Psi(k,c,c')\sum_{\substack{\sigma \in \cS(k,c) \\ \sigma' \in \cS(k,c')}}\prod_{u=1}^k L_{d-1}(t_{\sigma(u)},t'_{\sigma'(u)}),
\end{equation} where $c$ (resp. $c'$) is the degree of the root in $t$ (resp. in $t'$), and $\Psi(k,c,c')$ is the following shorthand notation
\begin{equation}\label{eq:lemma:LR_rec_psi}
\Psi(k,c,c') := \frac{\p_{\lambda s}(k)\p_{\lambda (1-s)}(c-k)\p_{\lambda(1-s)}(c'-k)} {\p_\lambda(c)\p_\lambda(c')} \cdot \frac{(c-k)! \cdot (c-k')! }{c! \cdot c'!} =  \frac{e^{\lambda s} s^k (1-s)^{c+c'-2k}}{\lambda^{k} k!} \, .
\end{equation}

\end{lemma}

\begin{remark}\label{remark:util_rec_algo}
The above expression \eqref{eq:lemma:LR_rec} in Lemma \ref{lemma:LR_rec} will be useful for efficient computations of the likelihood ratio through message-passing in the \texttt{MPAlign} method -- see Algorithm \ref{algo_GA} in Section \ref{section:graph_matching}.
\end{remark}

\subsection{Explicit computation}
We will now use the recursive expression of Lemma \ref{lemma:LR_rec} to prove by induction on $d$ the following explicit formula for $L_d$:
\begin{lemma}[Explicit formula for $L_d$]\label{lemma:LR_developed}
With the previous notations, we have
\begin{equation}\label{eq:lemma:LR_developed}
L_d(t,t')=\sum_{\tau \in \cX_d} \sum_{\substack{\sigma \in \cS(\tau,t) \\ \sigma' \in \cS(\tau,t')} }\prod_{u \in \cV_{d-1}(\tau)}\Psi\left(c_\tau(u),c_t(\sigma(u)),c_{t'}(\sigma'(u))\right), 
\end{equation} where we recall that $\Psi$ is defined in \eqref{eq:lemma:LR_rec_psi}.
\end{lemma}

\begin{proof}[Proof of Lemma \ref{lemma:LR_developed}]
We prove this result by recursion. An empty product being set to $1$, there is nothing to prove in the case $d=0$. Let us first establish formula \eqref{eq:lemma:LR_developed} for $d=1$. In that case, trees $t$, $t'$ of depth $1$ are identified by the degrees $c$, $c'$ of their root node. Since $\cX_0$ is a singleton, $L_0$ is identically 1, and from \eqref{eq:lemma:LR_rec} we have that
\begin{equation}\label{eq:expr_L1_rec}
    L_1(t,t')=\sum_{k=0}^{c \wedge c'}\frac{\p_{\lambda s}(k)\p_{\lambda (1-s)}(c-k)\p_{\lambda(1-s)}(c'-k)}{\p_\lambda(c)\p_\lambda(c')}.
\end{equation} On the other hand, when evaluating expression \eqref{eq:lemma:LR_developed} for $d=1$ we only need consider trees $\tau$ in $\cX_1$ with root degree $k\leq c\wedge c'$, since for a larger $k$ one of the two sets $\cS(\tau,t)$ or $\cS(\tau,t')$ will be empty. When $\tau$ is the tree of $\cX_1$ with root degree $k$, we have $$\card{\cS(\tau,t)}= \frac{c!}{(c-k)!} \quad \mbox{and} \quad \card{\cS(\tau,t')}= \frac{c'!}{(c'-k)!} \, .$$ The right-hand term in \eqref{eq:lemma:LR_developed} thus writes
\begin{equation*}
    \sum_{k=0}^{c \wedge c'} \frac{c! \times c'!}{(c-k)! \times (c'-k)!} \Psi(k,c,c'),
\end{equation*} which gives precisely \eqref{eq:expr_L1_rec}.  Now assume that \eqref{eq:lemma:LR_developed} has been established up to $d-1\geq 1$. Expressing $L_{d}$ in terms of $L_{d-1}$ based on \eqref{eq:lemma:LR_rec}, and replacing in there the explicit expression of $L_{d-1}$ in \eqref{eq:lemma:LR_developed}, we get
\begin{multline*}
   L_{d}(t,t') = \sum_{k=0}^{c \wedge c'} {\Psi(k,c,c')}\\ \times
   \sum_{\substack{\sigma \in \cS(k,c) \\ \sigma' \in \cS(k,c')}} \prod_{u=1}^k \left[ \sum_{\tau_u\in\cX_{d-1}} \sum_{\substack{\sigma_u \in \cS(\tau_u,t_{\sigma(u)}) \\ \sigma'_u \in \cS(\tau_u,t'_{\sigma(u)})} }\prod_{v\in \cV_{d-1}(\tau_u)}\Psi\left(c_{\tau_u}(v),c_{t_{\sigma(u)}}(\sigma_u(v)),c_{t'_{\sigma'(u)}}(\sigma'_u(v))\right)\right]. 
\end{multline*}

Moreover, there is a bijective correspondence between 
$$
\begin{cases}
\mbox{an integer $0 \leq k \leq c \land c'$,} \\
\mbox{pairs of injections $\sigma \in \cS(k,c)$, $\sigma' \in \cS(k,c')$,} \\
\mbox{$k$ trees $\tau_1, \ldots ,\tau_k\in\cX_{d-1}$},\\
\mbox{injections $\sigma_u \in \cS(\tau_u,t_{\sigma(u)})$ for all $u \in [k]$},\\
\mbox{injections $\sigma'_u \in \cS(\tau_u,t'_{\sigma'(u)})$ for all $u \in [k]$},
\end{cases}
\; \mbox{and} \quad
\begin{cases}
\mbox{a tree $\tau \in \cX_d$ with root degree $\leq c\land c'$}, \\
\mbox{an injection $\Sigma \in \cS(\tau,t)$}, \\
\mbox{an injection $\Sigma' \in \cS(\tau,t')$} \, .
\end{cases}
$$
This establishes formula \eqref{eq:lemma:LR_developed} at step $d$, and hence Lemma \ref{lemma:LR_developed}.
\end{proof} 

\subsection{Martingale properties of $L_d$}\label{subsection:martingale_prop}
In this part, we establish that $L_d$ is a martingale, with respect to a good filtration and probability distribution. To establish such a property, assume that $T,T'$ are drawn under model $\dPl_{\infty}$. For $d \geq 0$, we define $$\cF_d:=\sigma(p_d(T),p_d(T'))$$ the sigma-field spanned by the observation of the two trees $T,T'$ up to depth $d$. 
\begin{proposition}\label{prop:LR_martingale}
The stochastic process $$\set{L_d = L_d(p_d(T),p_d(T'))}_{d \geq 0}$$ is a $\cF_d$-martingale under $\dPl_{\infty}$.
\end{proposition}
The above martingale property is not specific to the structure of our problem and follows from general properties of likelihood ratios. It is however informative to derive it by calculus, which we now do.
\begin{proof}[Proof of Proposition \ref{prop:LR_martingale}]
There are several ways to see that $\set{L_d}_{d \geq 0}$ is a $\cF_d$-martingale under $\dPl_{\infty}$, depending on the formula used to write $L_{d+1}$ in terms of $L_d$. We here choose to use the developed expression \eqref{eq:lemma:LR_developed} of Lemma \ref{lemma:LR_developed}, enabling simple computations:
\begin{flalign}\label{eq:proof:lemma:LR_martingale_1}
    L_{d+1} &  = \sum_{\tau \in \cX_{d+1}} \sum_{\substack{\sigma \in \cS(\tau,T) \\ {\sigma}' \in \cS({\tau},T')} }\prod_{u \in \cV_{d}({\tau})}\Psi\left(c_{{\tau}}(u),c_T({\sigma}(u)),c_{T'}({\sigma}'(u))\right) \nonumber \\
     & = \sum_{\tau_0 \in \cX_{d}} \sum_{\substack{\sigma \in \cS({\tau_0},p_d(T)) \\ \sigma' \in \cS(\tau_0,p_d(T'))} }\prod_{u \in \cV_{d-1}(\tau_0)}\Psi\left(c_{\tau_0}(u),c_{p_d(T)}(\sigma(u)),c_{p_d(T')}(\sigma'(u))\right) \nonumber \\
     & \quad \times \prod_{v \in \cL_d(\tau_0)} \sum_{k = 0}^{c_{T}(\sigma(v)) \wedge c_{T'}(\sigma'(v))} \frac{[c_{T}(\sigma(v))]! \cdot [c_{T'}(\sigma'(v))]!}{[c_{T}(\sigma(v))-k]! \cdot [c_{T'}(\sigma'(v))-k]!} \Psi(k,c_{T}(\sigma(v)),c_{T'}(\sigma'(v))) \, .
\end{flalign}
The last product in the last line of \eqref{eq:proof:lemma:LR_martingale_1} is independent from $\cF_d$. Moreover, under $\dPl_{\infty}$, all terms in the last product are independent, the $c_{T}(v)$ and $c_{T'}(v)$ being i.i.d. $\Poi(\lambda)$ random variables. By definition of $\Psi$ \eqref{eq:lemma:LR_rec_psi}, each term is of the form 
$$ X(c,c') := \sum_{k=0}^{c \wedge c'}\frac{\p_{\lambda s}(k)\p_{\lambda (1-s)}(c-k)\p_{\lambda(1-s)}(c'-k)}{\p_\lambda(c)\p_\lambda(c')}, $$
where $(c,c')$ are i.i.d. $\Poi(\lambda)$ variables. It is then straightaway to check that $\dE[X(c,c')]=1$. Since the second line in \eqref{eq:proof:lemma:LR_martingale_1} is $\cF_d-$measurable, taking the expectation conditionally to $\cF_d$ in \eqref{eq:proof:lemma:LR_martingale_1} entails
\begin{flalign*}
    \dE[L_{d+1} | \cF_d ] 
     & = \sum_{\tau_0 \in \cX_{d}} \sum_{\substack{\sigma \in \cS({\tau_0},p_d(T)) \\ \sigma' \in \cS(\tau_0,p_d(T'))} }\prod_{u \in \cV_{d-1}(\tau_0)}\Psi\left(c_{\tau_0}(u),c_{p_d(T)}(\sigma(u)),c_{p_d(T')}(\sigma'(u))\right) \times 1\\
     & = L_{d},
\end{flalign*}
hence the desired martingale property.
\end{proof}

The martingale property established in Proposition \ref{prop:LR_martingale} has several interesting corollaries, that will be developed in the rest of the paper. A first one -- which may be the most natural one could think of -- is to consider the \emph{almost sure convergence} of $(L_d)_{d \geq 0}$. The martingale being non-negative, we can now consider the martingale $\dPl_{\infty}-$almost sure limit $L_{\infty}$, and define $\ell :=\dEl_{\infty}\left[ L_{\infty} \right]$. Using the recursive formula \eqref{eq:lemma:LR_rec} of Lemma \ref{lemma:LR_rec}, conditioning on the root degrees $c$ and $c'$, and taking the limit $d \to \infty$ gives the equality in distribution 
\begin{equation}\label{eq:l_infty_1}
    L_\infty \overset{(d)}{=} \sum_{k=0}^{c \land c'} \Psi(k,c,c') \sum_{\sigma\in\cS_c,\sigma'\in\cS_{c'}}\prod_{u=1}^k L_{\infty, (\sigma(u),\sigma'(u))}\, ,
\end{equation} where $c,c'$ are i.i.d. $\Poi(\lambda)$, and the $(L_{\infty,v,v'})_{v \in [c], v' \in [c']}$ are identically distributed as $L_\infty$, and such that $L_{\infty,v,v'}$ and $L_{\infty,w,w'}$ are independent when $v \neq w$ and $v' \neq w'$. Taking the expectation in \eqref{eq:l_infty_1} yields 
\begin{flalign}\label{eq:fixed_point_l}
\ell & = \dE_{c,c'} \left[ \sum_{k \geq 0} \one_{c \geq k}  \one_{c' \geq k} \Psi(k,c,c') \frac{c! \cdot c'!}{(c-k)! \cdot (c'-k)!} \ell^k \right] \nonumber \\
& = \sum_{k \geq 0} \sum_{\substack{c, c' \geq 0}} \p_\lambda(c) \p_\lambda(c')  \one_{c \geq k}  \one_{c' \geq k} \frac{\p_{\lambda s}(k)\p_{\lambda (1-s)}(c-k)\p_{\lambda(1-s)}(c'-k)}{\p_\lambda(c)\p_\lambda(c')} \ell^k \nonumber \\
& = \sum_{k \geq 0} \p_{\lambda s}(k) \ell^k  \sum_{c \geq 0} \one_{c \geq k}  \p_{\lambda (1-s)}(c-k) \sum_{c' \geq 0} \one_{c' \geq k} \p_{\lambda(1-s)}(c'-k) = \sum_{k \geq 0}  \p_{\lambda s}(k) \ell^k \, .
\end{flalign} 
Hence $\dEl_\infty[L_\infty]$ verifies the previous fixed point equation \eqref{eq:fixed_point_l}, which is also (!) the fixed point equation for the extinction probability $\pext(\lambda s)$ of a Galton-Watson branching process with offspring distribution $\Poi(\lambda s)$.  

The question of knowing whether $\dEl_\infty[L_\infty]=1$ when $\lambda s > 1$ is crucial for feasibility of one-sided detection, by Theorem \ref{theorem:main_result_TREES}. For $\lambda s \leq 1$, the only solution of \eqref{eq:fixed_point_l} is $\ell=1$, hence condition $(iv)$ in Theorem \ref{theorem:main_result_TREES} is never satisfied when $\lambda s \leq 1$, which is in line with the results of \cite{ganassali2021impossibility} showing that one-sided partial\footnote{\cite{ganassali2021impossibility} shows that even partial recovery is impossible when $\lambda s \leq 1$.} recovery is impossible when $\lambda s \leq 1$.  
For $\lambda s>1$, the equation \eqref{eq:fixed_point_l} however admits a non-trivial solution $\pext(\lambda s) \in (0,1)$. 

We are now ready to prove Theorem \ref{theorem:main_result_TREES}, which is the object of next section. 

\subsection{Proof of Theorem \ref{theorem:main_result_TREES}}\label{subsection:proof_th1}
Recall that the $\KL-$divergence is obtained from $L_d$ as follows:
\begin{equation}\label{eq:def_KL}
\KL_d=\KL(\dPls_{d}\Vert\dPl_{d})= \dEls_{d} \left[ \log(L_d) \right] = \dEls_{\infty} \left[ \log(L_d) \right]\, .
\end{equation}

We start with a clarifying remark.
\begin{remark}[On condition $(iv)$]\label{rem:condition4}
Let us give a proof of the claimed equivalence inside condition $(iv)$, that is a proof of the fact that $(L_d)_d$ is not uniformly integrable (u.i. hereafter) iff $\dEl_\infty[L_\infty]<1$. First, recall that uniform integrability of a martingale is equivalent to its $L^1$ and almost sure convergence. We already know that $L_d \underset{d \to \infty}{\longrightarrow} L_\infty$ a.s. One one hand, if $\dEl_\infty[L_\infty]<1$ then $(L_d)_d$ does not converge in $L^1$, hence is not u.i.

One the other hand, let us assume that $\dEl_\infty[L_\infty] \geq 1$, then $\dEl_\infty[L_\infty] = 1$ by the fixed point equation \eqref{eq:fixed_point_l}. Since $L_d \underset{d \to \infty}{\longrightarrow} L_\infty$ a.s. and $\dEl_d[L_d] = 1 \underset{d \to \infty}{\longrightarrow} 1 = \dEl_\infty[L_\infty]$, we can apply Scheffé's Lemma to conclude that $(L_d)_d$ converges in $L^1$ and is thus u.i. 
\end{remark}

We can now dive into the proof.

\begin{proof}[Proof of Theorem \ref{theorem:main_result_TREES}]
\, 

\proofstep{Step 1: proof of $(iii) \implies (iv)$} 
Assume $(iii)$, i.e. that there exists a sequence $(a_d)_d$ such that $a_d \underset{d \to \infty}{\longrightarrow} \infty$, $\dPl_d(L_d > a_d) \underset{d \to \infty}{\longrightarrow} 0$ and $\liminf_d \dPls_d(L_d > a_d) =: \beta >0$. By definition, $(L_d)_d$ is u.i. if and only if 
$$ \lim_{K \to \infty} \left( \sup_d \dEl_d[|L_d| \one_{L_d\geq K}]\right) = \lim_{K \to \infty} \left( \sup_d \dPls_d(L_d\geq K)\right) = 0 \, .$$
For all $K>0$ there exists $d_K$ such that $a_d >K$ for all $d \geq d_K$, hence $\sup_d \dPls_d(L_d\geq K) \geq \beta$. This contradicts the previous condition and thus $(L_d)_d$ is not u.i. \\

\proofstep{Step 2: proof of $(iv) \implies (iii)$} 
Assume $(iv)$, that is that $\dEl_\infty[L_\infty]<1$. Let $\gamma := 1 -\dEl_\infty[L_\infty]$. Let $\cC$ be the set of continuity points of $L_\infty$ under $\dPl_\infty$. 
For all $a \in \cC$ and $d>0$, one has:
\begin{equation}\label{eq:4implies3_1}
1 = \dPls_d(L_d >a) + \dEl_d[L_d \one_{L_d \leq a}] 
\end{equation}
and since $a \in \cC$, by dominated convergence one has 
\begin{equation}\label{eq:4implies3_2}
    \dEl_d[L_d \one_{L_d \leq a}] \underset{d \to \infty}{\longrightarrow} \dEl_\infty[L_\infty \one_{L_\infty \leq a}] \leq \dEl_\infty[L_\infty] \, .
\end{equation} Putting together \eqref{eq:4implies3_1} and \eqref{eq:4implies3_2} gives that for all $a \in \cC$
\begin{equation}\label{eq:4implies3_3}
    \liminf_d \dPls_d(L_d>a) \geq 1 - \dEl_\infty[L_\infty] = \gamma \, .
\end{equation}

We will now build a suitable sequence $(a_d)_d$ for proving assumption $(iii)$. For all $a \geq 0$, define 
\begin{equation}\label{eq:4implies3_4}
d(a) := \inf\set{k \in \dN, \, \forall d' \geq k, \, \dPls_{d'}(L_{d'}>a) \geq \gamma/2} - 1\, .
\end{equation}
Note that $d(a)$ is finite since \eqref{eq:4implies3_3} holds for all $a \in \cC$ and $\cC$ is not upper-bounded (its complementary in $\dR_+$ being at most countable). By definition, the map $a \mapsto d(a)$ is non-increasing from $\dR_+$ to $\dN$. Now define for all $d \in \dN$

\begin{equation}\label{eq:4implies3_5}
a_d := \sup\set{a' \in \dR_+, \, d(a') \leq d} \, .
\end{equation} Note that $a_d$ is well defined since $d(0)=0$ and $a_d < \infty$ for all $d \in \dN$ (otherwise, we would have $\dPls_{d}(L_{d} = \infty) \geq \gamma/2$ for $d$ large enough, which is absurd since $L_{d} < \infty$ almost surely). 

Let us check that $a_d  \underset{d \to \infty}{\longrightarrow} \infty$. The sequence $(a_d)_d$ is non-decreasing by definition, and if $(a_d)_d$ is bounded by $A$, then for all $a'>A$, $d(a')=\infty$, which is absurd by the above. Hence $a_d  \underset{d \to \infty}{\longrightarrow} \infty$.

By definition, for all $d \in \dN$ we have $\dPls_{d}(L_{d} > a_d) \geq \gamma/2$, and Markov's inequality yields $\dPl_{d}(L_{d} > a_d) \leq \frac{1}{a_d} \underset{d \to \infty}{\longrightarrow} 0$. Hence $(iii)$ is proved.\\

\proofstep{Step 3: proof of $(i) \iff (iii)$} 
First note that $(iii)$ trivially implies $(i)$ by considering the tests $\cT_d := \one_{L_d > a_d}$ which are one-sided by definition. The remaining implication to prove is $(i) \implies (iii)$. Assume that there exists tests $\cU_d$ for all $d \in \dN$ achieving one-sided detection, i.e. verifying
$$
\alpha_d := \dPl_d(\cU_d = 1) \underset{d \to \infty}{\longrightarrow} 0 \quad \mbox{and} \quad
\beta := \liminf_d \dPls_d(\cU_d = 1)>0 \, .
$$

Neyman-Person's Lemma gives that for all $d \in \dN$, there exists $b_d >0$ and $\eps_d >0$ such that the test 
$$
\cT_d := 
\begin{cases}
1 & \mbox{if } L_d>b_d \\
\xi_d & \mbox{if } L_d = b_d \\
0 & \mbox{if } L_d < b_d,
\end{cases}
$$ where the $\xi_d$ are independent Bernoulli variables of parameters $\eps_d$, verifies
\begin{equation}\label{eq:1implies3_1}
    \dPl_d(\cT_d = 1) = \alpha_d \quad \mbox{and} \quad \liminf_d \dPls_d(\cT_d = 1) \geq \beta \, .
\end{equation}

\begin{itemize}
    \item {if $b_d \underset{d \to \infty}{\longrightarrow} \infty$:} define $a_d := b_d-1$. The above implies in particular that $$\liminf_d \dPls_d(L_d > a_d) \geq \liminf_d\dPls_d(\cT_d = 1) \geq \beta, $$ and Markov's inequality gives $$ \dPl_d(L_d > a_d) \leq \frac{1}{b_d-1} \underset{d \to \infty}{\longrightarrow} 0 \, .$$ 
    Hence, $(iii)$ holds.
    
    \item {if $\limsup_d b_d < \infty$:} up to extraction we may assume that $b_d \leq M$. Portmanteau's theorem gives that 
    $$ \dPl_\infty(L_\infty >M) \leq \liminf_d \dPl_d(L_d >M) \leq \liminf_d \dPl_d(L_d >a_d) = 0 \, .$$ Hence the limit martingale $L_\infty$ has a compact support in $[0,M]$ under $\dPl_\infty$. By dominated convergence, one has
    $$ \dEl_d[(M-L_d)_+] \underset{d \to \infty}{\longrightarrow} \dEl_\infty[(M-L_\infty)_+] = M - \dEl_\infty[L_\infty] \, .  $$ Moreover, since $x \mapsto (M-x)_+$ is convex and $(L_d)_d$ is a martingale, the left hand side in the above is non-decreasing in $d$, and thus for all $d \in \dN$, 
    
    \begin{equation}\label{eq:1implies3_2}
    \dEl_d[(M-L_d)_+] \leq M - \dEl_\infty[L_\infty] \, . 
\end{equation}  
    
    Note that $(L_d)_d$ cannot be bounded $\dPl_d-$a.s. by this constant $M$, otherwise one would have $\dPls_d(\cT_d = 1) \leq M \cdot \dPl_d(\cT_d = 1) \underset{d \to \infty}{\longrightarrow} 0$ and \eqref{eq:1implies3_1} would fail. Hence there exists $d$ such that $\dPl_d(L_d>M)>0$. For such a $d$, we hence have that $\dEl_d[(M-L_d)_+] > \dEl_d[(M-L_d)]= M-1$, which together with \eqref{eq:1implies3_2} gives
    $$ M-1 < \dEl_d[(M-L_d)_+] \leq M - \dEl_\infty[L_\infty], $$ which in turn entails $$\dEl_\infty[L_\infty] < 1 \, .$$
    Hence $(iv)$ holds, and $(iii)$ also, by step 2.
\end{itemize}

\proofstep{Step 4: proof of $(v) \implies (i)$}
Note that $(v)$ straightaway implies $(i)$ by considering the tests $$\cT_d := \one_{L_d > \exp\left( a (\lambda s)^d\right)}$$ for $a>0$ small enough. This test has positive asymptotic power (at least $1-\pext(\lambda s)$ by assumption), and has vanishing type I error, since $\lambda s>1$, by Markov's inequality.  \\

\proofstep{Step 5: proof of $(i) \implies (ii)$}\footnote{This implication, namely that the $\KL-$divergence diverges when there exists a one-sided test, is a very general result which is not specific to our context. The last implication $(ii) \implies (v)$ however, strongly relies upon the structure of the problem, see Step 6 hereafter.}

Assume $(i)$, that is existence of a one-sided test. Then for every $d \geq 0$ there is an event $A_d \subset \cX_d^2$ such that $\dPl_d(A_d) \underset{d \to\infty}{\longrightarrow} 0$ and $\beta' := \liminf_{d \to\infty} \dPls_d(A_d) >0$. Elementary properties of the Kullback-Leibler divergence entail
\begin{align*}
\KL_d & \geq \dPls_d(A_d) \log\frac{\dPls_d(A_d)}{\dPl_d(A_d)}  + (1-\dPls_d(A_d)) \log\frac{1-\dPls_d(A_d)}{1-\dPl_d(A_d)} \\
& = - \dPls_d(A_d)\log \dPl_d(A_d)   + \dPls_d(A_d) \log \dPls_d(A_d) \\ & \quad \quad + (1-\dPls_d(A_d))\log(1-\dPls_d(A_d)) - \underbrace{(1-\dPls_d(A_d))\log(1-\dPl_d(A_d))}_{\leq 0} \\
& \geq - \dPls_d(A_d)\log \dPl_d(A_d) + g(\dPls_d(A_d)) \, ,
\end{align*} 
where for $x \in [0,1]$, $g$ is defined by $g(x) :=  x \log(x) + (1-x) \log(1-x)$. Function $g$ is minimal at $x=1/2$ and $g(1/2)=-\log(2)$, which gives the final bound $$\liminf_{d \to\infty}  \KL_d \geq \beta' \liminf_{d \to\infty}(  - \log \dPl_d(A_d) )- \log 2 = + \infty \, .$$ Hence, $(ii)$ holds.\\

\proofstep{Step 6: proof of $(ii) \implies (v)$}
Assume $(ii)$, that is $\KL_d = \dEls_{\infty} \left[\log L_d\right] \underset{d \to \infty}{\longrightarrow} +\infty$. Recall that $\lambda s>1$.  Under model $\dPls_{\infty}$, recall that $\tau^*$ denotes the intersection tree. Let us define $$W:=\lim_{d\to\infty} \card{\cL_d(\tau^*)} (\lambda s)^{-d} \, .$$ The random variable $W$ is defined as an almost sure limit, which exists from general branching process theory, $(\card{\cL_d(\tau^*)} (\lambda s)^{-d})_{d \geq 0}$ begin a non-negative martingale. 

Under $\dPls_{\infty}$, on the event $\cA_\infty$ that $\tau^*$ survives, which has strictly positive probability for $\lambda s>1$, it holds that $W>0$. In addition, consider $\sigma_*$ (resp. $\sigma'_*$) be the natural injection from $\tau^*$ to $T$ (resp. to $T'$) -- such natural injections have to exist in the correlated model.

Let $d,m$ be two integers. In view of the explicit formula \eqref{eq:lemma:LR_developed} of Lemma \ref{lemma:LR_developed}, we have the lower bound on event $\cA_\infty$:
\begin{flalign*}
L_{d+m}(T,T')& \geq \prod_{u \in \cV_{d-1}(\tau^*)}\Psi(c_{\tau^*}(u),c_{T}(\sigma_*(u)),c_{T'}(\sigma'_*(u))\prod_{v \in \cL_d (\tau^*)}L_m(T_{\sigma_*(v)},T'_{\sigma'_*(v)})
\\
&\geq \prod_{u \in \cV_{d-1}(\tau^*)}\Psi(c_{\tau^*}(u),c_{T}(\sigma_*(u)),c_{T'}(\sigma'_*(u)) \times e^{\card{\cL_d(\tau^*)}[\dEls_{\infty} \left[\log L_m\right]-o_d(1)]} ,
\end{flalign*} where the second line is obtained by the law of large numbers, which can be applied since $\cL_d(\tau^*) \underset{d \to \infty}{\longrightarrow} + \infty$ on event $\cA_\infty$.

For $d$ large, here again by the law of large numbers, the first product is with high probability lower-bounded by $e^{C W (\lambda s)^d}$ for some fixed constant $C$ which may be negative. Choosing $m$ of order 1 but sufficiently large, since by assumption $\dEls_{\infty} \left[\log L_m\right] \underset{m \to \infty}{\longrightarrow} +\infty$, we can ensure that the second factor is larger than $e^{C' W (\lambda s)^d}$ for some arbitrary $C'$. Taking $C'>0$ large enough (namely $C'>-C$) ensures that, on the event $\cA_\infty$ that $\tau^*$ survives, $(\lambda s)^{-d} \log L_d$ is lower-bounded $C'+C>0$, which implies $(v)$, since $\dPls_\infty(\cA_\infty) = 1-\pext(\lambda s) $.\\

This last step concludes the proof of Theorem \ref{theorem:main_result_TREES}.
\end{proof}
We end this section taking a detour, introducing a Markov transition semi-group on trees that arises naturally in our study.
\subsection{A Markov transition kernel on trees}
We can take a dynamic view on the correlated model: the joint distribution of the pair of trees $(t,t')$ under $\dPls_{\infty}$ will be, up to relabeling, interpreted as the joint distribution of $(X_0,X_T)$, where $X_0$ is the initial state of a Markov process, distributed according to its stationary distribution $\GWl_{d}$, and $X_T$ is its state at time $T$. The time parameter $T$ is in one-to-one correspondence with the correlation parameter $s$ of our model, namely
\begin{equation*}
    T=-\log(s) \, .
\end{equation*}

For $d \geq 0$, we define $\mathfrak{M}_d$ the linear operator indexed on trees of $\cX_{d}$, defined as follows:
\begin{equation}\label{eq:def_M}
	\mathfrak{M}_d(t,t') := \frac{\dPls_{d}(t,t')}{\dPl_{d}(t)} \, .
\end{equation}
$\mathfrak{M}_d$ is identified to the \emph{transition kernel} of the Markov chain with transitions denoted by $$t \overset{\lambda,s}{\longrightarrow} t',$$  where $t'$ is obtained from $t$ from the following two-step procedure:
\begin{itemize}
\item[$1.$] Extract $\tau$, a $s-$subsampling of $t$;
\item[$2.$] Draw a $(\lambda,s)-$augmentation of $\tau$ at depth $d$, that is $t' \sim \Augls_d(\tau)$. 
\end{itemize}
We next denote $\mathfrak{M}_d(s) := \mathfrak{M}_d$  to emphasize its dependence on $s$. See figure \ref{fig:transition_markov} for an illustration.

\begin{figure}[H]
    \centering
    \includegraphics[scale=0.42]{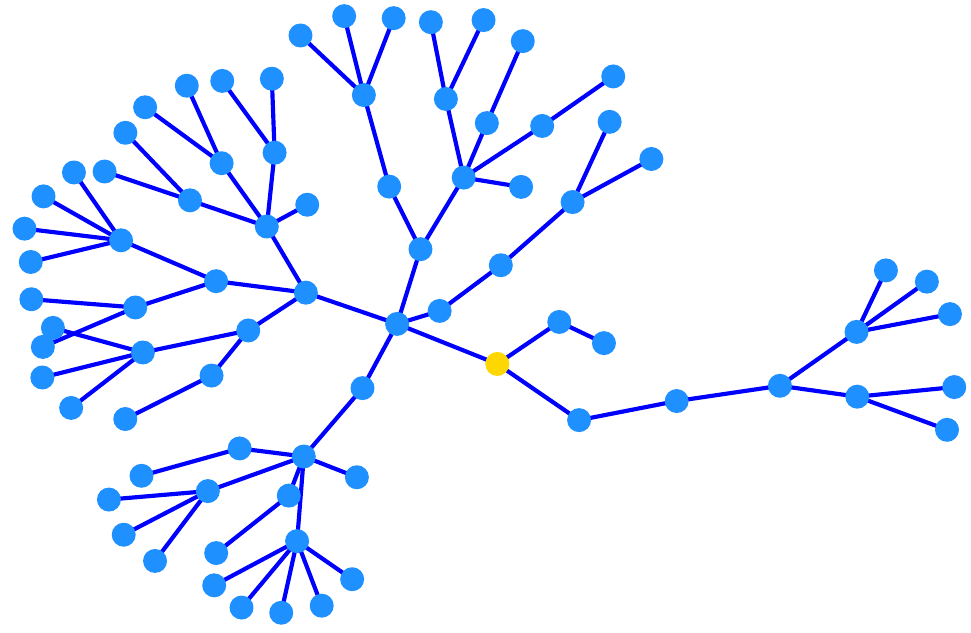}
    \hspace{0.06cm}
    \includegraphics[scale=0.42]{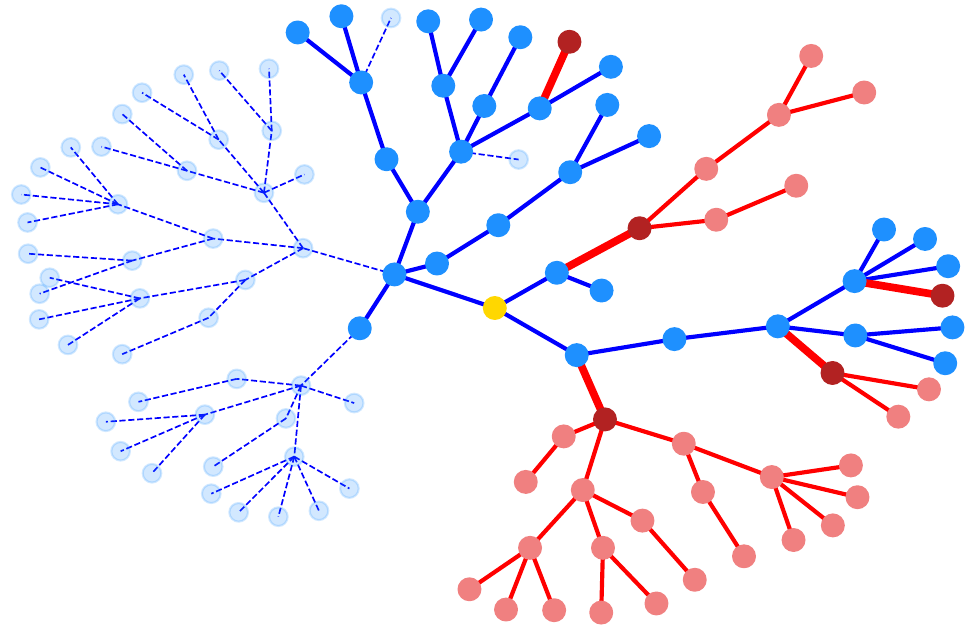}
    \caption{Example of a transition described hereabove, with $\lambda = 1.85$, $s=0.85$, at depth $d=5$. The original tree $t$ is drawn on the left. On the right, $t'$ is obtained as follows: first extracting a $s-$subsampling $\tau$ of $t$ (dashed blue edges are deleted), and drawing a $(\lambda,s)-$augmentation of $\tau$ -- first attaching new children to all vertices of $\tau$ (dark red nodes with thick edges), and attaching new Galton-Watson trees to these new children (light red nodes with standard edges). Labels are not shown.}
    \label{fig:transition_markov}
\end{figure}

A remarkable property of this kernel is the following semi-group structure:

\begin{proposition}[Consistency of kernels $\mathfrak{M}_d(s)$]\label{prop:consistency}
Let $\lambda >0$ and $s, s' \in [0,1]$. Then, for all $n \geq 1$,
\begin{equation}\label{eq:prop:consistency}
	\mathfrak{M}_d(s) \mathfrak{M}_d(s') = \mathfrak{M}_d(s') \mathfrak{M}_d(s) = \mathfrak{M}_d(ss').
\end{equation}
\end{proposition}
\begin{proof} 
The proof consists in verifying that applying transitions $\mathfrak{M}_d(s)$ and $\mathfrak{M}_d(s')$ successively is equivalent in distribution to applying transition $\mathfrak{M}_d(s s')$. 
To do so, we use a coupling argument. For $t \in \cX_d$, let us apply two transitions $t \overset{\lambda, s}{\longrightarrow} \widetilde{t} \overset{\lambda, s'}{\longrightarrow} t'$ in a coupled fashion. Since the relabelings are all uniform, we can check this consistency forgetting the labels.

\proofstep{Step 1: first transition} 
First extract $\widetilde{\tau}$, a $s-$subsampling of $t$. To each vertex $u$ of $\widetilde{\tau}$ we attach an independent number $\Poi(\lambda (1-s))$ of new children. The set of these new vertices is denoted by $\widetilde{V}^+$. Then, to each vertex $u \in \widetilde{V}^+$ we attach an independent tree $\widetilde{t}_u$ with distribution $\GW_\lambda$. This resulting tree $\widetilde{t}$ is indeed obtained by a transition $t \overset{\lambda, s}{\longrightarrow} \widetilde{t}$.\\

\proofstep{Step 2: second transition}
For the second transition, we sample $t$ from $\widetilde{t}$ as follows:
\begin{itemize}
    \item[$1.$] First, we extract $\tau$ as a $s'-$subsampling of $\widetilde{\tau}$;
    \item[$2.$] To any vertex $u$ of $\tau$, we keep each previous child $v$ of $u$ in tree $\widetilde{t}$ that lied in $\widetilde{V}_+$ independently with probability $s'$, the set of children that are kept is denoted by $V^+_1$;
    \item[$3.$] To any vertex $u$ of $\tau$, we also attach an independent number $\Poi(\lambda (1-s'))$ of new children. The set of these new vertices are referred to as $V^+_2$;
    \item[$4.$] To any vertex $v \in V^+_1$, we attach a tree $t'_v$ to node $v$, where $t'_v$ is obtained from the previous $\widetilde{t}_v$ by a transition $\widetilde{t}_v \overset{\lambda, s'}{\longrightarrow} t'_v$;
    \item[$5.$] To each vertex $w \in V^+_2$ we attach an independent tree $t'_w$ with distribution $\GW_\lambda$.
\end{itemize}

Two simple properties denoted by $(\rPone)$ and $(\rPtwo)$ are key to our result. First, a $s'-$ subsampling of a $s-$ subsampling is distributed as a $(ss')-$ subsampling $(\rPone)$. Second, by elemetary properties of the Poisson distribution, a $s-$subsampling of a $\GWl_d$ tree has distribution $\GWls_d$ $(\rPtwo)$.

Property $(\rPtwo)$ implies that the tree $t'$ obtained after the above five steps can be indeed sampled from a transition $\widetilde{t} \overset{\lambda, s}{\longrightarrow} t'$.

On the other hand, $t'$ has also been obtained from $t$ by the following process: from the initial tree $t$, by step 1. hereabove and $(\rPone)$ we extracted $\tau$ as a $ss'-$subsampling of $t$, and we attached to each vertex of $\tau$ some new children: the sum of two independent $\Poi(\lambda(1-s)s')$ (step 2., children in $V^+_1$) and $\Poi(\lambda(1-s'))$ (step 3. children in $V^+_1$), hence again of Poisson distribution with parameter $\lambda(1-s)s' + \lambda(1-s') = \lambda(1-ss')$. By steps 4. and 5. and $\rPtwo$ the trees attached to every vertex in $V^+ := V^+_1 \cup V^+_2$ are i.i.d. with distribution $\GW_\lambda$, independent of $t$. Hence, $t'$ can also have been obtained from $t$ from a transition $t \overset{\lambda, ss'}{\longrightarrow} t'$.

This proves the desired consistency property \eqref{eq:prop:consistency}. 
\end{proof}

\section{Sufficient conditions based on Kullback-Leibler divergence}\label{section:KL}
In view of Theorem \ref{theorem:main_result_TREES}, we are interested in in finding sufficient conditions on $\lambda$ and $s$ for $\KL_d$ to diverge for large $d$. This is the object of the following short section. We refer to \eqref{eq:def_KL} for the definition of $\KL_d$. 

We first start by the following easy Lemma
\begin{lemma}\label{lemma:KL_strict_croissante}
    The sequence $\KL_d$ is strictly increasing with $d$. In particular $\KL_d$ has a (possibly infinite) limit when $d \to \infty$.
\end{lemma}
\begin{proof}
\begin{flalign*}
    \KL_{d+1} &= \dEls_{\infty}[\log (L_{d+1})]
    = \dEl_{\infty}[L_{d+1}\log (L_{d+1})]= \dEl_{\infty}[\dE[L_{d+1}\log (L_{d+1}) | \cF_d]] \\ & \overset{(a)}{>} \dEl_{\infty}[\dE[L_{d+1} | \cF_d] \log (\dE[L_{d+1} | \cF_d]) ]  \overset{(b)}{=} \dEl_{\infty}[L_{d} \log (L_{d}) ] = \KL_d \, .
\end{flalign*} We applied Jensen's inequality in $(a)$: the inequality is strict since $x \to x\log(x)$ is strictly convex, and since $L_{d+1}=L_d$ does not holds almost surely. In $(b)$, we used Proposition \ref{prop:LR_martingale}, i.e. that $(L_d)_{d \geq 0}$ is a $\cF_d-$martingale under $\dPl_{\infty}$.
\end{proof}

\begin{remark}[Corroborating impossibility results]
Although the following remark will not be used directly in the sequel, it remains of interest in our study. Recall that the entropy $\Ent(\p)$ of a probability distribution $\p: \cX \to [0,1]$ is defined by 
    $$\Ent(\p) := - \dE_{X \sim \p}[\p(X)] \, .$$
    
    We can easily show that $\KL_d$  satisfies
    $$\KL_d \leq \Ent (\GW_{\lambda s,d}) \, .$$

This is a consequence of data processing inequality. Note that models $\dPls_d$ and $\dPls_d$ can be described by the channel represented in Figure \ref{fig:channel}.

\begin{figure}[H]
\centering
\begin{tikzpicture}[node distance=1cm]  
\draw (0,0)  node (tautau) {$(\tau,\tau')=(\tau,\tau)$ for $\tau^*\sim \GWls_d$};
\draw (0,-2) node (tautaup) {$(\tau,\tau') \sim \GWls_d \otimes \GWls_d$};
\draw (4.5,-1) node[rectangle,draw,draw opacity=0.5] (center) {$\Augls_d(\tau) \otimes \Augls_d (\tau')$};
\draw (9,0)  node (Pls) {$(T,T')\sim \dPls_d$};
\draw (9,-2) node (Pl) {$(T,T')\sim \dPl_d$};
\draw[->,thick,red] (tautau) -- (center);
\draw[->,thick,blue] (tautaup) -- (center);
\draw[->,thick,red] (center) -- (Pls);
\draw[->,thick,blue] (center) -- (Pl);
\end{tikzpicture} 
\medskip
\caption{Channel structure for $\dPls_d$ and $\dPls_d$.} 
\label{fig:channel}
\end{figure}
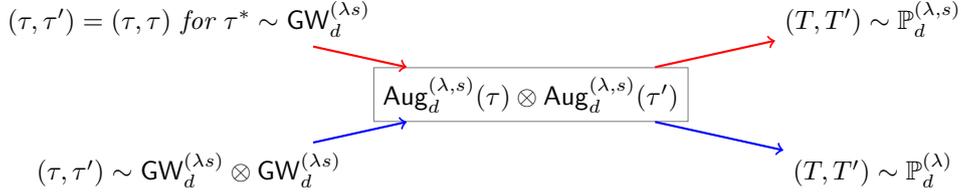

Since $\dPls_{d}$ and $\dPl_d$ are obtained applying the same kernel to different input distributions, data processing inequality states that the $\KL-$divergence between the outputs (i.e. $\KL_d$) is less than the $\KL-$divergence between the inputs, which reads
\begin{equation*}
    \sum_{\tau\in\cX_d}\GW_{\lambda s,d}(\tau)\log\left( \frac{\GW_{\lambda s,d}(\tau)}{\GW_{\lambda s,d}(\tau)^2}\right)=\Ent(\GW_{\lambda s,d}).
\end{equation*}

Moreover, note that entropy $\Ent(\GW_{\lambda s,d})$ can be evaluated by the conditional entropy formula as
\begin{equation*}
    \Ent(\GW_{\lambda s,d})=\Ent(\GW_{\lambda s,d-1})+(\lambda s)^{d-1}\Ent(\p_{\lambda s}), 
\end{equation*} which implies that whenever $\lambda s <1$, $\Ent(\GW_{\lambda s,d})$ is uniformly bounded in $d$ and thus so is $\KL_d$ by the above Fact:
\begin{equation*}
    \lim_{d\to\infty} \KL_d \leq \frac{1}{1-\lambda s}\Ent(\p_{\lambda s})<+\infty.
\end{equation*} This again corroborates the results of \cite{ganassali2021impossibility}. In fact, we will next do better, providing a better upper bound for $\KL_d$ giving a looser sufficient condition for $\KL_d$ to remain bounded, see Section \ref{section:hard_phase}.
\end{remark}

We are now ready to give some insights for condition $(ii)$ of Theorem \ref{theorem:main_result_TREES} to hold. First, let us establish the following Lemma:
\begin{lemma}\label{lemma:geo_rec_KL}
For all $d \geq 1$, one has
\begin{equation}\label{eq:lemma:geo_rec_KL}
\KL_{d}\geq \lambda s \KL_{d-1} +\lambda s \left(\log(s/\lambda)  +1\right) +2\lambda (1-s)\log(1-s) \, .
\end{equation}
\end{lemma}
\begin{proof}
Under $\dPls_{\infty}$, let $c$ be the degree of the root in $\tau^*$ and $c+\Delta$ (resp. $c+\Delta')$ the degree of the root in $T$ (resp. in $T'$). In the recursive formula \eqref{eq:lemma:LR_rec} of Lemma \ref{lemma:LR_rec} for $L_d$, setting $k=c$ in the first summation and considering only the $c!$ pairs $(\sigma,\sigma')$ that correctly make the $c$ children of $\tau^*$'s root in $T$ and $T'$ correspond, we obtain the following lower-bound:
\begin{flalign*}
L_d(T,T') &\geq \Psi(c,c+\Delta,c+\Delta') \times c!\times \prod_{u=1}^{c} L_{d-1}(T_u,T'_u) \\
& = e^{\lambda s} (s/\lambda)^{c}(1-s)^{\Delta+\Delta'}  \prod_{u=1}^{c} L_{d-1}(T_u,T'_u) \, .
\end{flalign*}
Taking logarithms and then expectations, we get 
\begin{flalign*}
\KL_d = \dEls_\infty[\log(L_d(T,T'))] & \geq  \lambda s +  \dEls_\infty[c] \cdot  \log(s/\lambda) + \dEls_\infty[\Delta+\Delta'] \cdot \log(1-s) \\
& \quad \quad \quad + \dEls_\infty[c] \cdot \dEls_\infty[\log(L_{d-1}(T,T'))] \, .
\end{flalign*}
since $\dEls_{\infty} \left[c\right]=\lambda s$ and $\dEls_{\infty} \left[ \Delta \right]=\dEls_{\infty} \left[\Delta'\right]=\lambda(1-s)$, the result follows. 
\end{proof}

The recursive inequality \eqref{eq:lemma:geo_rec_KL} of Lemma \ref{lemma:geo_rec_KL} can sometimes be enough to show the divergence of $\KL_d$, provided favorable initial conditions. These initial conditions are introduced in the following
\begin{corollary}\label{cor:suff_cond_KL}
Assume that $\lambda s>1$ and that
\begin{equation}\label{eq:cor:suff_cond_KL}
\KL_1 \geq \frac{1}{\lambda s-1} \left[\lambda s(\log(\lambda/s)-1)-2\lambda(1-s)\log(1-s)\right] \, .
\end{equation}
Then $\KL_d \underset{d \to \infty}{\longrightarrow} +\infty$.
\end{corollary}

\begin{proof}
The proof follows from the recursive inequality \eqref{eq:lemma:geo_rec_KL} of Lemma \ref{lemma:geo_rec_KL}. Indeed, let us assume \eqref{eq:cor:suff_cond_KL}. Then, for all $d \geq 1$, 
\begin{flalign*}
    \KL_{d+1} - \lambda s \KL_{d} & \geq \lambda s \left(\log(s/\lambda)  +1\right) +2\lambda (1-s)\log(1-s) \\
    & \overset{(a)}{\geq} -(\lambda s - 1) \KL_1 ,
\end{flalign*} where we used inequality \eqref{eq:cor:suff_cond_KL} and condition $\lambda s > 1$ in $(a)$. The latter implies that for all $d \geq 1$, $\KL_{d+1}-\KL_1 \geq \lambda s(\KL_{d}-\KL_1)$,
hence $\KL_{d}$ diverges geometrically to infinity, since we have $\KL_2>\KL_1$ by Lemma \ref{lemma:KL_strict_croissante}.
\end{proof}

Sufficient conditions on $\lambda$ and $s$ can be derived when enforcing the right-hand side of \ref{eq:cor:suff_cond_KL} to non-positive, that is studying the sign of
\begin{equation*}
    f_\lambda(s) := s(\log(\lambda)- \log(s)-1)-2(1-s)\log(1-s) \, ,
\end{equation*} with the continuous extensions $f_\lambda(0)=0$ and $f_\lambda(1)=\log(\lambda)-1$.

Studying the derivative of $f'_\lambda = \log(\lambda) - \log(s) + 2\log(1-s)$, one can easily show that $f_\lambda$ is strictly increasing from $s=0$ to $s=s_0(\lambda) \in (0,1)$ where $s_0(\lambda)$ is given by 
$$s_0(\lambda):= 1 - \frac{\sqrt{4\lambda + 1}-1}{2\lambda}, $$ 
and then strictly decreases from $s=s_0(\lambda)$ to $s=1$. Since $f_\lambda(0)=0$, this shows that there is at most a single root of $f_\lambda$ in the open interval $(0,1)$. 

Since $f_\lambda(1)=\log(\lambda)-1$, there is a root of $f_\lambda$ in $(0,1)$ is and only if $\lambda < e$. In this case we denote by $s^*(\lambda)$ this root. We also remark that 
$$f_\lambda(1/\lambda)= (2 \log(\lambda)-1)/\lambda - 2(1-1/\lambda)\log(1-1/\lambda) ,   $$
and studying the above expression, there exists a value $\lambda^*>1$ such that $f_\lambda(1/\lambda) \leq 0$ for all $\lambda \in (1,\lambda^*]$. Numerically, $\lambda^* \sim 1.17789383...$ This implies that for $\lambda \in (1,\lambda^*]$, the previously defined root $s^*(\lambda)$ satisfies $s^*(\lambda) \leq 1/\lambda$, and the condition $\lambda s>1$ of Theorem \ref{theorem:main_result_TREES} becomes the only one to be sufficient for existence of one-sided tests. 

The above results, together with Theorem \ref{theorem:main_result_TREES}, prove the main Theorem of this section. One-sided detectability always refers for the tree correlation detection problem.

\begin{theorem}\label{theorem:suff_cond_KL}
We have the following:
\begin{itemize}
    \item if $1 < \lambda \leq \lambda^*$, where $\lambda^* \sim 1.178...$ (see above), then one-sided detectability holds as soon as $\lambda s >1$.

    \item $\lambda^* < \lambda < e$, then one-sided detectability holds as soon as $s^*(\lambda) \leq s \leq 1$, where $s^*(\lambda)$ is defined by 
    \begin{equation*}
        s^*(\lambda):=\sup\{s\in[0,1]: s(\log(\lambda)- \log(s)-1)-2(1-s)\log(1-s) \geq 0\}
    \end{equation*} and verifies $1/\lambda < s^*(\lambda)<1$.
\end{itemize}
\end{theorem}

\begin{remark}
The result stated in the first point of Theorem \ref{theorem:suff_cond_KL} is similar to those obtained in \cite{Ganassali20a}, however the present derivation is more direct and explicit. Note that anticipating Section \ref{section:graph_matching}, this result states that for small values of $\lambda$, there is no hard phase for sparse graph alignment, and one gets immediately from the impossible phase to the easy phase when crossing the border $\lambda s = 1$. 
\end{remark}

\begin{remark}
    Other sufficient conditions based on asymptotic expansions of $\KL_1$, this time for large $\lambda$, could also be obtained from condition \eqref{eq:cor:suff_cond_KL} of Corollary \ref{cor:suff_cond_KL}. However, the resulting conditions do not appear as sharp as those obtained by the analysis of automorphisms of $\tau^*$, which is the object of the next section.
\end{remark}

\section{Sufficient conditions based on counting automorphisms of Galton-Watson trees}\label{section:autos_GW}
In this Section, we show how counting automorphisms of Galton-Watson trees gives a sufficient condition for the existence of one-sided tests at large $\lambda$, and provide along the way evaluations of this number of automorphisms. 

\subsection{A lower bound on the likelihood ratio}
We work under model $\dPls_{d}$. As in step 6 of the proof of Theorem \ref{theorem:main_result_TREES}, we will consider $\sigma_*$ (resp. $\sigma'_*$) be the natural injection from $\tau^*$ to $T$ (resp. to $T'$) -- such natural injections have to exist in the correlated model. We denote, for each $u \in \cV_{d-1}(\tau^*)$:
\begin{equation}
c_u:=c_{\tau^*}(u),\; \Delta_u:=c_{T}(\sigma_*(u))-c_{\tau^*}(u),\; \Delta'_u:=c_{T'}(\sigma'_*(u))-c_{\tau^*}(u).
\end{equation} We now prove the following
\begin{lemma}\label{lemma:lower_bound_LR}
Under $\dPls_{d}$, we have the following inequality:
\begin{equation}\label{eq:lemma:lower_bound_LR}
L_d(T,T') \geq \card{\Aut(\tau^*)} \prod_{u \in \cV_{d-1}(\tau^*)}e^{\lambda s} (s/\lambda)^{ c_u}(1-s)^{\Delta_u+\Delta'_u} \prod_{v \in \cL_{d-1}(\tau^*)}\binom{c_v+\Delta_v}{c_v}\binom{c_v+\Delta'_v}{c_v},
\end{equation}
where we recall that $\Aut(\tau^*)$ denotes the set of automorphisms of tree $\tau^*$.
\end{lemma}

\begin{proof}
With the developed expression \eqref{eq:lemma:LR_developed} of $L_d$ in Lemma \ref{lemma:LR_developed}, we are going to sum only over subtrees $\tau$ that are equal to $\tau^*$, up to some relabeling. We denote this equivalence by $\tau \equiv \tau^*$. For each of these $\tau$, in the second sum (over injections) we are only going to consider pairs $\sigma \in \widetilde{\cS}(\tau,t),\sigma' \in \widetilde{\cS}(\tau,t')$ where $\widetilde{\cS}(\tau,t)$ (resp. $\widetilde{\cS}(\tau,t')$) is the set of injections in $\cS(\tau,t)$ (resp. in $\cS(\tau,t')$) that coincide with $\sigma_*$ and $\sigma'_*$, almost everywhere, but on final leaves at depth $d$. 

We thus have the following lower bound

\begin{flalign}\label{eq:lower_bound_LR_1}
    L_d(T,T') & \geq \sum_{\substack{\tau \in \cX_d \\ \tau \equiv \tau^*}} \sum_{ \substack{\sigma \in \widetilde{\cS}(\tau,t) \\ \sigma' \in \widetilde{\cS}(\tau,t')}} \prod_{u \in \cV_{d-1}(\tau)}\Psi(c_u,c_{u}+\Delta_{u},c_{u}+\Delta'_{u}) \nonumber \\
    & = \left(\sum_{\substack{\tau \in \cX_d \\ \tau \equiv \tau^*}} |\widetilde{\cS}(\tau,t)| \cdot |\widetilde{\cS}(\tau,t')|\right) \cdot \left( \prod_{u \in \cV_{d-1}(\tau^*)}  \frac{1}{c_u !} e^{\lambda s} (s/\lambda)^{ c_u}(1-s)^{\Delta_u+\Delta'_u}  \right) \, .
\end{flalign}
Note that any tree $\tau\in \cX_d$ such that $\tau\equiv \tau^*$ can be described by a relabeling of the form
\begin{equation*}
    \left \lbrace \xi_u \in \cS_{c_u}, u \in \cV_{d-1}(\tau^*)\right \rbrace,
\end{equation*} giving the reordering of the children of each node of $\tau^*$ at depth $d-1$.
Moreover, the number of relabelings $r$ that produce this particular tree $\tau$ is precisely given by $\card{\Aut(\tau^*)}$. Thus the number of trees in the summation \eqref{eq:lower_bound_LR_1} is precisely
\begin{equation}\label{eq:toto}
\card{\{\tau\in \cX_d: \tau \equiv \tau^*\}}= \frac{1}{\card{\Aut(\tau^*)}} \prod_{u \in \cV_{d-1}(\tau^*)}c_u! \, .
\end{equation}

Now, note that for any tree $\tau \equiv \tau^*$, we have 
\begin{equation}\label{eq:tildeS}
|\widetilde{\cS}(\tau,t)| = \card{\Aut(\tau^*)} \prod_{v \in\cL_{d-1}(\tau^*)}\binom{c_v+\Delta_v}{c_v} \quad \mbox{and} \quad 
|\widetilde{\cS}(\tau,t')| = \card{\Aut(\tau^*)} \prod_{v \in\cL_{d-1}(\tau^*)}\binom{c_v+\Delta'_v}{c_v} \, .
\end{equation} 
Indeed, factors $\binom{c_v+\Delta_v}{c_v}$ (resp.  $\binom{c_v+\Delta'_v}{c_v}$) denotes the number of subsets of the $c_v+\Delta_v$ children of $v$ at depth $d$ in $t$ (respectively, of the $c_v+\Delta'_v$ children of $v$ at depth $d$ in $t'$) that we can associate as children of $v$ in the injection $\sigma \in \widetilde{\cS}(\tau,t)$ (respectively, $\sigma' \in \widetilde{\cS}(\tau,t')$), the order in which they are considered being determined by the permutation $\xi_v$ in the relabeling characterizing $\tau \equiv \tau^*$.

The proof is concluded by combining \eqref{eq:lower_bound_LR_1}, \eqref{eq:toto} and \eqref{eq:tildeS}.
\end{proof}

\subsection{Another sufficient condition for one-sided tests}
We now state a result which establishes a lower-bound on the number $|\Aut(\tau^*)|$ of automorphisms for $\tau^*\sim \GWls_{d}$.
\begin{proposition}\label{proposition:auto_GW}
Let $r$ be a sufficiently large constant (in particular, $r>1$). For $\tau^*\sim \GW^{(r)}_{d}$, let us denote by $W$ the almost sure limit:
\begin{equation}\label{eq:mgle_limit}
W := \lim_{d \to\infty} r^{-d}\card{\cL_d(\tau^*)} \, .
\end{equation}
Let $\cA_\infty$ be the event on which $\tau^*$ survives, which occurs with probability $1-\pext(r)>0$, and on which $W>0$. 
We then have with high probability the lower bound
\begin{equation}\label{eq:prop_lower_bound_aut}
\log\left(\frac{\card{\Aut(\tau^*)}}{\prod_{u \in \cV_{d-1}(\tau^*)}e^{-r}r^{c_{\tau^*}(i)}}\right)\one_{\cA_\infty} \geq (1-o_{\dP}(1)) \frac{W r^{d}}{r -1} \left[\frac{\log^{3/2} r}{3 \sqrt{r}}+O_r\left(\frac{\log^{5/4} r}{\sqrt{r}}\right)\right]\one_{\cA_\infty} \, .
\end{equation}
\end{proposition}

Proposition \ref{proposition:auto_GW}, proved in Appendix \ref{appendix:proof:proposition:auto_GW}, could be of independent interest. We believe that a little more work could easily show that inequality \eqref{eq:prop_lower_bound_aut} is exponentially tight, i.e. gives the right exponential order for the estimation of the number of automorphism of a Galton-Watson tree. 
We next show that Lemma \ref{lemma:lower_bound_LR} together with Proposition \ref{proposition:auto_GW} yield a sufficient condition for the existence of  one-sided test.

We are now in a position to prove the central result of this section.

\begin{theorem}\label{theorem:suff_cond_auto}
There exists a constant $r_0$ such that if 
\begin{equation}\label{eq:theorem:suff_cond_auto}
\lambda s > r_0 \quad \mbox{and} \quad 1-s \leq \frac{1}{(3+\eta)}\sqrt{\frac{\log(\lambda s)}{\lambda ^3 s}},
\end{equation} for some $\eta>0$, then one-sided detectability holds.
\end{theorem}

\begin{proof} 
We recall that $\cA_\infty$ denotes the event on which $\tau^*$ survives, and we define
\begin{equation*}
W := \lim_{d \to\infty} (\lambda s)^{-d}\card{\cL_d(\tau^*)} \, .
\end{equation*}The proof consists in showing that under assumptions \eqref{eq:theorem:suff_cond_auto}, condition $(iii)$ of Theorem \ref{theorem:main_result_TREES} is satisfied.
In the lower bound \eqref{eq:lemma:lower_bound_LR} of Lemma \ref{lemma:lower_bound_LR}, consider the factor
\begin{equation*}
     X := \prod_{u \in \cV_{d-1}(\tau^*)}e^{\lambda s} (s/\lambda)^{ c_u}(1-s)^{\Delta_u+\Delta'_u} \prod_{v \in \cL_{d-1}(\tau^*)}\binom{c_v+\Delta_v}{c_v}\binom{c_v+\Delta'_v}{c_v} \, .
 \end{equation*} Note that $X$ is a random variable.
On event $\cA_\infty$, we can appeal to the law of large numbers to compute a high probability equivalent of $\log(X)$, since
\begin{flalign}\label{eq:equiv_A1}
    A & := \log\left(\prod_{u \in \cV_{d-1}} \frac{s^{c_i} (1-s)^{\Delta_i+\Delta'_i}}{e^{-\lambda s}\lambda ^{c_i}} \right) \sim \frac{W (\lambda s)^{d}}{\lambda s -1} [ \lambda s (\log(s/\lambda)+1) + 2\lambda (1-s) \log (1-s)]
\end{flalign} and 
\begin{flalign}\label{eq:equiv_B1}
     B & := \log\left(\prod_{i\in\cL_{d-1}(\tau^*)}\binom{c_i+\Delta_i}{c_i}\binom{c_i+\Delta'_i}{c_i}\right) \nonumber \\ 
     & \sim 2 W (\lambda s)^{d-1} \left( \dE[\log(\Poi(\lambda)!)]-\dE[\log(\Poi(\lambda(1-s))!)]-\dE[\log(\Poi(\lambda s)!)]\right).
\end{flalign}
Let us introduce the notations $r := \lambda s$, $\alpha := \lambda (1-s)$, such that $\lambda = \alpha + r$ and $s = \frac{r}{\alpha + r}$. We will identify equivalents of exponents of interest as  $\alpha \to 0$ and $r \to \infty$. In this regime, \eqref{eq:equiv_A1} becomes
\begin{flalign}\label{eq:equiv_A2}
    A  & \sim \frac{W r^{d}}{r -1} \left(-2 r\log(1+\alpha/r)+2\alpha\log\left(\frac{\alpha/r}{1+\alpha/r}\right)-r\log r+r\right) \nonumber \\
    & \sim \frac{W r^{d}}{r -1}  \left(-r \log r + r - 2 \alpha \log r + 2\alpha\log\alpha  + O(\alpha) \right) \, .
\end{flalign} We have the classical estimate for large $\mu$:
\begin{equation}\label{eq:moment_factoriel}
\dE [\log(\Poi(\mu)!) ]=\mu\log(\mu)-\mu+\frac{1}{2}\log(2\pi e \mu)+O\left(\frac{1}{\mu}\right),
\end{equation}
Using \eqref{eq:moment_factoriel} and noting that in this regime, $\dE [\log(\Poi(\alpha)!) ] = O(\alpha^2)$, \eqref{eq:equiv_B1} becomes
\begin{flalign}\label{eq:equiv_B2}
B & \sim 2 W r^{d-1}\left( \alpha\log(r)+O(\alpha)\right) \, .
\end{flalign} Combined, \eqref{eq:equiv_A2} and \eqref{eq:equiv_B2} give that with high probability, under $\cA_\infty$, 
\begin{flalign}\label{eq:equiv_A+B}
\log(X) = A+B & \sim \frac{W r^{d}}{r -1} \left(\left(1-\frac{1}{r}\right)\times 2 \alpha\log(r) -r\log r + r - 2\alpha\log(r)+2\alpha\log(\alpha)+O(\alpha)\right) \nonumber\\
&\sim \frac{W r^{d}}{r -1}\left(-r\log r + r + 2\alpha\log(\alpha)+O(\alpha)\right) \, .
\end{flalign}

Combining \eqref{eq:equiv_A+B} with the results of Proposition \ref{proposition:auto_GW} entails that with high probability, under $\cA_\infty$, 
\begin{flalign*}
\log L_d & \geq \frac{W r^{d}}{r -1} \left[r\log r - r +  \frac{\log^{3/2}(r)}{3\sqrt{r}}+O\left(\frac{\log^{5/4} r}{\sqrt{r}}\right) \right] + \frac{W r^{d}}{r -1} \left[-r\log r + r + 2\alpha\log(\alpha)+O(\alpha)\right]\\
& = \frac{W r^{d}}{r -1} \left[2\alpha \log \alpha +  \frac{\log^{3/2}(r)}{3\sqrt{r}}+O\left(\frac{\log^{5/4} r}{\sqrt{r}}\right)  + O(\alpha) \right].
\end{flalign*} 
Then, under assumption \eqref{eq:theorem:suff_cond_auto}, we have $\alpha\leq \frac{1}{3+\eta}\sqrt{\log(r)/r}$ so that, for sufficiently large $r$,
\begin{equation*}
    2\alpha \log \alpha +  \frac{\log^{3/2}(r)}{3\sqrt{r}} > \Omega\left(\frac{\log^{3/2}(r)}{\sqrt{r}}\right).
\end{equation*}
It follows that condition $(iii)$ of Theorem \ref{theorem:main_result_TREES} holds with e.g. $a_d=d$. Hence, one-sided detectability holds. 
\end{proof}

\section{A sufficient condition for impossibility of correlation detection: conjectured hard phase for partial graph alignment}
\label{section:hard_phase}
In the present section we establish that, for $\lambda s^2<1$ and sufficiently large $\lambda$, $\lim_d \KL_d<+\infty$ and hence, by Theorem \ref{theorem:main_result_TREES}, one-sided detectability fails for our tree correlation problem. Since there exists a range of parameters $(\lambda,s)$ for which partial alignment can be information-theoretically achieved while $\lambda s^2<1$ (it suffices to have $4<\lambda s< s^{-1}$ in view of \cite{Wu2021SettlingTS}) we therefore conclude that the conjectured hard phase for partial graph alignment (see Conjecture \ref{conjecture:hard_phase}) is non empty. \\

Note that the Kullback-Leibler divergence $\KL_d$ also coincides with the mutual information $\rI_d(T;T')$ between $T$ and $T'$ under $\dPls_d$. 
Note that under $\dPls_d$, $T$ and $T'$ are mutually independent conditionally on $\tau^*$, hence by the data processing inequality we have

\begin{equation}\label{eq:bounding_KLd_with_I}
    \KL_d=\rI_d(T,T') \leq  \rI_d(T;\tau^*) \, .
\end{equation} To establish that $\lim_d \KL_d<\infty$, it therefore suffices to prove that $\rI_d(T;\tau^*)$ is uniformly bounded. Using the simple inequality $\log(x) \leq x-1 \leq x$, we have
\begin{flalign}\label{eq:upper_bound_mutual_info}
\rI_d(T;\tau^*) &\leq \dEls_{d}\left[\frac{\dPls_{\infty}(T,\tau^*)}{\GWl_d(T) \GWls_d(\tau^*)}\right] 
 = \dEls_{d}\left[\frac{\dPls_{d}(\tau^* \, | \, T)}{\GWls_d(\tau^*)}\right] =: V_d \, .
\end{flalign}

The following Lemma gives a recursive inequality for the previously defined $V_d$. 
\begin{lemma}\label{lemma:bounding_v_d}
The quantity $V_d$ defined in \eqref{eq:upper_bound_mutual_info} verifies
\begin{equation}\label{eq:lemma:bounding_v_d}
V_d\leq f(V_{d-1}),
\end{equation}
where 
\begin{equation}\label{eq:f_bound_v_d}
f(x):=\frac{1}{1- s x}\exp\left(\frac{\kappa(1-s)(x-1)}{1-sx} \right) ,
\end{equation} with $\kappa := \lambda s^2$.
\end{lemma}

\begin{proof}[Proof of Lemma \ref{lemma:bounding_v_d}] \,

\proofstep{Step 1: a nearly-recursive formula} As previously done under $\dPls_d$, Let us denote by $c$ the degree of the root node in $\tau^*$ and $c+\Delta$ the degree of the root node in $t$.

Using Bayes's rule and partitioning on $\sigma \in \cS(c,c+\Delta)$, we have\
\begin{flalign*}
     & \dPls_{d}(\tau^* = \tau \, | \, T=t) \\
     & \quad = \frac{\GW_{\lambda s,d}(\tau)}{\GWl_{d}(t)} \p_{\lambda(1-s)}(\Delta) \sum_{\sigma \in \cS(c,c+\Delta)} \frac{\Delta!}{(c + \Delta)!} \prod_{u \in [c]} \dPls_{d-1}(T_{\sigma(u)}=t_{\sigma(u)} \, | \, \tau^*_u = \tau_u) \prod_{i = c+1}^{c+\Delta} \GW_{\lambda,d-1}(t_{\sigma(u)})   \\
    & \quad = \frac{ \p_{\lambda s}(c) \p_{\lambda(1-s)}(\Delta) }{ \p_{\lambda}(c+\Delta) } \sum_{\sigma \in \cS(c,c+\Delta)} \frac{\Delta!}{(c + \Delta)!} \prod_{u \in [c]} \dPls_{d-1}(\tau^*_u = \tau_u \, | \, T_{\sigma(u)}=t_{\sigma(u)}) \\
    & \quad = \frac{s^{c} (1-s)^{\Delta}}{c !} \sum_{\sigma \in \cS(c,c+\Delta)} \prod_{u \in [c]} \dPls_{d-1}(\tau^*_u = \tau_u \, | \, T_{\sigma(u)}=t_{\sigma(u)}), \\ &&
\end{flalign*}
so that
\begin{equation*}
\frac{\dPls_{d}(\tau^* \, | \, T)}{\GWls_d(\tau^*)}=\frac{1}{\p_{\lambda s}(c)} \cdot \frac{s^{c} (1-s)^{\Delta}}{c !}
\sum_{\sigma \in \cS(c,c+\Delta)} \prod_{u \in [c]} \frac{\dPls_{d-1}(\tau^*_u  \, | \, T_{\sigma(u)})}{\GWls_{d-1}(\tau^*_u)} \, .
\end{equation*}
Taking expectation with respect to $\dPls_{d}$ entails
\begin{equation}\label{eq_V_d_recursive}
V_d=\sum_{c \geq 0}\sum_{\Delta \geq 0} \p_{\lambda(1-s)}(\Delta)\frac{s^c (1-s)^\Delta}{c!}\sum_{\sigma\in \cS(c,c+\Delta)}\dEls_{d-1}\left[\prod_{u \in [c]} \frac{\dPls_{d-1}(\tau^*_u  \, | \, T_{\sigma(u)})}{\GWls_{d-1}(\tau^*_u)} \, \bigg| \, c,  \Delta \right] \, .
\end{equation}
To evaluate the expectation of the products in \eqref{eq_V_d_recursive} and related them to $V_{d-1}$, we need to introduce the following definitions.\\

\textit{Open paths, closed cycles.}
For two integers $c,\Delta \geq 0$ and an injective mapping $\sigma \in \cS(c,c+\Delta)$, a sequence $(u_1,\ldots,u_\ell)$ of elements of $[c]$ is
\begin{itemize}
    \item an \emph{open path of $\sigma$} if
    \begin{equation*}
        u_1 \notin \sigma([c]), \quad \forall k=1,\ldots, \ell-1, \; \sigma(u_{k})=u_{k+1}, \quad \mbox{ and } \sigma(u_\ell)\notin [c].
    \end{equation*}
    
    \item a \emph{closed cycle of $\sigma$} if
    \begin{equation*}
        \forall k=1,\ldots, \ell-1, \; \sigma(u_{k})=u_{k+1} \quad \mbox{ and }\sigma(u_\ell)=u_1.
    \end{equation*}
    See an example of such open paths and closed cycles in Figure \ref{fig:examples_orbits}.
\end{itemize}

It is an easy fact to check that each injective mapping $\sigma \in \cS(c,c+\Delta)$ can be factorized in disjoint open paths and closed cycles. Since each term in $u$ in the product in \eqref{eq_V_d_recursive} only depends on the other terms $v$ in their own open path or closed cycle, the expectation term in \eqref{eq_V_d_recursive} factorizes according to the path/cycle decomposition of $\sigma$. \\
 
\proofstep{Step 2.1: open paths.} First consider an open path $O_\ell$ of $\sigma$ of length $\ell$, assumed without loss of generality to be given by $(1,\ldots,\ell)$, so that $\sigma(1)=2,\ldots,\sigma(\ell-1)=\ell$, and $\sigma(\ell)=c+1$. The expectation of the corresponding factor reads:
\begin{equation*} \label{eq:expectation_Ol}
\dEls_{d-1}\left[\prod_{u \in O_\ell} \frac{\dPls_{d-1}(\tau^*_u  \, | \, T_{\sigma(u)})}{\GWls_{d-1}(\tau^*_u)} \right] =  \dEls_{d-1} \left[ \prod_{k=1}^{\ell-1} \frac{\dPls_{d-1}(\tau^*_k  \, | \, T_{k+1})}{\GWls_{d-1}(\tau^*_k)} \times \frac{\dPls_{d-1}(\tau^*_\ell  \, | \, T_{1})}{\GWls_{d-1}(\tau^*_\ell)}  \right].
\end{equation*}
Now integrated over $T_1$, $\dPls_{d-1}(\tau^*_\ell  \, | \, T_{1})$ evaluates to $\GWls_{d-1}(\tau^*_\ell)$ and the last factor disappears. Integrating then successively with respect to $T_k$ for $k=\ell,\ell-1,\ldots,2$, we obtain that the factors corresponding to open cycles evaluate to $1$.

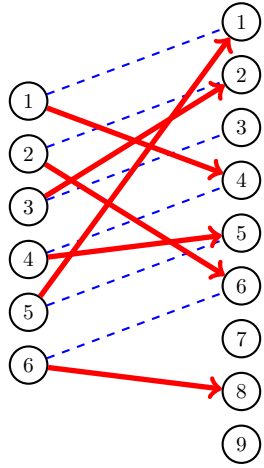
\begin{figure}[H]
	\centering
	\begin{tikzpicture}[scale=0.7,line width=0.4mm]
		\node[draw,circle,thick,scale=0.8] (A1) at (0,8.5) {$1$};
		\node[draw,circle,thick,scale=0.8] (B1) at (0,7.5) {$2$};
		\node[draw,circle,thick,scale=0.8] (C1) at (0,6.5) {$3$};
		\node[draw,circle,thick,scale=0.8] (D1) at (0,5.5) {$4$};
		\node[draw,circle,thick,scale=0.8] (E1) at (0,4.5) {$5$};
		\node[draw,circle,thick,scale=0.8] (F1) at (0,3.5) {$6$};

		\node[draw,circle,thick,scale=0.8] (A2) at (4,10) {$1$};
		\node[draw,circle,thick,scale=0.8] (B2) at (4,9) {$2$};
		\node[draw,circle,thick,scale=0.8] (C2) at (4,8) {$3$};
		\node[draw,circle,thick,scale=0.8] (D2) at (4,7) {$4$};
		\node[draw,circle,thick,scale=0.8] (E2) at (4,6) {$5$};
		\node[draw,circle,thick,scale=0.8] (F2) at (4,5) {$6$};
		\node[draw,circle,thick,scale=0.8] (G2) at (4,4) {$7$};
		\node[draw,circle,thick,scale=0.8] (H2) at (4,3) {$8$};
		\node[draw,circle,thick,scale=0.8] (I2) at (4,2) {$9$};
		
		\draw[blue,dashed,line width=0.8pt] (A1) to[bend right=0] (A2);
		\draw[blue,dashed,line width=0.8pt] (B1) to[bend right=0] (B2);
		\draw[blue,dashed,line width=0.8pt] (C1) to[bend right=0] (C2);
		\draw[blue,dashed,line width=0.8pt] (D1) to[bend right=0] (D2);
		\draw[blue,dashed,line width=0.8pt] (E1) to[bend right=0] (E2);
		\draw[blue,dashed,line width=0.8pt] (F1) to[bend right=0] (F2);

		\draw[red,line width=2pt,->] (A1) to[bend right=0] (D2);
		\draw[red,line width=2pt,->] (B1) to[bend right=0] (F2);
		\draw[red,line width=2pt,->] (C1) to[bend right=0] (B2);
		\draw[red,line width=2pt,->] (D1) to[bend right=0] (E2);
		\draw[red,line width=2pt,->] (E1) to[bend right=0] (A2);
		\draw[red,line width=2pt,->] (F1) to[bend right=0] (H2);
	\end{tikzpicture}
	\caption{Representation of $\sigma \in \cS(c,c+\Delta)$ with $c=6, \Delta=3$, and $\sigma(1)=4$, $\sigma(2)=6$, $\sigma(3)=2$, $\sigma(4)=5$, $\sigma(5)=1$, $\sigma(6) = 8$. In this example, $(1,4,5)$ (resp. $(3,2,6)$) is an open path (resp. closed cycle) of $\sigma$.}
	\label{fig:examples_orbits}
\end{figure}

\proofstep{Step 2.2: closed cycles.}
Consider next a closed cycle $C_\ell$ of $\sigma$ of length $\ell$, assumed without loss of generality to be given by $(1,\ldots,\ell)$. We also assume that the relabeling of children of the root in $\tau$ after the augmentation process is given by $\id$. Then, the expectation of the corresponding factor reads:
\begin{flalign} \label{eq:expectation_Cl}
\dEls_{d-1}\left[\prod_{u \in C_\ell} \frac{\dPls_{d-1}(\tau^*_u  \, | \, T_{\sigma(u)})}{\GWls_{d-1}(\tau^*_u)} \right] & =  \dEls_{d-1} \left[ \prod_{k\in [\ell]} \frac{\dPls_{d-1}(\tau^*_k  \, | \, T_{(k+1) \Mod \ell})}{\GWls_{d-1}(\tau^*_k)}  \right] \nonumber \\
& = \sum_{\substack{\tau_1,t_1 \in \cX_{d-1} \\ \ldots \\ \tau_\ell, t_\ell \in \cX_{d-1}}}\prod_{k\in [\ell]} \GWl_{d-1}(t_k) \dPls_{d-1}(\tau_k  \, | \, t_{k}) \cdot \frac{\dPls_{d-1}(\tau_k  \, | \, t_{(k+1) \Mod \ell})}{\GWls_{d-1}(\tau_k)} \, .
\end{flalign}
Now, for all $d \geq 0$, let us introduce the operator $\Phi_{d}$, indexed by trees in $\cX_{d}$: 
\begin{equation}\label{eq:def_op_M}
    \Phi_{d}(\tau_1,\tau_2):=\sum_{t\in \cX_{d}}\GWl_d(t)\frac{\dPls_{d}(\tau_1  \, | \, t) \cdot \dPls_{d}(\tau_2  \, | \, t)}{\sqrt{\GWls_{d}(\tau_1) \cdot \GWls_{d}(\tau_2)}} \, .
\end{equation}
$\Phi_d$ is symmetric and semi-definite positive, hence the operator is diagonalizable and its spectrum lies in $\dR_+$. Note that 
\begin{itemize}
    \item We have $$\Tr(\Phi_{d-1}) = \sum_{\tau,t \in \cX_d} \GWl_d(t)\frac{[\dPls_{d}(\tau \, | \, t)]^2}{\GWls_{d}(\tau)} = \sum_{\tau,t \in \cX_d} \dPls_{d}(\tau , t) \frac{\dPls_{d}(\tau \, | \, t)}{\GWls_{d}(\tau)} = V_d$$
    \item Moreover, when rearranging the terms, the expectation in \eqref{eq:expectation_Cl} exactly coincides with the trace of $\Phi_{d-1}^\ell$.
\end{itemize} It follows from these observations that\footnote{To make this argument fully rigorous, we consider truncated summations so that we are dealing with finite dimensional matrices, for which the trace inequality to follow clearly holds, and then use monotone convergence to obtain the desired inequality as written.}
\begin{equation} \label{eq:trace_phi_d}
    \dEls_{d-1}\left[\prod_{u \in C_\ell} \frac{\dPls_{d-1}(\tau^*_u  \, | \, T_{\sigma(u)})}{\GWls_{d-1}(\tau^*_u)} \right] =  
    \Tr(\Phi_{d-1}^\ell)\leq \Tr(\Phi_{d-1})^\ell= V_{d-1}^\ell.
\end{equation}

\proofstep{Step 3: recursive inequality.}
Now, for given $c,\Delta \geq 0$ and an injection $\sigma \in \cS(c,c+\Delta)$, let $F(\sigma)$ denote the number of elements $i\in [c]$ that belong to  closed cycles of $\sigma$. Putting together equations \eqref{eq_V_d_recursive} and \eqref{eq:trace_phi_d}, we have the following
\begin{equation*}
    V_d\leq \sum_{c,\Delta \geq 0} \p_{\lambda(1-s)}(\Delta)\frac{s^c (1-s)^\Delta}{c!}\sum_{\sigma\in \cS(c,c+\Delta)} V_{d-1}^{F(\sigma)} \, .
\end{equation*}
By \eqref{eq:upper_bound_mutual_info}, $V_d-1$ is greater that a mutual information which is alwyas non-negative. Hence, $V_{d} \geq 1$ for all $d$. To deal with the combinatorics of $F(\sigma)$, let us remark that $$F(\sigma) \leq \card{[c] \cap \sigma([c])} \, .$$ 
Then, for any $0 \leq k \leq c$, there are $\binom{c}{k}\binom{\Delta}{c-k}$ ways to chose the set $\sigma([c])$ such that $\card{[c] \cap \sigma([c])} = k$, and $c!$ distinct injections $\sigma$ with the same set $\sigma([c])$. Hence $V_d \leq f(V_{d-1})$ with
\begin{flalign*}
f(x) & := \sum_{c, \Delta \geq 0} \p_{\lambda(1-s)}(\Delta) s^c (1-s)^\Delta \sum_{k=0}^{c}\binom{c}{k}\binom{\Delta}{c-k} x^k\\
& = e^{-\lambda(1-s)} \sum_{k, \Delta \geq 0} x^k \frac{(\lambda(1-s)^2)^\Delta}{\Delta!}\sum_{c \geq k}  \binom{c}{k}\binom{\Delta}{c-k} s^c \\
& = e^{-\lambda(1-s)} \sum_{k, \Delta \geq 0} x^k \frac{(\lambda(1-s)^2)^\Delta}{\Delta!} s^k \sum_{c = 0}^{\Delta} \binom{c+k}{k}\binom{\Delta}{c} s^c \\
& = e^{-\lambda(1-s)} \sum_{k, c \geq 0} \frac{1}{c!} \binom{c+k}{k} (sx)^k s^c (\lambda(1-s)^2)^{c} \sum_{\Delta \geq c} \frac{(\lambda(1-s)^2)^{\Delta-c}}{(\Delta-c)!}\\
& = e^{-\lambda s (1-s)} \sum_{c \geq 0} \frac{(\lambda s (1-s)^2)^{c} }{c!} \sum_{k \geq 0} \binom{c+k}{k} (sx)^k\\
& = e^{-\lambda s (1-s)} \frac{1}{1-sx} \sum_{c \geq 0} \frac{1}{c!}\left(\frac{\lambda s (1-s)^2}{1-sx}\right)^{c} = \frac{1}{1-sx} \exp\left( \frac{\lambda s^2(1-s)(x-1)}{1-sx}\right) \, .
\end{flalign*} This ends the proof.
\end{proof}

We are now in a position to prove the main Theorem of this section.

\begin{theorem}\label{theorem:suff_hard_phase}
Assume $\kappa=\lambda s^2$ is fixed such that $\kappa<1$. Then for $\lambda$ sufficiently large, one-sided detectability fails.
\end{theorem}
\begin{proof}
Throughout the proof $f$ is the function defined in equation \eqref{eq:f_bound_v_d} of Lemma \ref{lemma:bounding_v_d}. \\

\proofstep{Step 1: bounding $f(1+y)$ around $y=0$.} We fix $\kappa<1$ together with $\eps \in (0,4\kappa)$ such that $\kappa+\eps<1$. Let $\gamma>0$ be an arbitrary constant chosen such that $\gamma > \frac{1}{1-\kappa-\eps}$. We emphasize that these three constants $\kappa, \eps$ and $\gamma$ are \emph{fixed}.

We shall consider $s>0$ sufficiently small, or equivalently $\lambda$ large enough, in particular such that $\gamma s<1$. Let  $y\in[0,\gamma s]$. Note that
\begin{equation*}
    \exp\left(\kappa\frac{y(1-s)}{1-s(y+1)}\right) \leq \exp\left(\kappa y /(1-2s)\right).
\end{equation*} Then, assuming $\lambda$ large enough to ensure $\frac{1}{1-2s} \leq 1+\frac{\eps}{4 \kappa}$ as well as $2 e^2 \kappa^2 \gamma s \leq \eps$, we get
\begin{equation}\label{eq:proof_th4_1}
\exp\left(\kappa y /(1-2s)\right) \leq\exp\left(\kappa y + \eps y/4 \right) \leq 1+ (\kappa+\eps/2)y \, .
\end{equation}
Note also that, $1/(1-t) \leq 1+t+3t^2$ for $t \in (0,2/3)$. Assuming $\lambda$ large enough to ensure $s(y+1)\leq 2s < 2/3$, and using $y \leq \gamma s \leq 1$, we get
\begin{equation}\label{eq:proof_th4_2}
    \frac{1}{1-s(y+1)} \leq 1+s(y+1)+3(s(y+1))^2 \leq 1+s + C s^2,
\end{equation} where $C := \gamma + 12$. Together, these last two bounds \eqref{eq:proof_th4_1} and \eqref{eq:proof_th4_2} entail, for any $y \in [0,\gamma s]$:
\begin{equation}\label{eq:bound_iter_f}
    f(1+y)-1 \leq (1+(\kappa +\eps/2)y)(1+s+ C s^2)-1 \leq s+ C s^2 +(\kappa+\eps)y,
\end{equation}
where we assumed $\lambda$ large enough to ensure $(\kappa+\eps/2)(1+s+C s^2) \leq \kappa+\eps$. 
Note now that since $1 + (\kappa+\eps)\gamma < \gamma$ (see the intial assumptions), we can take $\lambda$ large enough so that $1+(\gamma+12)s + (\kappa+\eps)\gamma \leq \gamma$, and then it holds that for any $y \in [0,\gamma s]$:
\begin{equation}\label{eq:bound_induction_ok}
0 \leq f(1+y)-1 \leq \gamma s \, .
\end{equation}

\proofstep{Step 2: use induction.} Let us now prove by recursion that $V_d \leq 1+\gamma s$. Note that $V_0=1$, hence the statement is true at $d=0$. Assume that $V_0, \ldots, V_d$ are all less than $1+\gamma s$. Then, $V_d-1 \in [0,\gamma s]$ and by Lemma \ref{lemma:bounding_v_d} and equation \eqref{eq:bound_induction_ok}, we have
$$ V_{d+1} \leq f(V_d) = f(V_d-1+1) \leq 1 + \gamma s \, .$$

Since we stablished ealier on in \eqref{eq:bounding_KLd_with_I} and \eqref{eq:upper_bound_mutual_info} that $$ \KL_d=\rI_d(T,T') \leq  \rI_d(T;\tau^*) \leq V_d ,$$ this proves that $\lim_d \KL_d < +\infty$, and applying Theorem \ref{theorem:main_result_TREES} ends the proof.
\end{proof}

\section{Consequences for polynomial time partial graph alignment}\label{section:graph_matching}
We now apply the previous results of Sections \ref{section:LR}, \ref{section:KL} and \ref{section:autos_GW} to one-sided partial graph alignment. Recall that we work under the correlated \ER model \eqref{eq:CER_model}. 
We will now describe our polynomial-time algorithm and its theoretical guarantees when one-sided detectablity holds in Theorem \ref{theorem:main_result_TREES}  -- in particular under the conditions of Theorem \ref{theorem:suff_cond_KL} or condition \eqref{eq:theorem:suff_cond_auto} of Theorem \ref{theorem:suff_cond_auto}.

\subsection{Intuition, algorithm description}\label{subsection:intuition_algorithm_desription}
In all this part we assume that $(\lambda,s)$ satisfy one of the conditions in Theorem \ref{theorem:main_result_TREES}.\\

\textit{Extending the tree correlation detection problem.} 
Let $(G, H) \sim \G(n,q=\lambda/n,s)$, with underlying alignment $\pi^*$. In order to distinguish matched pairs of nodes $(u,u')$, we consider their neighborhoods $\cN_{d,G}(u)$ and $\cN_{d,H}(u')$ at a given depth $d$: these neighborhoods are asymptotically distributed as Galton-Watson trees. In the case where the two vertices are actual matches, i.e. $u' = \pi^*(u)$, we are exactly in the setting of our tree correlation detection problem under $\dPls_{d}$: Point $(iii)$ of in Theorem \ref{theorem:main_result_TREES} shows that there exists a threshold $\theta_d$ such that with probability at least $1-\pext(\lambda s)>0$,
$$L_d(u,u') := L_d \left(\cN_{d,G}(u),\cN_{d,H}(u')\right) > \theta_d,$$ when $d \to \infty$. 
Point $(v)$ of Theorem \ref{theorem:main_result_TREES} shows that this threshold $\theta_d$ can be e.g. taken to be $\exp(n^\gamma)$ for some $\gamma \in (0,c \log (\lambda s))$.

At the same time, when nodes $u'$ and $\pi^*(u)$ are distinct and sufficiently far away in the underlying union graph, we can argue that we are also -- with high probability -- in the setting of the tree correlation detection problem under $\dPl_{d}$: since $\dEl_{d}\left[L_d\right]=1$, Markov's inequality shows that with high probability when $d \to \infty$,
$$L_d(u,u') \leq \theta_d.$$

\textit{Computation of the likelihood ratios.}
As mentioned in Remark \ref{remark:util_rec_algo}, formula \eqref{eq:lemma:LR_rec} of Lemma \ref{lemma:LR_rec} enables to compute such likelihood ratios efficiently on a graph, giving the exact expression for a \emph{message-passing} procedure, assuming that all neighborhoods are locally tree-like at depth $d$. Let us first define \emph{oriented likelihood ratios}: for any nodes $u,v$ of $G$ and nodes $u',v'$ of $H$, we write $L_d(u \leftarrow v,u' \leftarrow v')$ for the likelihood ratio at depth $d$ of two trees, the first one (resp. second one) being rooted at $u$ in $G$ (resp. $u'$ in $H$) where the edge $\left\lbrace u,v \right\rbrace$ (resp. $\left\lbrace u',v' \right\rbrace$), if present in the first place, has been deleted. In view of the recursive formula \eqref{eq:lemma:LR_rec} of Lemma \ref{lemma:LR_rec} these oriented likelihood ratios satisfy the following recursion:

\begin{equation}\label{eq:rec_oriented_LR}
    L_d(u \leftarrow v,u' \leftarrow v') = \sum_{k=0}^{d_u \wedge d'_{u'} - 1} \Psi \left(k, d_u - 1,d'_{u'} - 1\right) \sum_{\substack{\sigma \in \cS\left([k],\cN_{G}(u) \setminus \{ v\}  \right) \\ \sigma' \in \cS\left([k],\cN_{H}(u') \setminus \{ v'\} \right)}} \prod_{\ell=1}^{k} L_{d-1}(\sigma(\ell) \leftarrow u,\sigma'(\ell) \leftarrow u'),
\end{equation} where $d_u$ (resp. $d'_{u'}$) is the degree of $u$ in $G$ (resp. of $u'$ in $H$). The likelihood ratio $L_d(u,u')$ is then obtained by computing
\begin{equation}\label{eq:total_oriented_LR}
    L_d(u,u') = \sum_{k=0}^{d_u \wedge d'_{u'}} \Psi \left(k, d_u,d'_{u'}\right) \sum_{\substack{\sigma \in \cS\left([k],\cN_{G}(u) \right) \\ \sigma' \in \cS\left([k],\cN_{H}(u')  \right)}} \prod_{\ell=1}^{k} L_{d-1}(\sigma(\ell) \leftarrow u,\sigma'(\ell) \leftarrow u') \, .
\end{equation}

A natural idea is then to compute $L_d(i,u)$ for each pair $(u,u')$ with $d$ large enough (typically scaled in $\Theta(\log n)$ where $n$ is the number of vertices in $G$ and $H$) and to compare it to $\theta_d$ to decide whether $u$ in $G$ is matched to $u'$ in $H$. \\

\textit{A refined dangling trees trick.} 
However, as previously noted in \cite{Ganassali20a}, without  additional constraint, this strategy produces many falsely positive matches, tending e.g. to match $u$ with $u' \neq \pi^*(u)$ if there exists $v$ such that $\set{u,v}$ is an edge of $G$ and $\set{u',\pi^*(v)}$ is an edge of $H$, making the errors increase and the performance collapse. 

To fix this issue, we use the \emph{dangling trees trick}, already introduced in \cite{Ganassali20a}, improved here by considering three rather than two dangling trees: instead of just looking at their neighborhoods, we look for the downstream trees from distinct neighbors of $u$ in $G$ and of $u'$ in $H$. The trick is now to match $u$ with $u'$ if and only if there exists three distinct neighbors $v,w,x$ of $u$ in $G$ (resp. $v',w',x'$ of $u'$ in $H$) such that all three of the likelihood ratios $L_{d-1}(v \leftarrow u, v' \leftarrow u')$, $L_{d-1}(w \leftarrow u, w' \leftarrow u')$ and $L_{d-1}(x \leftarrow u, x' \leftarrow u')$ are larger than $\theta_{d-1}$. The proof of Theorem \ref{theorem:no_mismatchs} explains how this trick avoids false positives and why three dangling trees is a good choice.\\

\textit{Algorithm description.}
Our algorithm is as follows:
\begin{algorithm}[h]
\caption{\label{algo_GA} \texttt{MPAlign}: Message-passing algorithm for sparse graph alignment}
\SetAlgoLined

\textbf{Input:} Two graphs $G$ and $H$ of size $n$, average degree $\lambda$, depth $d$, threshold parameter $\theta$

\textbf{Output:} A set of pairs $\cM \subset V(G) \times V(H)$.

$\cM \gets \varnothing$

Compute $L_d(u \leftarrow v,u' \leftarrow v')$ for all $\set{u,v}\in E(G)$ and $\set{u',v'}\in E(H)$ with \eqref{eq:rec_oriented_LR}

\For{$(i,u) \in V(G) \times V(H)$}{
	\If{$\cN_{G}(u,d)$ and $\cN_{H}(u',d)$ contain no cycle, and $\exists \left\lbrace v,w,x \right\rbrace \subset \cN_{G}(u), \exists \left\lbrace v',w',x' \right\rbrace \subset \cN_{H}(u')$ such that $L_{d-1}(v \leftarrow u, v' \leftarrow u')> \theta$, $L_{d-1}(w \leftarrow u, w' \leftarrow u')> \theta$ and $L_{d-1}(x \leftarrow u, x' \leftarrow u')> \theta$}
	{	
	$\cM \gets \cM \cup \left\lbrace (u,u') \right\rbrace $
	}
}
\textbf{return} $\cM$

\end{algorithm}

\begin{remark}
To update the matrix of all likelihood ratios with \eqref{eq:rec_oriented_LR}, we update a matrix of size $O(n^2)$, each entry of which can be computed in time $O\left((d_{\max}!)^2\right)$ -- where $d_{\max}$ is the maximum degree in $G$ and $H$. Under the correlated \ER model,  $d_{\max} = O\left( \frac{\log n}{\log \log n} \right)$ \cite{Bollobas2001}, so that $d_{\max}!$ is polynomial in $n$.  Each iteration is thus polynomial in $n$ and since $d$ is taken order $\log(n)$, \texttt{MPAlign} (Algorithm \ref{algo_GA}) runs in polynomial time.
\end{remark}

We now state two results that will readily imply Theorem \ref{theorem:main_result_GRAPHS}.
\begin{theorem} \label{theorem:good_matches}
Let $(G,H) \sim \G(n,q=\lambda/n,s)$ and assume that any of the equivalent conditions of Theorem 1 holds.
Let $d = \lfloor c \log n \rfloor$ with $c \log \left(\lambda\left(2-s\right)\right)<1/2$. Let $\cM$ be the output of Alg. \ref{algo_GA}, taking $\theta  = \exp(n^\gamma)$ for some $\gamma \in (0,c \log (\lambda s))$. Then with high probability
\begin{equation}
\label{eq:theorem:good_matches}
\frac{1}{n} \sum_{u=1}^{n} \one_{\lbrace (u,\pi^*(u)) \in \cM \rbrace} \geq \Omega(1).
\end{equation} 
In other words, a non vanishing fraction of nodes is correctly recovered by \texttt{MPAlign}.
\end{theorem}

\begin{theorem}\label{theorem:no_mismatchs}
Let $(G,H) \sim \G(n,q=\lambda/n,s)$. Assume that $d = \lfloor c \log n \rfloor$ with $c \log \lambda<1/4$. Let $\cM$ be the output of Alg. \ref{algo_GA}, taking $\theta  = \exp(n^\gamma)$ for some $\gamma \in (0,c \log (\lambda s))$. Then with high probability
\begin{equation}
\label{eq:theorem:no_mismatchs}
\mathrm{err}(n):=\frac{1}{n}\sum_{u=1}^{n} \one_{\lbrace \exists u'  \neq \pi^*(u), \; (u,u') \in \cM \rbrace}=o(1),
\end{equation} 
i.e. at most a vanishing fraction of nodes are incorrectly matched by \texttt{MPAlign}.
\end{theorem}

\begin{remarque}
The set $\cM$ returned by Algorithm \ref{algo_GA} is not necessarily an injective mapping. Let $\cM'$ be obtained by removing all pairs $(u,u')$ of $\cM$ such that $u$ or $u'$ appears at least twice. Theorems \ref{theorem:good_matches} and \ref{theorem:no_mismatchs} guarantee that $\cM'$ still contains a non-vanishing number of correct matches and a vanishing number of incorrect matches, hence one-sided partial alignment holds. Since \texttt{MPAlign} achieves one-sided partial graph alignment, Theorem \ref{theorem:main_result_GRAPHS} easily follows.

A slight adaptation \texttt{MPAlign2} (Alg. \ref{algo_MPAlign2}) of \texttt{MPAlign} (Alg. \ref{algo_GA}) can be found in Appendix \ref{appendix:numerical}, where some results are also reported. These confirm our theory, as the algorithm returns many good matches and few mismatches. A similar algorithm has been recently studied in \cite{piccioli2021aligning}.
\end{remarque}

\subsection{Proof strategy}\label{subsection:proof_strategy_GA}
We start by stating Lemmas that precise the correspondence between sparse graph alignment and correlation detection in trees, as explained in Section \ref{subsection:intuition_algorithm_desription}. These Lemmas are directly taken from \cite{Ganassali20a} (to which we refer for the proofs, see Lemmas 2.1, 2.2, 2.3 and 2.4) and are instrumental in the proofs of Theorems \ref{theorem:good_matches} and \ref{theorem:no_mismatchs}.

\begin{lemme}[Control of the sizes of the neighborhoods]
	\label{lemma:control_S}
	Let $G \sim \G(n,\lambda/n)$, $d = \lfloor c \log n \rfloor$ with $c \log \lambda <1$. For all $\gamma>0$, there is a constant $C=C(\gamma)>0$ such that with probability $1-O\left(n^{-\gamma}\right)$, for all $u \in [n]$, $t \in [d]$:
	\begin{equation}
	\label{eq:lemma:control_S}
	\left| \cS_{G}(u,t) \right| \leq C (\log n) \lambda^t \, .
	\end{equation}
\end{lemme}

\begin{lemme}[Cycles in the neighborhoods in an \ER graph]
\label{lemma:cycles_ER}
Let $G \sim \G(n,\lambda/n)$, $d = \lfloor c \log n \rfloor$ with $c \log \lambda <1/2$. Then there exists $\eps>0$ such that for any vertex $u \in [n]$, one has
\begin{equation}
\label{eq:lemma:cycles_ER}
\mathbb{P}\left(\cN_{G,d}(u) \mbox{ contains a cycle}\right) = O\left( n^{-\eps}\right).
\end{equation}
\end{lemme}

\begin{lemme}[Two neighborhoods are typically independent]
\label{lemma:indep_deighborhoods}
Let $G \sim \G(n,\lambda/n)$ with $\lambda >1$, $d = \lfloor c \log n \rfloor$ with $c \log \lambda < 1/2 $. Then there exists $\eps>0$ such that for any fixed nodes $u \neq v$ of $G$, the total variation distance between the joint distribution of the neighborhoods $\mathcal{L} \left(\left(\cS_{G}(u,t),\cS_{G}(v,t)\right)_{t \leq d}\right)$ and the product distribution $\mathcal{L} \left(\left(\cS_{G}(u,t)\right)_{t \leq d}\right) \otimes \mathcal{L} \left(\left(\cS_{G}(v,t)\right)_{t \leq d}\right)$ tends to $0$ as $O\left(n^{-\eps}\right)$ when $n \to \infty$.
\end{lemme}

\begin{lemme}[Coupling neighborhoods with Galton-Watson trees]
\label{lemma:coupling_GW} We have the following couplings:
\begin{itemize}
    \item[$(i)$] Let $G \sim \G(n,\lambda/n)$, $d = \lfloor c \log n \rfloor$ with $c \log \lambda<1/2$. Then there exists $\eps>0$ such that for any fixed node $u$ of $G$, the variation distance between the distribution of $\cN_{G,d}(u)$ and the distribution $\GWl_{d}$ tends to 0 as $O\left( n^{-\eps}\right)$ when $n \to \infty$.
    \item[$(ii)$] For $(G,H) \sim \G(n,q=\lambda/n,s)$ with planted alignment $\pi^*$, $d = \lfloor c \log n \rfloor$ with $c \log (\lambda s)<1/2$ and $c \log (\lambda (1-s))<1/2$, there exists $\eps>0$ such that for any fixed node $u$ of $G$, the variation distance between the distribution of $(\cN_{G,d}(u),\cN_{H,d}(\pi^*(u)))$ and the distribution $\dPls_{\infty}$ (as defined in Section \ref{subsection:model_random_trees}) tends to 0 as $O\left( n^{-\eps}\right)$ when $n \to \infty$.
\end{itemize}
\end{lemme}

\textbf{Proof of Theorems \ref{theorem:good_matches} and \ref{theorem:no_mismatchs}}
\begin{proof}[Proof of Theorem \ref{theorem:good_matches}] 
First, since $c \log \left(\lambda\left(2-s\right)\right)<1/2$, we also have $c \log \left(\lambda\left(1-s\right)\right)<1/2$ and $c \log \left(\lambda s\right)<1/2$. For $i \in [n]$, point $(ii)$ of Lemma \ref{lemma:coupling_GW} thus implies that the two neighborhoods $\cN_{G,d}(u)$ and $\cN_{H,d}(\pi^*(u))$ can be coupled with trees drawn under $\dPls_{\infty}$ as defined in Section \ref{subsection:model_random_trees} with probability $\geq 1-O(n^{- \eps})$. 

Under this coupling, there is a probability $\alpha_3 >0$ that the root in the intersection tree has at least three children, and since we work under the conditions of Theorem \ref{theorem:main_result_TREES} point $(v)$ implies that the three likelihood ratios are greater than $\theta$ with positive probability $(1-\pext(\lambda s))^3>0$. Hence, the probability of $M_u := \left\lbrace (u,\pi^*(u)) \in \cM \right\rbrace$ is at least $(1-o(1)) \alpha_3 (1-\pext(\lambda s))^3  =: \alpha >0$. 

Let $G_{\cup}$ be the underlying aligned union graph, that is the graph made of the union of the edge set of $G $ and the edge set of $G'$ -- which we recall is the version of $H$ \emph{before} the relabeling step.
$G^{\pi^*}$ is the relabeling of $G$ according to permutation $\pi^*$. We have $G_{\cup} \sim \G(n,\lambda(2-s)/n)$. For $u \neq v \in [n]$, define $I_{u,v}$ the event on which the two neighborhoods of $u$ and $v$ in $G_{\cup}$ coincide with their independent couplings up to depth $d$. Since $c \log \left(\lambda\left(2-s\right)\right)<1/2$, by Lemma \ref{lemma:indep_deighborhoods}, $\mathbb{P}(I_{u,v})=1-o(1)$. Then for $0<\eps<\alpha$, Markov's inequality yields

\begin{flalign*}
\mathbb{P}\left(\frac{1}{n} \sum_{i=1}^{n} \one_{\lbrace (u,\pi^*(u)) \in \cM \rbrace}<\alpha-\eps\right) & \leq \mathbb{P}\left(\sum_{u \in [n]} \left(\mathbb{P}(M_u)-\one_{M_u}\right)>\eps n\right)\\
& \leq \frac{1}{n^2 \eps^2} \left(n \mathrm{Var}\left(\one_{M_1}\right)+ n(n-1)\mathrm{Cov}\left(\one_{M_1},\one_{M_2}\right) \right)\\
& \leq \frac{\mathrm{Var}\left(\one_{M_1}\right)}{n \eps^2} + \frac{1-\mathbb{P}\left(I_{1,2}\right)}{ \eps^2} \to 0,
\end{flalign*}
which ends the proof.
\end{proof}

\begin{remark}
Note that in view of the proof here above, the recovered fraction $\Omega(1)$ guaranteed by in Theorem \ref{theorem:good_matches} can be taken as close as wanted to
 $$ \alpha(\lambda s) := (1-\pext(\lambda s))^3 \left(1- \p_{\lambda s}(0) - \p_{\lambda s}(1) - \p_{\lambda s}(2) \right). $$
This fraction is a priori not optimal, and can be interestingly compared with recent results in \cite{ganassali2021impossibility} showing that no more than a fraction $1-\pext(\lambda s)$ of the nodes can be recovered. 
\end{remark}

\begin{proof}[Proof of Theorem \ref{theorem:no_mismatchs}]
First, we condition on the event $\cA$ that all $d-$neighborhoods in $G$ and $H$ are of size at most $C( \log n)\lambda^d$, which happens with probability $1-o(1)$ by Lemma \ref{lemma:control_S}. Note that by assumption this uniform upper bound is $O((\log n)n^{1/4})$.

In order to control the probability that $u$ is matched with some 'wrong' $u'$ by our algorithm, we follow the same first steps as in the proof of Theorem 2.2. in \cite{Ganassali20a}: we will first fix $u$ in $G$ and work on the event $\cE_u$ where $\cN_{G_{\cup},2d}(u)$ has no cycle. Since $c\log(\lambda) <1/4$, this event happens with probability $1-o(1)$ by Lemma \ref{lemma:cycles_ER}. 

Consider then $u'$ in $H$ such that $u \neq \pi^*(u)$. If $u$ and $u'$ are matched by \texttt{MPAlign}, then necessarily $\cN_{G}(u,d)$ and $\cN_{H}(u',d)$ contain no cycle: the $d-$neighborhoods are thus tree-like. For any choice of distinct neighbors $v,w,x$ of $u$ in $G$ (resp. $v',w',x'$ of $u'$ in $H$), we define the corresponding pairs of trees of the form $(T_\ell,T'_\ell)$, where $T_\ell$ (resp. $T'_\ell$) is the tree of depth $d-1$ rooted at $\ell \in \set{v,x,w}$ in $G$ (resp. at $\ell \in \set{v',x',w'}$ in $H$) after deletion of edge $\left\lbrace u,\ell \right\rbrace$ in $G$ (resp. $\left\lbrace u',\ell \right\rbrace$ in $H$). A moment of thought shows that, no matter the choice of $v,x,w$ and $v',x',w'$, on event $\cE_i$, one of these pairs must be made of \emph{two disjoint trees}. 

We now focus on a pair $(T,T')$ of such disjoint trees: these trees of depth $d-1$ can be built recursively by sampling a binomial number of children for each vertex. Since we condition on the fact that the trees are not intersecting, if at some point $k$ vertices have been uncovered, then the number of children to be drawn is exactly of distribution $\Bin\left(n-k , \lambda/n \right)$. With this exact construction, we denote by $\widetilde{\dP}_d$ the distribution of the pair $(t,t')$. Define 

\begin{equation}\label{eq:M_def}
    N_{d-1} := \frac{\widetilde{\dP}_{d-1}(t,t')}{\dPl_{d-1}(t,t')}.
\end{equation}
We have that
\begin{flalign*}
\widetilde{\dP}_{d-1}(L_{d-1}(T,T')> \theta \cap \cA ) & = \dEl_{d-1}\left[N_{d-1} \times \one_{\cA} \times \one_{L_{d-1}(T,T')> \theta} \right] \\
& \leq \dEl_{d-1}[N_{d-1}^2 \one_{\cA} ]^{1/2} \theta^{-1/2},
\end{flalign*} by a successive use of Cauchy-Schwarz and Markov's inequalities. We now state the following Lemma, proved in Appendix \ref{appendix:proof_lemma:control_M2}:

\begin{lemma}\label{lemma:control_M2}
With the previous notations, we have
\begin{equation}\label{eq:lemma:control_M2}
    \dEl_{d-1}\left[ N_{d-1}^2 \one_{\cA} \right] = O(1).
\end{equation}
\end{lemma}

Together with the previous Lemma, noting that with high probability the maximum degree in $G$ and $H$ is less than $\log n$, union bound gives
\begin{flalign*}
\mathbb{P}\left(\cA \cap \left\{\exists u  \neq \pi^*(u), \; (u,u') \in \cM \right\}\right) & \leq \dP(\bar{\cE_i})+ o(1) + n \times \log^6 n \times \theta^{-1/2} \\
& = O\left((\log^6 n) \times n \times \exp(-n^{\gamma/2})\right) = o(1).
\end{flalign*}
The proof follows by appealing to Markov's inequality.
\end{proof}

\section{Conclusion, open questions}\label{section:conclusion}
Detection of correlation in trees, introduced and studied in this paper, is a fundamental statistical task of intrinsic interest besides its original motivation from  graph alignment. While in this paper we focus on \ER graphs and hence Poisson branching trees, more general locally tree-like graphs could be considered, giving rise to correlation detection problems on more general branching trees.\\

The present work suggests the following open questions:

\begin{itemize}
    \item[Q1.] Recall that the non-planted version of graph alignment of two graphs with adjacency matrices $A$ and $B$ consists in solving the quadratic assignment problem \eqref{eq:QAP}.
    
    \emph{In the case where $A,B$ are independent \ER graphs, what is the value of the objective $$\max_\Pi \langle A, \Pi B \Pi^T\rangle$$ in the large size limit? }
    
    Some upper bounds are  obtained in the literature \cite{Wu20} -- to study the detection problem -- but to the best of our knowledge no exact equivalent is known.
    
    \item[Q2.] In a previous paper \cite{Ganassali20a}, another similarity score between trees $t$ are $t'$ is studied: the tests are based on the matching weight, defined as the largest number of leaves at depth $d$ of a common subtree of $t$ and $t'$. \emph{Under the null model, where $t$ and $t'$ are e.g. independent Galton-Watson trees, what is the typical matching weight of $t$ and $t'$?}
     
    \item[Q3.] A locally tree-like model in which graph alignment appears very challenging is the regular model. In particular, any method based on exploiting the locally tree-like structure -- if no other information such as labels on nodes is known -- will fail. So, we may ask the question: \emph{what are the information-theoretic and computational limits for regular graph alignment?}

    \item[Q4.] \emph{What is the optimal overlap -- or, the largest subset $\cC^*$ -- that one can hope to align in the sparse regime?} It is shown in \cite{ganassali2021impossibility} that -- up to some vanishing fraction of the nodes --  $\cC^*$ is contained in the giant component $\cC_1$ of the aligned intersection graph. 

\end{itemize}

We end this conclusion by providing an answer to question Q4. in the isomorphism case $s=1$, and a conjecture for the general case $s \in [0,1]$.

In the exact isomorphism case $s=1$, leveraging on results of Luczak \cite{luczak1988} on the structure of the automorphism group of random graphs, we are able to show that in the exact isomorphism case $s=1$, $\cC^*$ is almost -- i.e, up to some vanishing fraction -- the set of all points invariant by any automorphism. 
More precisely, it is shown in Appendix \ref{appendix:discussion} that for $\lambda$ greater than some constant, with high probability, any isomorphism $\hat{\pi}$ between $G$ and $H$ will achieve partial recovery and will satisfy
$$ \ov \left(\hat{\pi},\pi^* \right) \geq 1 - \pext(\lambda) - \lambda(\lambda + 5)e^{-2\lambda}, $$
where we recall that $\pext(\lambda)$ is defined as the probability that a Galton-Watson
tree of offspring $\Poi(\lambda)$ survives. This is the best overlap up to some vanishing terms when $\lambda \to \infty$. We refer to Appendix \ref{appendix:discussion} for further discussion.

We conjecture that this observation could be generalized to the non-isomorphic case $s <1$, namely that $\cC^*$ is almost the set $\cI$ of invariant nodes \emph{in the aligned intersection graph}.

\addcontentsline{toc}{section}{Conclusion, open questions}

\section*{Acknowledgments}
The authors would like to thank Guilhem Semerjian for helpful discussions, and Jakob Maier for its feedback on some of the proofs.
This work was partially supported by the French government under management of Agence Nationale de la Recherche as part of the “Investissements d’avenir” program, reference ANR19-P3IA-0001 (PRAIRIE 3IA Institute).

\bibliographystyle{plain}
\bibliography{biblio.bib}

\appendix

\section{Discussion: the isomorphic case ($s=1$)}
\label{appendix:discussion}
We discuss here the graph alignment problem in the case where $s=1$ in the correlated \ER model \eqref{eq:CER_model}, namely when the graphs $G$ and $H$ are isomorphic, $\pi^*$ being  one of the graph isomorphisms between $G$ and $H$. We then ask the question (Q4. hereabove): \emph{what is the best fraction of nodes that can be recovered with high probability?} 

The answer to the above question comes with the following easy remark: the joint distribution of $(G,H)$ is invariant by any relabeling of $G$ according to some $\sigma \in \Aut(G)$, where $\Aut(G)$ denotes the automorphism group of $G$. The set of nodes that can be aligned w.h.p. is hence
\begin{equation}
    \label{eq:def_I(G)}
    \cI(G) := \left\lbrace i \in V(G), \; \forall \, \sigma \in \Aut(G), \sigma(i)=i \right\rbrace.
\end{equation}
In other words, $\cI(G)$ is the set of vertices of $G$ invariant under any automorphism.

Let us denote $\cC_1(G)$ the largest connected component of $G$ (the \emph{giant component}), and $\overline{\cC_1(G)}$ the subgraph made of all the smaller components. It is clear that
\begin{equation*}
    \Aut(G) = \Aut(\cC_1(G)) \times \Aut(\overline{\cC_1(G)}).
\end{equation*}
Recent work \cite{ganassali2021impossibility} shows that $\cI(G) \, \cap \, \overline{\cC_1(G)}$ contains at most a vanishing fraction of the points: it is not hard to see indeed that smaller components mainly consist in isolated trees, which are proved to have many copies in the graph when $n$ gets large, yielding some automorphisms that swap almost all vertices in $\overline{\cC_1(G)}$.
Hence, for our purpose, the main part of $\cI(G)$ comes from the study of $\Aut(\cC_1(G))$ and $\cI(\cC_1(G))$. 

When $G \sim \G(n,q)$, these sets have been thoroughly studied by Łuczak in \cite{luczak1988}. Vertices of the giant component that are not invariant under automorphism are mainly (i.e. up to $o(n)$ errors) vertices that do not belong to the \emph{2-core}\footnote{The \emph{2-core} of a graph is defined as the maximal subgraph of minimal degree at least $2$.} of $G$, denoted by $\cC^{(2)}(G)$. 

Simple structures appearing in $\cC_1(G) \setminus \cI(G)$ are leaves (degree one nodes) $j,k$ with common neighbor $u$ in $\cC_1(G)$.
\cite{luczak1988} upper-bounds the size of $\cC_1(G) \setminus \cI(G)$ by the number of (generalizations) of such structures, thus obtaining the following

\begin{theorem}[\cite{luczak1988}, Theorems 3 and 4] \label{th:luczak}
Let $G \sim \G(n,q)$ with $q = \lambda/n$. Let $(K_n)_n$ be a sequence such that $K_n \to \infty$. There exists $\lambda_0 >0$ such that if $\lambda > \lambda_0$, then with high probability,
\begin{equation}\label{eq:th:luczak}
    \card{\cC^{(2)}(G)} - \card{\cI(\cC^{(2)}(G))} \leq K_n, \quad \mbox{and} \quad \card{\cC_1(G)} - \card{\cI(\cC_1(G))} \leq \lambda(\lambda + 5)e^{-2\lambda} n.
\end{equation}
\end{theorem}

Equation \eqref{eq:th:luczak} of Theorem \ref{th:luczak} states that for $\lambda$ greater than some constant, almost all vertices of the 2-core of $G$ are invariant, whereas at most a fraction $\lambda(\lambda + 5)e^{-2\lambda}$ of the nodes are in the giant component and not in $\cI(G)$. In this case, with high probability, any isomorphism $\hat{\pi}$ between $G$ and $H$ will achieve partial recovery and will satisfy
$$ \ov \left(\hat{\pi},\pi^* \right) \geq 1 - \pext(\lambda) - \lambda(\lambda + 5)e^{-2\lambda}, $$
where we recall that $\pext(\lambda)$ is defined as the probability that a Galton-Watson
tree of offspring $\Poi(\lambda)$ survives. 

However, finding efficiently such an isomorphism $\hat{\sigma}$ is known to be challenging in the general case (see e.g. \cite{Arvind2002}): hence, whether there exists a polynomial time algorithm achieving this optimal bound remains an open question.

\begin{figure}[H]
    \centering
    \vspace{0.2cm}
    \includegraphics[scale=0.5]{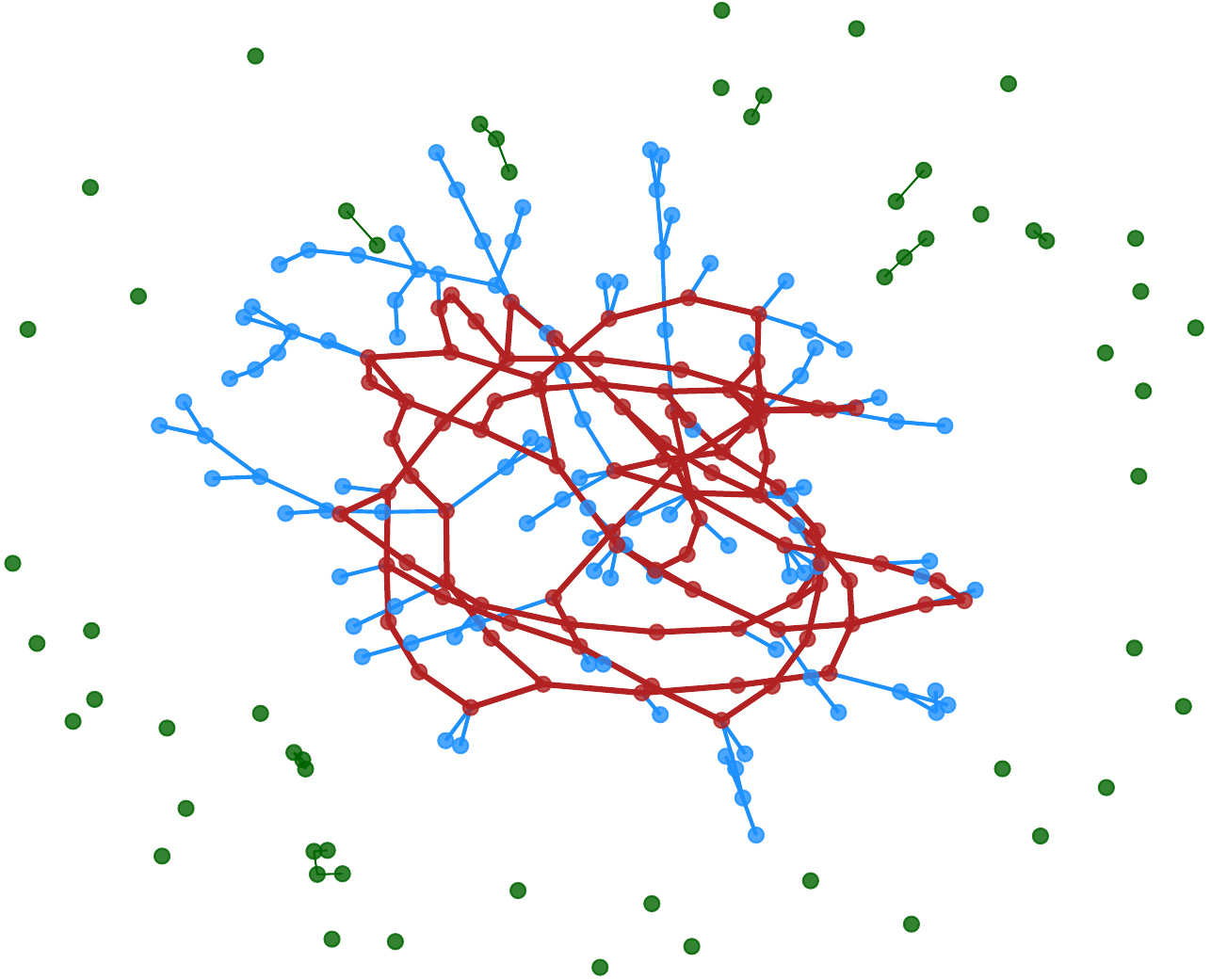}
    \caption{Sample $G$ from model $\G(n,\lambda/n)$, with $\lambda = 2$ and $n=250$. Vertices of $\overline{\cC_1(G)}$ (resp. of $\cC_1(G) \setminus \cC^{(2)}(G)$, $\cC^{(2)}(G)$) are shown in green (resp. blue, red).}
    \label{fig:example_two_core}
\end{figure}

\section{Numerical experiments for \texttt{MPAlign2}}\label{appendix:numerical}

In this section, we give some details on a practical implementation of our algorithm. The code used for these experiments is available at https://github.com/mlelarge/\texttt{MPAlign}.jl.

Given an edge $\set{u,v}$ of a graph, we denote by $u\to v$ and $v \to u$ the directed edges from $u$ to $v$ and $v$ to $u$ respectively. Now given two graphs $G=(V,E)$ and $H=(V',E')$, we define the matrix $(m^t_{u \to v, u'\to v'})_{\set{u,v}\in E , \set{u',v'}\in E'}\in \R_+^{2|E|\times 2|E'|}$ recursively in $t$, as follows:
\begin{equation}
\label{eq:message}
    m^{t+1}_{u \to v, u'\to v'} =\sum_{k=0}^{d_u\wedge d'_{u'} -1}\widetilde{\Psi}(k,d_u-1,d'_{u'}-1)\sum_{\substack{\{\ell_1,\dots \ell_k\}\in \partial u\backslash v \\ \{w_1,\dots w_k\}\in \partial u'\backslash v'}}\sum_{\sigma\in \cS_k}\prod_{a=1}^k m^t_{\ell_a\to u, w_{\sigma(a)}\to u'},
\end{equation}
where where $d_u$ (resp. $d'_{u'}$) is the degree of $u$ in $G$ (resp. of $u'$ in $H$), $\widetilde{\Psi}(k,d_1,d_2) = k!\Psi(k,d_1,d_2)$, and $\partial u \backslash v$ (resp. $\partial u' \backslash v'$) is a shorthand notation for $\cN_G(u)\setminus \{v\}$ (resp. $\cN_{H}(u')\setminus \{v'\}$) and by convention $m^0_{u \to v, u'\to v'}=1$.

Denoting $\partial u := \cN_G(u)$ (resp. $\partial u' := \cN_{H}(u')$), for $t \geq 0$ we define the matrix $(m^t_{u,u'}) \in \R_+^{V\times V'}$:
\begin{equation}
\label{eq:aggregation}
    m^t_{u,u'} = \sum_{k=0}^{d_u\wedge d'_{u'}}\widetilde{\Psi}(k,d_u,d'_{u'})\sum_{\substack{\{\ell_1,\dots \ell_k\}\in \partial u \\ \{w_1,\dots w_k\}\in \partial u'}}\sum_{\sigma\in \cS_k}\prod_{a=1}^k m^t_{\ell_a\to u, w_{\sigma(a)}\to u'}.
\end{equation}

We straightaway see that if the graphs $G$ and $H$ are tree-like up to depth $t$, then $m^t_{u,u'}$ is exactly the likelihood ratio $L_t(u,u')$ previously defined.

In experiments, we run our algorithms on correlated \ER model with possible cycles so that $m^t$ is interpreted as an approximation of the true likelihood ratio. From such an approximation, we compute two mappings $\pi_{\mathrm{left}}^t:V\to V'$ as $$\pi_{\mathrm{left}}^t(u) = \argmax(m^t_{u,\cdot})$$ and $\pi_{\mathrm{right}}^t:V'\to V$ as $$\pi_{\mathrm{right}}^t(u') = \argmax(m^t_{\cdot, u'})$$ which are candidates for matching vertices from $G$ to $H$ or from $H$ to $G$.
If $t$ is small, then the approximation $m^t_{u,u'}$ will not be accurate as it does not incorporate sufficient information (only at depth $t$ in both graphs). When $t$ is large, cycles will appear in both graphs so that the recursion is not anymore valid. In order to choose an appropriate number of iterations $t$, we adopt the following simple strategy: we compute all the matrices $m^t_{u,u'}$ for all values of $t$ less than a parameter $d$; then from these matrices, we compute the corresponding mappings $\pi_{\mathrm{left}}^t$ and $\pi_{\mathrm{right}}^t$ as described above; we then compute:
\begin{eqnarray}
\nonumber
e(t) &:=& \text{match-edges}(G,H,\pi_{\mathrm{left}}^t,\pi_{\mathrm{right}}^t)\\
\label{eq:edges} &:=& \frac{1}{|E|}\sum_{\set{u,v}\in E} \one_{(\pi_{\mathrm{left}}^t(u),\pi_{\mathrm{left}}^t(v))\in E'} + \frac{1}{|E'|}\sum_{\set{u',v'}\in E'} \one_{(\pi_{\mathrm{right}}^t(u'),\pi_{\mathrm{right}}^t(v'))\in E} \, .
\end{eqnarray}
Finally, we choose $$t^* =\argmax(e(t)) \, .$$
Note that we are considering sparse \ER graphs which are typically not connected (the diameter is infinite). We know from \cite{ganassali2021impossibility}, that only the giant connected component of $G$ and $H$ can possibly be aligned. Hence as a first pre-processing step, we remove all the small connected components from $G$ and $H$ and keep only the largest one. As a result, our algorithm takes as input two connected graphs of possibly different sizes. The pseudo-code for our algorithm is given below:

\begin{algorithm}[h]
\caption{\label{algo_MPAlign2} \texttt{MPAlign2}}
\SetAlgoLined

\textbf{Input:} Two connected graphs $G=(V,E)$ and $H=(V',E')$, parameter $d$ and parameters of the correlated \ER model $\lambda$ (average degree) and $s$

\For{$t \in \{1,\dots, d\}$}{
	compute $m^t_{u \to v, u'\to v'}$ thanks to \eqref{eq:message}
	
	compute $m^t_{u,u'}$ thanks to \eqref{eq:aggregation}

	compute $\pi_{\mathrm{left}}^t:V\to V'$ as $\pi_{\mathrm{left}}^t(u) = \argmax(m^t_{u,\cdot})$
	
	compute $\pi_{\mathrm{right}}^t:V'\to V$ as $\pi_{\mathrm{right}}^t(u') = \argmax(m^t_{\cdot, u'})$
	
	compute $e(t) = \text{match-edges}(G,H,\pi_{\mathrm{left}}^t,\pi_{\mathrm{right}}^t)$ thanks to \eqref{eq:edges}
	
	}
$t^* = \argmax(e(t))$

\textbf{Return} $\pi_{\mathrm{left}}^{t^*}$, $\pi_{\mathrm{right}}^{t^*}$, $m^{t^*}$

\end{algorithm}

Figure \ref{fig:overlap} shows some empirical results for graphs of size $200$ for values $\lambda = 2;2,5;3$ where the overall overlap is defined by
\begin{equation}\label{eq:def:mean_overlap_tstar}
    \mathrm{overlap} := \frac{1}{2}\left( \ov(\pi_{\mathrm{left}}^{t^*},\pi^*) + \ov(\pi_{\mathrm{right}}^{t^*},(\pi^*)^{-1}) \right) \, ,
\end{equation}
namely the mean of the overlaps given by $\pi_{\mathrm{left}}^{t^*}$ and $\pi_{\mathrm{right}}^{t^*}$. The maximum number of iterations is fixed to $d=15$. For more numerical experiments on this algorithm, see \cite{piccioli2021aligning}.

\begin{figure}[h]

\includegraphics[width=14cm]{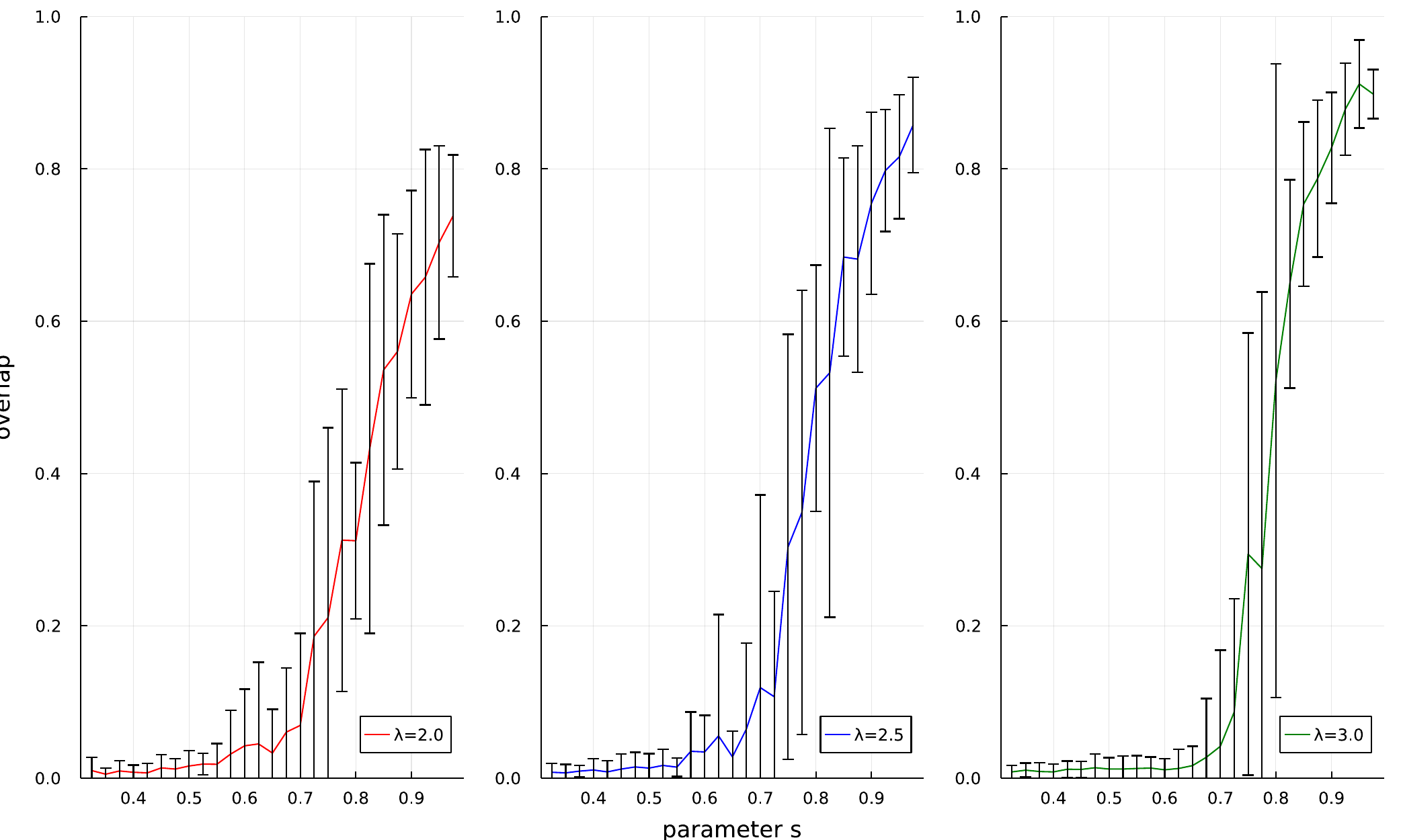}
\centering
\caption{\label{fig:overlap} Overlap as a function of the parameter $s$ for graphs with (initial) size $n=200$ for various values of $\lambda$ (parameter $d=15$). Each point is the average of $10$ simulations. 
}
\end{figure}

 This choice of $d=15$ is validated by the results presented in Figure \ref{fig:iter}. We plot for each simulation the time-dependent $\mathrm{overlap}(t)$ defined by 
 \begin{equation}\label{eq:def:mean_overlap_t}
    \mathrm{overlap}(t) := \frac{1}{2}\left( \ov(\pi_{\mathrm{left}}^{t},\pi^*) + \ov(\pi_{\mathrm{right}}^{t},(\pi^*)^{-1}) \right)
\end{equation} as a function of $t\leq 15$. We see that for low values of $s$ (on the left $s=0.4$), the overlap behaves randomly. In this scenario, increasing the value of $d$ will probably not help as cycles will deteriorate the performance of the algorithm. For high value of $s$ (on the right $s=0.9$), we see that the overlap starts by increasing and then decreases abruptly to zero, this is due to numerical issues: some messages in $m^t$ are too large for our implementation of the algorithm to be able to deal with them. Finally for values of $s$, where signal is detected (in the middle $s=0.675$), we see that when the signal is detected, the overlap start by increasing until reaching a maximum and then decreases before numerical instability. We also note that our choice of $t^*$ thanks to the number of matched edges can be fairly sub-optimal. We believe that a better understanding of the performance of our algorithm for finite $n$ is an interesting open problem. We refer to \cite{piccioli2021aligning} which provides more detailed experimental results on a similar algorithm.

\begin{figure}[h]

\includegraphics[width=14cm]{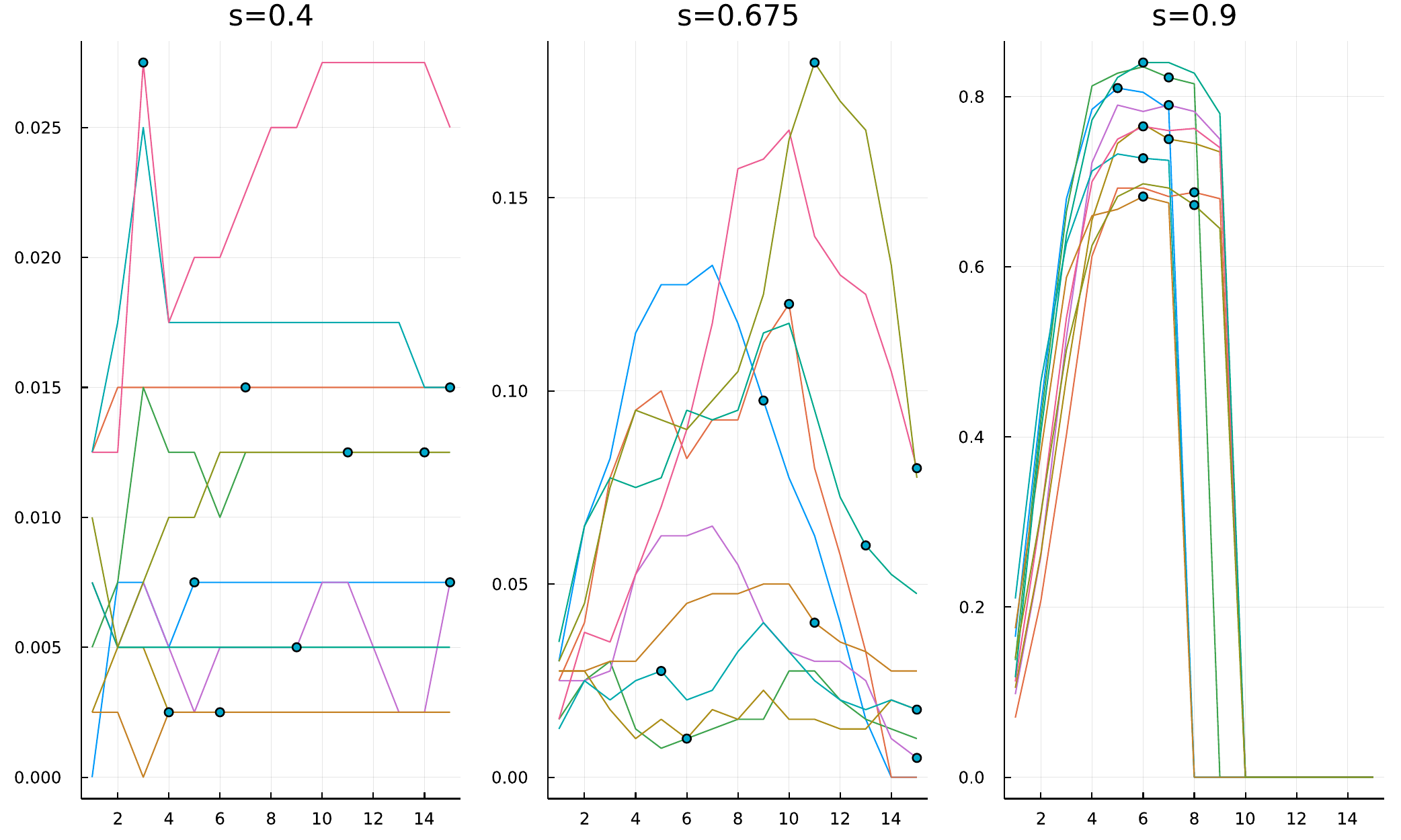}
\centering
\caption{\label{fig:iter} Overlap \eqref{eq:def:mean_overlap_t} as a function the number of iterations $t$ for graphs with (initial) size $n=200$ for $\lambda=2.5$ (parameter $d=15$) and various values of $s$. The dotted point on each curve corresponds to $t^*$. Note that the $y$-axis of each plot have different scale. When overlap reaches zero, our algorithm hits infinity.
}
\end{figure}

\section{Additional proofs}\label{appendix:additional_proofs}

\subsection{Proof of Proposition \ref{proposition:auto_GW}}
\label{appendix:proof:proposition:auto_GW}
\begin{proof}
Throughout the proof, let $X_\mu$ denote a Poisson random variable with parameter $\mu$. 
A node $u\in\cL_{d-2}(\tau^*)$ has, independently for each $k\in \dN$, a number $N_k \sim \Poi(r\pi_r(k))$ children who themselves have $k$ children. To each such node, we can associate 
\begin{equation*}
    \prod_{k\in \dN}N_k!
\end{equation*} permutations of its children that will preserve the labeled tree. Likewise, for each node $u\in \cL_{d-1}(\tau^*)$, there are $c_u!$ permutations of its children that don't modify the tree, where $c_u := c_{\tau^*}(u)$. Thus by the strong law of large numbers, we have:
\begin{equation}\label{eq:equivalent_logaut1}
\log \card{\Aut(\tau^*)} \geq (1+o_{\dP}(1))\left[W r^{d-1}\dE\left[\log(X_r!)\right]+ W r^{d-2}\sum_{k\in \dN}\dE\left[\log(X_{r\pi_r(k)}!)\right]\right].
\end{equation}
Recall the classical estimate for large $\mu$:
\begin{equation}\label{eq:moment_factoriel2}
\dE \log(X_\mu!)=\mu\log(\mu)-\mu+\frac{1}{2}\log(2\pi e \mu)+O\left(\frac{1}{\mu}\right),
\end{equation}
and Stirling's formula gives
\begin{equation}\label{eq:stirling}
    \log(k!)= k \log k - k + \frac{1}{2}\log(2 \pi k) + O\left( \frac{1}{k} \right).
\end{equation}
We now give some estimates of the distribution $\pi_r(k)$ in the following Lemma, which proof is deferred to Appendix \ref{appendix:proof_lemma_equivalent_pi}.
\begin{lemma}\label{lemma_equivalent_pi}
Let $\eps(r)$ be such that $\eps(r) \to 0$ and $\eps(r)\log r \to +\infty$ when $r \to +\infty$. Let 
\begin{equation*}
    I_{r,\eps}:=\left[r-(1-\eps(r))\sqrt{r\log r},r+(1-\eps(r))\sqrt{r\log r}\right].
\end{equation*} Then 
\begin{itemize}
    \item[$(i)$] we have
    \begin{equation}\label{eq:lemma_concentration_poisson}
        \dP\left( X_r \notin I_{r,\eps} \right) = O\left(r^{-1/2} e^{\eps(r)\log r}\right).
    \end{equation}
    \item[$(ii)$] for all $k \in I_{r,\eps}$, letting $x_k=\frac{k-r}{\sqrt{r}}$, we have
\begin{equation}\label{eq:lemma_equivalent_pi}
\pi_r(k)=\frac{1}{\sqrt{2\pi r}}e^{-{x_k^2}/{2}}\left[1+\frac{x_k^3}{6\sqrt{r}}-\frac{x_k}{2\sqrt{r}}+O\left(\frac{x_k^6}{r}\right)\right].
\end{equation}
  \item[$(iii)$] Note that \eqref{eq:lemma_equivalent_pi} implies that for each $k\in I_{r,\eps}$, it holds that $r\pi_r(k) = \Omega\left(e^{\eps(r) \log r (1-o(1))}\right) $, thus diverges to $+\infty$.
\end{itemize}
\end{lemma}

Consider the function $\eps(r) := \frac{\log\log r}{4 \log r}$, which satisfies the assumptions of Lemma \ref{lemma_equivalent_pi}. Using expansion \eqref{eq:lemma_equivalent_pi} together with \eqref{eq:moment_factoriel2} gives:
\begin{flalign}\label{eq:logiso_2}
\sum_{k \in I_{r,\eps}} \dE & \left[\log(X_{r\pi_r(k)}!)\right] = \sum_{k\in I_{r,\eps}} r\pi_r(k)\log(r \pi_r(k))-r\pi_r(k)+\frac{1}{2}\log(2\pi e r\pi_r(k))+O\left(\frac{1}{r\pi_r(k)}\right) \nonumber\\
& =\sum_{k\in I_{r,\eps}}r\pi_r(k)\left[\frac{1}{2}\log(r)-\frac{1}{2}\log(2\pi)-\frac{x_k^2}{2}+\frac{x_k^3}{6\sqrt{r}}-\frac{x_k}{2\sqrt{r}}+O\left(\frac{\log^2 r}{r}\right)-1\right] \nonumber\\
 & \quad \quad +\sum_{k\in I_{r,\eps}}\frac{1}{2} \left[\log( 2\pi e)+\frac{1}{2}\log(r)-\frac{1}{2}\log(2\pi)-\frac{x_k^2}{2}+O\left(\frac{\log^{3/2}(r)}{\sqrt{r}}\right)\right]+ O\left(\sqrt{r \log r} \right) \nonumber\\
 & \overset{(a)}{=} \frac{1}{2}r\log(r)-\left(\frac{1}{2}\log(2\pi)+\frac{1}{2}+1\right)r +O(\sqrt{r}\log^{5/4} r ) \nonumber\\
 & \quad \quad + O(\sqrt{r\log r}) + \frac{1}{2}\left( 1-\eps(r)\right)\sqrt{r}\log^{3/2}(r)-\frac{1}{4}\sum_{k\in I_{r,\eps}} x_k^2 \nonumber\\
 & \overset{(b)}{=} \frac{1}{2}r\log(r)-\left(\frac{1}{2}\log(2\pi)+\frac{3}{2}\right)r + \frac{1}{3} \sqrt{r}\log^{3/2}(r) +O(\sqrt{r}\log^{5/4} r).
\end{flalign}
Let us give hereafter all the required details for the above computation.
\begin{itemize}
    \item At step $(a)$, we first used point $(i)$ of Lemma \ref{lemma_equivalent_pi}, which gives that $$r \log r \times \dP\left( X_r \notin I_{r,\eps} \right) = O\left(\sqrt{r} \log^{1/4} r\right) = O\left(\sqrt{r} \log^{5/4} r\right).$$ For the sum of the $x_k^2$, we remark that
    \begin{flalign*}
        \sum_{k\in I_{r,\eps}} r \pi_r(k) \frac{x_k^2}{2} &= \frac{r}{2} \left(1- \dE\left[\left(\frac{X_r-r}{\sqrt{r}}\right)^2 \one_{X_r \notin I_{r,\eps}}\right]\right),
    \end{flalign*} and that the expectation in the right-hand term can be written as follows
    \begin{flalign*}
    \dE&  \left[\left(\frac{X_r-r}{\sqrt{r}}\right)^2 \one_{\left|\frac{X_r-r}{\sqrt{r}}\right| \geq 2 \sqrt{\log r}}\right] + \dE\left[\left(\frac{X_r-r}{\sqrt{r}}\right)^2 \one_{ (1-\eps(r)) \sqrt{\log r} \leq \left|\frac{X_r-r}{\sqrt{r}}\right| \leq 2 \sqrt{\log r}}\right]\\
    & \leq \dE\left[\left(\frac{X_r-r}{\sqrt{r}}\right)^4 \right]^{1/2} \dP\left( \left|\frac{X_r-r}{\sqrt{r}}\right| \geq 2 \sqrt{\log r} \right)^{1/2} + 4\log r \times  \dP\left( X_r \notin I_{r,\eps} \right)\\
    & \leq O\left( r^{-1/2} \right) + O\left( r^{-1/2}  \log^{5/4} r \right).
    \end{flalign*} Hence, $\sum_{k\in I_{r,\eps}} r \pi_r(k) \frac{x_k^2}{2} = \frac{r}{2} - O\left(\sqrt{r} \log^{5/4} r \right)$. Finally, using the fact that $\dE\left[\left(\frac{X_r-r}{\sqrt{r}}\right)^3 \right]$ and $\dE\left[\frac{X_r-r}{\sqrt{r}} \right]$ are $O(1)$, the sums of the $x_k^3$ and $x_k$ easily incorporate into the $O\left(\sqrt{r} \log^{5/4} r \right)$ term.
    
    \item At step $(b)$, we first used the fact that $\eps(r)\sqrt{r}\log^{3/2} = O\left(\sqrt{r} \log^{5/4} r \right)$. The only term needing more computations is
    \begin{flalign*}
        \sum_{k\in I_{r,\eps}} x_k^2 & = \sum_{k\in I_{r,\eps}} \left(\frac{k-r}{\sqrt{r}}\right)^2 = 2 \cdot  \sum_{\ell = 0}^{(1-\eps(r))\sqrt{r\log r}} \frac{\ell^2}{r} = \frac{2}{3} \sqrt{r} \log^{3/2} r + O\left(\sqrt{r} \log^{5/4} r \right).
    \end{flalign*}
\end{itemize}

Copying \eqref{eq:logiso_2} together with \eqref{eq:moment_factoriel2} in \eqref{eq:equivalent_logaut1} yields:
\begin{flalign*}
\log(\card{\Aut(\tau^*)}) & \geq (1+o_{\dP}(1)) W r^{d-1} \left[r\log(r)-r+\frac{1}{2}\log(2\pi e r)+ O\left( \frac{1}{r}\right) \right] \\
& \quad \quad + (1+o_{\dP}(1)) W r^{d-1} \left[\frac{1}{2}\log(r)-\frac{1}{2}\log(2\pi)-\frac{3}{2}+\frac{\log^{3/2} r}{3\sqrt{r}}+O\left(\frac{\log^{5/4} r}{\sqrt{r}}\right)\right]\\
& = (1+o_{\dP}(1)) W r^{d-1}\left[r\log(r)-r+\log(r)-1+\frac{\log^{3/2}(r)}{3\sqrt{r}}+O\left(\frac{\log^{5/4} r}{\sqrt{r}}\right)\right].
\end{flalign*}
Another appeal to the strong law of large numbers entails that 

\begin{flalign*}
    \log\left(\prod_{u \in \cV_{d-1}(\tau^*)}e^{-r} r^{c_u} \right) & = (1+o_{\dP}(1))\left|\cV_{d-1}(\tau^*) \right| \dE\left[-r + c_{\rho(\tau^*)} \log r \right]\\
    & = (1+o_{\dP}(1))\frac{Wr^{d}}{r-1}\left( -r+r\log(r)\right).
\end{flalign*}

Combined, these last two evaluations yield a lower bound of $\log\left(\frac{\card{\Aut(\tau^*)}}{\prod_{u \in \cV_{d-1}(\tau^*)}e^{-r} r^{c_{u}}}\right) $ under the event on which $\tau^*$ survives, of the form
\begin{flalign*}
& (1-o_{\dP}(1))\frac{Wr^{d}}{r-1} \left[-r\log(r)+r+\left(1-\frac{1}{r}\right)\left(r\log(r)-r+\log(r)-1+\frac{\log^{3/2}(r)}{3\sqrt{r}}+O\left(\frac{\log^{5/4} r}{\sqrt{r}}\right)\right)\right]\\
& = (1-o_{\dP}(1)) \frac{Wr^{d}}{r-1}\left[ \frac{\log^{3/2}(r)}{3\sqrt{r}}+O\left(\frac{\log^{5/4} r}{\sqrt{r}}\right) \right].
\end{flalign*}
\end{proof} 

\subsection{Proof of Lemma \ref{lemma_equivalent_pi}}
\label{appendix:proof_lemma_equivalent_pi}
\begin{proof}
\emph{$(i)$} The result follows directly from the classical Poisson concentration inequality 
\begin{equation*}
    \dP\left( \left|X_r - r \right| \geq x \right) \leq 2 \exp\left( - \frac{x^2}{2(r+x)}\right),
\end{equation*} noting that for $x=(1-\eps(r))\sqrt{r\log r}$,
$\frac{x^2}{2(r+x)} \geq \frac{1}{2} \log r - \eps \log r - o(1).$

\emph{$(ii)$}
When $k$ runs over $I_{r,\eps}$, $x_k$ runs over $\left[ -(1-\eps(r)) \sqrt{\log r}, (1-\eps(r)) \sqrt{\log r}\right]$. Using Stirling's formula \eqref{eq:stirling}, we get
\begin{flalign*}
\log\pi_r(k) & = \log\pi_r(r+x_k \sqrt{r}) = -r + k \log r - \log(k!) \\
& = -r + (r+x_k \sqrt{r}) \log r - (r+x_k \sqrt{r}) \log(r+x_k \sqrt{r}) + r+x_k \sqrt{r} - \frac{1}{2}\log(2 \pi (r+x_k \sqrt{r})) + O\left( \frac{1}{r} \right)\\
& = -r + r\log r + x_k \sqrt{r} \log r - (r+x_k \sqrt{r}) \left[ \log r + \frac{x_k}{r^{1/2}} - \frac{x_k^2}{2r} + \frac{x_k^3}{3r^{3/2}} + O\left( \frac{x_k^4}{ r^{2}}\right)  \right] \\
& \quad \quad \quad \quad + r + x_k \sqrt{r} - \frac{1}{2}\log(2 \pi) - \frac{1}{2}\log(r) - \frac{1}{2} \frac{x_k}{r^{1/2}} + O\left( \frac{x_k^2}{ r}\right)\\
& = -r - x_k \sqrt{r} -x_k^2 + \frac{x_k^2}{2} + \frac{x_k^3}{2 \sqrt{r}} - \frac{x_k^3}{3 \sqrt{r}} + O\left( \frac{x_k^4}{ r}\right) \\
& \quad \quad \quad \quad + r + x_k \sqrt{r} - \frac{1}{2}\log(2 \pi r) - \frac{1}{2} \frac{x_k}{r^{1/2}} + O\left( \frac{x_k^2}{ r}\right)\\
& = - \frac{x_k^2}{2} - \frac{1}{2}\log(2 \pi r) + \frac{x_k^3}{6 \sqrt{r}} - \frac{x_k}{2 \sqrt{r}} + O\left( \frac{x_k^4}{ r}\right).
\end{flalign*} Taking the exponential gives
\begin{flalign*}
\pi_r(k) & = \frac{1}{\sqrt{2\pi r}} e^{-{x_k^2}/{2}} \left[ 1+\frac{x_k^3}{6\sqrt{r}}-\frac{x_k}{2\sqrt{r}}+O\left(\frac{x_k^6}{r}\right)\right].
\end{flalign*}

$(iii)$ follows directly from $(ii)$.
\end{proof}

\subsection{Proof of Lemma \ref{lemma:control_M2}}\label{appendix:proof_lemma:control_M2}
\begin{proof}
We condition on $T$ be the number of recursive steps in the previous construction, which is $O((\log n)n^{1/4})$ under $\cA$. For each $s \in [T]$, we denote by $c_s$ the number of newly sampled children, and $v_s := \sum_{s'=0}^{s-1} c_{s'}$ the number of uncovered vertices before step $s$ (we set $v_0 : =0$). With these notations, it is easily seen than $N_{d-1}$ can be factorized as follows:
\begin{flalign*}
N_{d-1} & = \prod_{s \in [T]} \frac{\dP\left(\Bin\left(n-2-v_s , \lambda/n \right) = c_s \right)}{\p_{\lambda}(c_s)}  \leq \prod_{s \in [T]}\exp\left(  \frac{\lambda}{n}(v_s+2+c_s)\right) \\
& = \exp\left(\frac{2\lambda T}{n} + \frac{\lambda}{n} \sum_{s \in [T]} (T-s)c_s \right).
\end{flalign*}
Under $\dPl_{d-1}$, the variables $c_s$ are independent $\Poi(\lambda)$ variables, hence
\begin{flalign*}
 \dEl_{d-1}\left[ N_{d-1}^2 \one_{\cA} \right] & \leq \exp\left(\frac{4\lambda T}{n} + \lambda \sum_{s \in [T]} \left( e^{{2 \lambda}(T-s)/{n}}-1\right) \right) \one_{T = O((\log n)n^{1/4})}\\
 & \leq \exp\left( C' {T^2}/{n} + o({T^2}/{n})\right)\one_{T = O((\log n)n^{1/4})} = O(1).
\end{flalign*}
\end{proof}

\end{document}